\documentclass[a4paper,UKenglish]{lipics-neutered}
 
\usepackage{microtype}

\title{Expressibility in the Lambda Calculus with $\mu$}
  
\titlerunning{Expressibility in the Lambda Calculus with $\mu$}

\author[1]{Clemens Grabmayer}
\author[2]{Jan Rochel
                     }
\affil[1]{Department of Philosophy, Utrecht University\\
  PO Box 80126, 3508 TC Utrecht, The Netherlands\\
  \texttt{clemens@phil.uu.nl}}
\affil[2]{Department of Information and Computing Sciences\\
  PO Box 80089, 3508 TB Utrecht\\
  \texttt{jan@rochel.info}}
\authorrunning{C.\ Grabmayer and J.\ Rochel} 

\Copyright[by]{C.\ Grabmayer and J.\ Rochel}


\serieslogo{}
\volumeinfo
  {Editor A., Editors B.}
  {2}
  {Conference title on which this volume is based on}
  {1}
  {1}
  {1}
\EventShortName{}
\DOI{10.4230/LIPIcs.xxx.yyy.p}

\usepackage{bussproofs}

\usepackage{amsmath}
\usepackage{amssymb}
\usepackage{amsfonts}
\usepackage{graphicx}
\usepackage{enumerate}
\usepackage{color}
\usepackage[utf8]{inputenc}

\usepackage{MnSymbol}

\usepackage{microtype}

\usepackage{tikz}
\usetikzlibrary{matrix,positioning}

\theoremstyle{definition}
\newtheorem{proposition}[theorem]{Proposition}

\newcommand{\bindcaptchain}{bin\-ding--cap\-tu\-ring chain}
\newcommand{\bindcaptchains}{\bindcaptchain{s}}
\newcommand{\Bindcaptchain}{Bin\-ding--cap\-tu\-ring chain}
\newcommand{\Bindcaptchains}{\Bindcaptchain{s}}

\newcommand\fig[1]{\includegraphics[scale=0.81]{figs/{{#1}}}}

\newcommand\vcentered[1]{\raisebox{-0.5\height}{#1}}

\newcommand\lambdacal\blambda
\newcommand\blambda{\ensuremath{\bs{\lambda}}} 

\newcommand{\inflambdacal}{\ensuremath{\bs{\lambda}^{\hspace*{-1pt}\bs{\infty}}}}
\newcommand{\lambdaletreccal}{\ensuremath{\indap{\bs{\lambda}}{\stxtletrec}}}

\newcommand{\lambdamucal}{\ensuremath{\indap{\bs{\lambda}}{{\mu}}}}
\newcommand{\lambdaprefixcal}{\ensuremath{\bs{(\lambda)\lambda}}}
\newcommand{\lambdamuprefixcal}{\ensuremath{\bs{(\lambda)}\lambdamucal}}

\newcommand{\inflambdaprefixcal}{\ensuremath{\bs{(\lambda)}\inflambdacal}}
\newcommand{\inflambdaprefixposcal}{\ensuremath{\bs{(\lambda)}}_{\sbfpos}\bs{\lambda}^{\hspace*{-0.25pt}\bs{\infty}}}

\newcommand{\lambdacalculus}{$\lambda$\nb-calculus}

\let\oldlambda\lambda
\renewcommand\lambda{\ensuremath\oldlambda}

\newcommand{\lambdamu}{\ensuremath{\sslabs_{\ssmuabs}}}

\newcommand{\lambdaterm}{$\lambda$\nb-term}
\newcommand{\lambdaterms}{\lambdaterm{s}}
\newcommand{\lambdaabstraction}{$\lambda$\nb-ab\-strac\-tion}
\newcommand{\lambdaabstractions}{\lambdaabstraction{s}}
\newcommand{\lambdabinding}{$\lambda$\nb-bin\-ding}
\newcommand{\lambdabindings}{\lambdabinding{s}}

\newcommand{\lambdaexpressions}{\lambdaexpressions}

\newcommand{\lambdaletreccalexpressible}{$\lambdaletreccal$\nb-ex\-press\-ible}
\newcommand{\lambdaletreccalexpressibility}{$\lambdaletreccal$\nb-ex\-press\-ibi\-lity}
\newcommand{\lambdaletrecexpressible}{$\lambda_{\stxtletrec}$\nb-ex\-press\-ible}
\newcommand{\lambdaletrecexpressibility}{$\stxtlambdaletreccal$\nb-ex\-press\-ib\-ility}
\newcommand{\lambdamuexpressible}{$\stxtlambdamucal$\nb-ex\-press\-ible}
\newcommand{\lambdamuexpressibility}{$\stxtlambdamucal$\nb-ex\-press\-ib\-ility}

\newcommand{\lambdamucalexpressible}{$\lambdamucal$\nb-ex\-press\-ible}
\newcommand{\lambdamucalexpressibility}{$\lambdamucal$\nb-ex\-press\-ibi\-lity}
\newcommand{\lambdamuterm}{$\stxtlambdamucal$\protect\nb-term}
\newcommand{\lambdamuterms}{\lambdamuterm{s}}

\newcommand{\posannotated}{po\-si\-tion-anno\-ta\-ted}

\newcommand{\CRSabstraction}{\CRS\nb-ab\-strac\-tion}

\newcommand{\murecursion}{$\smu$\nb-re\-cur\-sion}
\newcommand{\mubinding}{$\mu$\nb-bin\-ding}
\newcommand{\mubindings}{\mubinding{s}}

\newcommand{\letrecbinding}{$\stxtletrec$\nb-bin\-ding}
\newcommand{\letrecbindings}{\letrecbinding{s}}

\newcommand{\sslabs}{\lambda}
\newcommand{\slabs}[1]{\sslabs{#1}}
\newcommand{\labs}[2]{\slabs{#1}.\hspace*{0.5pt}{#2}}
\newcommand{\sslapp}{@}
\newcommand{\lapp}[2]{{#1}\,{#2}}

\newcommand{\ssmuabs}{\mu}
\newcommand{\smuabs}[1]{\ssmuabs{#1}}
\newcommand{\muabs}[2]{\smuabs{#1}.\hspace*{0.5pt}{#2}}

\newcommand{\snlvar}{\ensuremath{\mathsf{0}}}
\newcommand{\nlvar}{\snlvar}
\newcommand{\snlvarsucc}{\ensuremath{\mathsf{S}}}

\newcommand{\smin}{\ensuremath{\mathbf{in}}}

\newcommand{\sletrec}{\ensuremath{\mathbf{letrec}}}

\newcommand{\letrecin}[2]{\sletrec\;{#1}\;\smin\;{#2}}

\newcommand{\stxtletrec}{\ensuremath{\text{\normalfont\sf letrec}}}

\newcommand{\stxtlambdaletreccal}{\indap{\lambda}{\stxtletrec}}
\newcommand{\stxtlambdamucal}{\indap{\lambda}{\mu}}

\newcommand{\smuCRS}{\ensuremath{\textsf{mu}}}

\newcommand{\allter}{L}

\newcommand{\almter}{M}
\newcommand{\blmter}{N}

\newcommand{\almteri}{\indap{\almter}}

\newcommand{\arecvar}{f}

\makeatletter
\def\a#1{\reflectbox{$\m@th#1{\lambda}$}}

\makeatother

\newcommand{\sflabs}[1]{(\slabs{#1})}
\newcommand{\flabs}[2]{\sflabs{#1}\hspace*{1pt}{#2}}
\newcommand{\femptylabs}[1]{()\hspace*{0.5pt}{#1}}

\newcommand{\flabspos}[4]{\bpap{(\slabs{#1})}{#2}{#3}\hspace*{0.5pt}{#4}}
\newcommand{\femptylabspos}[3]{\bpap{()}{#1}{#2}\hspace*{0.5pt}{#3}}

\newcommand{\sflabsposCRS}[2]{{\mathsf{pre}}^{#2}_{\enumsequence{#1}}}
\newcommand{\flabsposCRS}[4]{\funap{\sflabsposCRS{#2}{#3}}{\absCRS{#1}{#4}}}

\newcommand{\avar}{x}
\newcommand{\bvar}{y}
\newcommand{\cvar}{z}
\newcommand{\dvar}{u}
\newcommand{\evar}{v}

\newcommand{\avari}[1]{\avar_{#1}}
\newcommand{\bvari}[1]{\bvar_{#1}}
\newcommand{\cvari}[1]{\cvar_{#1}}

\newcommand{\ater}{T}
\newcommand{\bter}{U}
\newcommand{\cter}{V}

\newcommand{\ateri}[1]{\ater_{#1}}

\newcommand{\aiter}{T}
\newcommand{\biter}{U}
\newcommand{\citer}{V}
\newcommand{\diter}{W}
\newcommand{\aiteri}[1]{\aiter_{#1}}
\newcommand{\biteri}[1]{\biter_{#1}}
\newcommand{\citeri}[1]{\citer_{#1}}

\newcommand{\biteracc}{U'}

\newcommand{\diteracc}{W'}

\newcommand{\biteracci}[1]{\biteracc_{#1}}

\newcommand{\sann}{{:}}
\newcommand{\annexp}[2]{{{#1}\hspace*{3pt}\sann\hspace*{3pt}{#2}}}
\newcommand{\annexptxt}[2]{{{#1}\hspace*{2pt}\sann\hspace*{2pt}{#2}}}

\newcommand{\flabsann}[3]{(\slabs{#1})\hspace*{3pt}\annexp{#2}{#3}}
\newcommand{\femptylabsann}[2]{()\hspace*{2pt}\annexp{#1}{#2}}
\newcommand{\flabsanntxt}[3]{(\slabs{#1})\hspace*{3pt}\annexptxt{#2}{#3}}

\newcommand{\sTer}{\mit{Ter}}
\newcommand{\Ter}{\funap{\sTer}}
\newcommand{\siTer}{{\sTer^{\infty}}}
\newcommand{\iTer}{\funap{\siTer}}

\newcommand{\asig}{\Sigma}

\newcommand{\sametavar}{X}

\newcommand{\scmetavar}{Z}

\newcommand{\cmetavar}{\funap{\scmetavar}}

\newcommand{\aconstname}{\mathsf{c}}

\newcommand{\aconstnamei}{\subap{\aconstname}}

\newcommand{\sabsCRS}[1]{[{#1}]}
\newcommand{\absCRS}[2]{[{#1}]\hspace*{1pt}{#2}}

\newcommand{\CRS}{\ensuremath{\text{CRS}}}
\newcommand{\CRSs}{\ensuremath{\text{CRSs}}}
\newcommand{\iCRS}{\ensuremath{\text{iCRS}}}
\newcommand{\iCRSs}{\ensuremath{\text{iCRSs}}}

\newcommand{\slabsCRS}{\mathsf{abs}}
\newcommand{\labsCRS}[2]{\funap{\slabsCRS}{\absCRS{#1}{#2}}}
\newcommand{\slappCRS}{\mathsf{app}}
\newcommand{\lappCRS}[2]{\funap{\slappCRS}{#1,#2}}

\newcommand{\smuabsCRS}{\mathsf{mu}}
\newcommand{\muabsCRS}[2]{\funap{\smuabsCRS}{\absCRS{#1}{#2}}}

\newcommand{\sflabsCRS}[1]{\mathsf{pre}_{#1}}
\newcommand{\flabsCRS}[3]{\funap{\sflabsCRS{#1}}{\absCRS{#2}{#3}}}

\newcommand\sRegCRS{\mbox{\textit{\text{Reg}}}}

\newcommand{\RegCRS}{\sRegCRS}
\newcommand{\stRegCRS}{\sRegCRS^{\bs{+}}}

\newcommand{\decompARS}{\ensuremath{{\mathit{(\Lambda)}}}}

\newcommand{\stRegposCRS}{\sRegCRS_{\textit{\textbf{pos}}}^{\bs{+}}}

\newcommand{\RegARS}{\textit{Reg}}
\newcommand{\stRegARS}{\textit{Reg}^{+}}

\newcommand{\stRegposARS}{\textit{Reg}_{\textit{pos}}^+}

\newcommand{\siglcCRS}{\asig_{\lambda}}
\newcommand{\siglmcCRS}{\asig_{\lambda_{\smu}}}

\newcommand{\siglpcCRS}{\asig_{(\lambda)\lambda}}
\newcommand{\siglpposcCRS}{\asig_{(\lambda)_{\spos}\lambda}}

\newcommand{\ARS}{ARS}

\newcommand{\aARS}{{\cal A}}

\newcommand{\aocc}{o}

\newcommand{\apath}{\pi}
\newcommand{\athread}{\theta}

\newcommand{\sgST}{{\mit ST}}

\newcommand{\sgSTreg}{{\mit ST}}
\newcommand\gSTreg[1]{\funap\sgSTreg{#1}}

\newcommand{\sgSTstreg}{\sgST^+}
\newcommand\gSTstreg[1]{\funap\sgSTstreg{#1}}

\newcommand{\arewseq}{\tau}

\newcommand{\crewseq}{\pi}

\newcommand{\Checkreg}{\Check}

\newcommand{\sred}{{\to}}
\newcommand{\red}{\mathrel{\to}}

\newcommand{\smred}{{\twoheadrightarrow}}

\newcommand\sconvred{{\leftarrow}}

\newcommand{\sconvmred}{{\twoheadleftarrow}}

\newcommand{\snfred}{\to^{\scriptstyle !}} 

\newcommand{\sregred}{\subap{\sred}{\reg}}
\newcommand{\sstregred}{\subap{\sred}{\streg}}

\newcommand{\sregmred}{\subap{\smred}{\reg}}
\newcommand{\sstregmred}{\subap{\smred}{\streg}}

\newcommand{\sstregconvmred}{\subap{\sconvmred}{\streg}}

\newcommand{\sregeqred}{\bpap{\sred}{\reg}{=}}
\newcommand{\sstregeqred}{\bpap{\sred}{\streg}{=}}

\newcommand{\regred}{\mathrel{\sregred}}  
\newcommand{\stregred}{\mathrel{\sstregred}}

\newcommand{\stregconvmred}{\mathrel{\sstregconvmred}}

\newcommand{\regmred}{\mathrel{\sregmred}}  
\newcommand{\stregmred}{\mathrel{\sstregmred}}

\newcommand{\reg}{\ensuremath{\text{\normalfont reg}}}
\newcommand{\streg}{\ensuremath{\text{\normalfont reg}^+}}

\newcommand{\smu}{{\mu}}

\newcommand{\sunfold}{\mu}
\newcommand{\sunfoldred}{\indap{\sred}{\hspace*{-1pt}\sunfold}}

\newcommand{\sunfoldinfred}{\indap{\sinfred}{\hspace*{-1pt}\sunfold}}
\newcommand{\unfoldinfred}{\mathrel{\sunfoldinfred}}

\newcommand{\sunfoldominfnfred}{\mbox{$\indap{\sominfnfred}{\sunfold}$}}
\newcommand{\unfoldominfnfred}{\mathrel{\sunfoldominfnfred}}

\newcommand{\sunfoldinfnfred}{\indap{\sinfnfred}{\hspace*{-1pt}\sunfold}}

\newcommand{\unfemptyred}{\mathrel{\unfemptyred}}

\newcommand{\unffreered}{\mathrel{\unffreered}}

\newcommand{\scompressregnfred}{{\snfred_{\scompressreg}}}

\newcommand{\scompressstregnfred}{{\snfred_{\scompressstreg}}}

\newcommand\thsp{-1.785ex}
\newcommand\threeheadrightarrow{\twoheadrightarrow\hspace*\thsp\twoheadrightarrow}

\newcommand{\sinfred}{\threeheadrightarrow}

\newcommand{\sinfnfred}{\sinfred^{!}}

\newcommand{\sominfred}{\overset{\text{\normalfont out}\hspace*{3.5pt}}{\sinfred}}
\newcommand{\sominfnfred}{\sominfred{\hspace*{-4pt}}^{\scriptstyle !}}

\newcommand{\astrat}{\mathbb{S}}

\newcommand{\slabsdecomp}{\sslabs}
\newcommand{\slabsdecompred}{{\subap{\sred}{\slabsdecomp}}}

\newcommand{\slappdecompi}{\subap{@}}
\newcommand{\slappdecompired}[1]{{\subap{\sred}{\slappdecompi{#1}}}}
\newcommand{\lappdecompired}[1]{\mathrel{\slappdecompired{#1}}}

\newcommand{\scompress}{\text{\normalfont del}}
\newcommand{\scompressreg}{\scompress}
\newcommand{\scompressregred}{{\subap{\sred}{\scompress}}}

\newcommand{\scompressregconvred}{{\subap{\sconvred}{\scompress}}}

\newcommand{\scompressregconveqred}{{\bpap{\sconvred}{\scompress}{=}}}

\newcommand{\scompressregmred}{{\subap{\smred}{\scompress}}}
\newcommand{\compressregmred}{\mathrel{\scompressregmred}}
\newcommand{\scompressregconvmred}{{\subap{\sconvmred}{\scompress}}}

\newcommand{\scompressstreg}{\snlvarsucc} 
\newcommand{\scompressstregred}{{\subap{\sred}{\snlvarsucc}}}

\newcommand{\scompressstregeqred}{{\bpap{\sred}{\snlvarsucc}{=}}}

\newcommand{\constred}{\mathrel{\constred}}

\newcommand{\srule}{\varrho}

\newcommand{\srulep}{\supap{\srule}}

\newcommand\srulepos{\varrho_{pos}}
\newcommand\rulepos[1]{\srulepos^{#1}}

\newcommand\subrule[2]{{#1}.{#2}}

\newcommand{\lambdatg}{$\lambda$\nb-term-graph}
\newcommand{\lambdatgs}{$\lambda$\nb-term-graphs}

\newcommand{\agraph}{G}

\newcommand{\scope}{scope}
\newcommand{\extscope}{\scope$^+$}

\newcommand{\Vacregeager}{{\normalfont ($\Vacreg$)}\nb-eager}
\newcommand{\Vacstregeager}{{\normalfont ($\Vacstreg$)}\nb-eager}

\newcommand{\ainst}{\iota}
\newcommand{\binst}{\kappa}

\newcommand{\aproofsys}{{\cal S}}

\newcommand{\saprop}{P}

\newcommand{\aprop}{\funap{\saprop}}

\newcommand{\Reg}{\ensuremath{\normalfont\textbf{Reg}}}
\newcommand{\stReg}{\ensuremath{\normalfont\textbf{Reg}^{\bs{+}}}}
\newcommand{\stRegzero}{\ensuremath{\normalfont\textbf{Reg}_{\bs{0}}^{\bs{+}}}}

\newcommand{\Lambdaprefixreginf}{\ensuremath{\bs{(\lambda)}\bs{\Lambda}^{\hspace*{-1.5pt}\bs{\infty}}}}
\newcommand{\Lambdaprefixstreginf}{\ensuremath{\bs{(\lambda)}^{\hspace*{-0.75pt}\bs{+}}\hspace*{-1.75pt}\bs{\Lambda}^{\hspace*{-1.5pt}\bs{\infty}}}}
\newcommand{\Lambdaprefixstregposinf}{\ensuremath{\bs{(\lambda)}^{\hspace*{-0.75pt}\bs{+}}_{\hspace*{-0.5pt}\sbfpos}\hspace*{-0.25pt}\bs{\Lambda}^{\hspace*{-1.5pt}\bs{\infty}}}}

\newcommand{\Reginf}{\ensuremath{\supap{\normalfont\textbf{Reg}}{\bs{\infty}}}}
\newcommand{\stReginf}{\ensuremath{\supap{\normalfont\textbf{Reg}}{\bs{+}\bs{,}\bs{\infty}}}}

\newcommand{\Expr}{\ensuremath{{\normalfont\textbf{Expr}}}}
\newcommand{\Exprinf}{\ensuremath{\supap{\normalfont\textbf{Expr}}{\hspace*{-0.75pt}\bs{\infty}}}}
\newcommand{\Exprmu}{\ensuremath{\subap{\normalfont\textbf{Expr}}{\hspace*{-1pt}\bs{\ssmuabs}}}}
\newcommand{\Exprmumin}{\ensuremath{\subap{\normalfont\textbf{Expr}}{\hspace*{-1pt}\bs{\ssmuabs,\bs{-}}}}}
\newcommand{\Exprmuinf}{\ensuremath{\pbap{\normalfont\textbf{Expr}}{\hspace*{-0.75pt}\bs{\infty}}{{{\hspace*{-1pt}\bs{\ssmuabs}}}}}}

\newcommand{\Unfinf}{\ensuremath{\supap{\normalfont\textbf{Unf}}{\hspace*{1pt}\bs{\infty}}}}

\newcommand{\Vacreg}{\ensuremath{\text{\normalfont del}}} 
\newcommand{\Vacstreg}{\ensuremath{\snlvarsucc}} 

\newcommand{\annVacstreg}{\ensuremath{\snlvarsucc}} 

\newcommand{\sFIX}{\text{\normalfont FIX}}
\newcommand{\FIX}{{\normalfont (\ensuremath\sFIX)}}

\newcommand{\sFIXExpr}{\sFIX} 

\newcommand{\labscomp}{\sslabs}
\newcommand{\lappcomp}{@}
\newcommand{\bvarax}{\nlvar}

\newcommand{\sderivable}{\vdash}
\newcommand{\derivablein}[2]{\sderivable_{#1}\hspace*{0.5pt}{#2}}

\newcommand{\sDeriv}{{\cal D}}
\newcommand{\Deriv}{\sDeriv}
\newcommand{\Derivi}[1]{\sDeriv_{#1}}

\newcommand{\Derivann}{\hat{\Deriv}}
\newcommand{\Derivacc}{\Deriv'}

\newcommand{\sinfDeriv}{{\cal T}}
\newcommand{\infDeriv}{\sinfDeriv}
\newcommand{\infDerivi}[1]{\sinfDeriv_{#1}}
\newcommand{\infDerivacc}{{\sinfDeriv'}}

\newcommand{\Derivanni}{\subap{\hat{\Deriv}}}

\newcommand{\depth}[1]{\left|{#1}\right|}

\newcommand{\amarker}{l}

\newcommand{\funin}{\mathrel{:}}
\newcommand{\funap}[2]{{#1}({#2})}

\newcommand{\indap}[2]{#1_{#2}}

\newcommand{\sdefdby}{{:=}}
\newcommand{\defdby}{\mathrel{\sdefdby}}
\newcommand{\length}[1]{\left|{#1}\right|}
\newcommand{\lengthnormalsize}[1]{|{#1}\hspace*{-1pt}|}
\newcommand{\nb}{\nobreakdash}

\newcommand{\bs}{\boldsymbol}

\newcommand{\srestrictto}[2]{{#1}\!\mid_{#2}}
\newcommand{\restrictto}[2]{\funap{\srestrictto}}

\newcommand{\niks}{}

\newcommand{\ssbinrelcomp}{\cdot}
\newcommand{\sbinrelcomp}[2]{{#1}\mathrel{\ssbinrelcomp}{#2}}

\newcommand{\subap}[2]{#1_{#2}}
\newcommand{\supap}[2]{#1^{#2}}
\newcommand{\bpap}[3]{{#1}_{#2}^{#3}}
\newcommand{\pbap}[3]{{#1}^{#2}_{#3}}

\newcommand{\existsst}[2]{\exists{#1}.\;{#2}}

\newcommand{\descsetexpmid}{\mathrel{\vert}}

\newcommand{\descsetexp}[2]{\left\{{#1}\descsetexpmid{#2}\right\}}
\newcommand{\descsetexpnormalsize}[2]{\{{#1}\descsetexpmid{#2}\}}

\newcommand{\setexp}[1]{\left\{{#1}\right\}}

\newcommand{\aset}{A}
\newcommand{\bset}{B}

\newcommand{\sequence}[2]{\{{#1}\}_{#2}}

\newcommand{\enumsequence}[1]{\langle{#1}\rangle}

\newcommand{\punc}[1]{\ensuremath{\hspace*{1pt}{#1}}}

\newcommand{\tuple}[1]{\langle{#1}\rangle}

\newcommand{\nats}{\mathbb{N}}

\newcommand{\noopsort}[1]{}

\newcommand{\sbisim}[1][]{%
    \setbox0=\hbox{\kern-.1ex{$\leftrightarrow$}\kern-.1ex}
    \setbox1=\vbox{\hbox{\raise .1ex \box0}\hrule}%
    \ensuremath{\mathrel{\hbox{\kern.1ex\box1\kern.1ex}_{#1}}}
  }

\newcommand\sbisimstep\rightsquigarrow
\newcommand\sconvbisimstep\leftsquigarrow

\newcommand{\sltg}{{\cal G}}
\newcommand\altg\sltg

\definecolor{azure}{rgb}{0.94,1.00,1.00}
\definecolor{blue}{rgb}{0,0,0.5}
\definecolor{brown}{rgb}{.75,.25,.25}
\definecolor{cyan}{rgb}{0.25,0.88,0.82}
\definecolor{chocolate}{rgb}{0.82,0.41,0.12}
\definecolor{darkcyan}{rgb}{0.5,0,1}
\definecolor{darkgreen}{rgb}{0,0.39,0}
\definecolor{darkmagenta}{rgb}{0.5,0,0.5}
\definecolor{firebrick}{RGB}{175,25,25}
\definecolor{forestgreen}{rgb}{0.13,0.55,0.13}
\definecolor{lightcyan}{rgb}{0.88,1.00,1.00}
\definecolor{lightpink}{rgb}{1.00,0.71,0.76}
\definecolor{lightyellow}{rgb}{1.00,1.00,0.88}
\definecolor{lightgoldenrod}{rgb}{0.83,0.97,0.51}
\definecolor{lightgoldenrodyellow}{rgb}{0.98,0.98,0.82}
\definecolor{lightskyblue}{rgb}{0.53,0.81,0.98}
\definecolor{moccasin}{rgb}{1.00,0.89,0.71}
\definecolor{magenta}{rgb}{1,0,1}
\definecolor{navyblue}{rgb}{0,0,0.5}
\definecolor{orange}{rgb}{1.0,0.65,0.0}
\definecolor{orangered}{rgb}{1.0,0.27,0.0}
\definecolor{palegreen}{rgb}{0.60,0.98,0.60}
\definecolor{powderblue}{rgb}{0.69,0.88,0.90}
\definecolor{purple}{rgb}{1,0.5,1}
\definecolor{royalblue}{RGB}{65,105,225}
\definecolor{mediumblue}{RGB}{0,0,205}
\definecolor{cornflowerblue}{RGB}{100,149,237}
\definecolor{springgreen}{rgb}{0.0,1.0,0.5}
\definecolor{turquoise}{rgb}{0.25,0.88,0.82}
\definecolor{snow}{rgb}{1.00,0.98,0.98}
\definecolor{tan}{rgb}{0.82,0.71,0.55}
\definecolor{red}{rgb}{1,0,0}

\newcommand{\sbinds}{{\leftspoon}}
\newcommand{\binds}{\mathrel{\sbinds}}

\newcommand{\scaptures}{{\dashleftarrow}}
\newcommand{\captures}{\mathrel{\scaptures}}

\newcommand{\siscapturedby}{{\dashrightarrow}}
\newcommand{\iscapturedby}{\mathrel{\siscapturedby}}

\newcommand{\sbindsiter}[1]{{\indap{\leftspoon}{#1}}}

\newcommand{\siscapturedbyiter}[1]{{\indap{\dashrightarrow}{#1}}}

\newcommand{\spos}{\text{\normalfont pos}}
\newcommand{\sbfpos}{\text{\normalfont\bf pos}}
\newcommand{\sPositions}{{\mit Pos}}
\newcommand{\Positions}{\funap{\sPositions}}

\newcommand{\positionannotated}{po\-si\-tion-anno\-ta\-ted}

\newcommand{\positions}{\nats^*}
\newcommand{\vecpositions}{\vec{\nats}^*}

\newcommand{\rootpos}{\epsilon}

\newcommand{\apos}{p}
\newcommand{\bpos}{q}
\newcommand{\cpos}{r}
\newcommand{\dpos}{s}

\newcommand{\bposacc}{\bpos'}

\newcommand{\aposi}{\indap{\apos}}
\newcommand{\bposi}{\indap{\bpos}}
\newcommand{\cposi}{\indap{\cpos}}
\newcommand{\dposi}{\indap{\dpos}}

\newcommand{\aposacci}{\pbap{\apos}{\prime}}

\newcommand{\sunfoldsto}{\overset{\raisebox{-2.5pt}{\text{\scriptsize\normalfont unf}\hspace*{2.5pt}}}{\Longrightarrow}}
\newcommand{\unfoldsto}{\mathrel{\sunfoldsto}}

\newcommand{\sparselappred}{\indap{\sred}{\subrule\subparse@}}

\newcommand{\myparagraph}[1]{\noindent\emph{{#1}.}}

\begin{document}

\maketitle

\begin{abstract}
  We address a problem connected to the unfolding semantics of functional programming languages:
  give a useful characterization of those infinite \lambdaterms\ that are \lambdaletrecexpressible\  
  in the sense that they arise as infinite unfoldings of terms in \lambdaletreccal, the \lambdacalculus\ with ${\sf letrec}$.
  We provide two characterizations, using concepts we introduce for infinite \lambdaterms:
  regularity, strong regularity, and \bindcaptchains.
  It turns out that \lambdaletrecexpressible\ infinite \lambdaterms\ 
  form a proper subclass of the regular infinite \lambdaterms.
  In this paper we establish these characterizations only for expressibility in $\lambdamucal$,
  the \lambdacalculus\ with explicit $\mu$\nb-recursion.
  We show that for all infinite \lambdaterms~$\aiter$ the following are equivalent:
  (i):~$\aiter$~is $\lambdamucal$\nb-expressible;
  (ii):~$\aiter$~is strongly regular;
  (iii):~$\aiter$~is regular, and it only has finite \bindcaptchains.

  We define regularity and strong regularity for infinite \lambdaterms\ as two different
  generalizations of regularity for infinite first-order terms:
  as the existence of only finitely many subterms that are defined as the reducts of
  two rewrite systems for decomposing \lambdaterms. 
  These rewrite systems act on infinite \lambdaterms\ furnished with a bracketed prefix of abstractions 
  for collecting decomposed \lambdaabstraction{s} and keeping the terms closed under decomposition.
  They differ in which vacuous abstractions in the prefix are removed. 
  
  This report accompanies the article \cite{grab:roch:2013:RTA}, 
  and mainly differs from that by providing the proof
  of the characterization of $\lambdamucal$\nb-expressibility with \bindcaptchains. 
\end{abstract}   

\section{Introduction}
\label{sec:intro}

A syntactical core of functional programming languages
is formed by $\lambdaletreccal$, the \lambdacalculus\ with $\stxtletrec$, 
which can also be viewed as an abstract functional language.
Formally, $\lambdaletreccal$ is the extension of the \lambdacalculus\ 
by adding the construct $\stxtletrec$ for expressing recursion as well as explicit substitution. 
In a slightly enriched form (of e.g.\ Haskell's Core language) 
it is used as an intermediate language for the compilation of functional programs,
and as such it is the basis for optimizing program transformations.
A calculus that in some respects is weaker than $\lambdaletreccal$ is 
$\lambdamucal$, the \lambdacalculus\ with the binding construct $\smu$ for \murecursion.
Terms in $\lambdamucal$
can be interpreted directly as terms in $\lambdaletreccal$
(expressions $\muabs{\arecvar}{\funap{\almter}{\arecvar}}$ as 
 $\letrecin{\arecvar = \funap{\almter}{\arecvar}}{\arecvar}$),
but translations in the other direction are more complicated, and have weaker properties.  

For analyzing the execution behavior of functional programs,
and for constructing program transformations, 
expressions in \lambdaletreccal\ or in \lambdamucal\
are frequently viewed as finite representations of their unfolding semantics:
the infinite \lambdaterm\ that is obtained  
by completely unfolding all occurring recursive definitions, the $\stxtletrec$- or \mubindings, in the expression. 

In order to provide a theoretical foundation for such practical tasks,
we aim to understand how infinite \lambdaterms\ look like
that are expressible in $\lambdaletreccal$ or in $\lambdamucal$
in the sense that they are infinite unfoldings of expressions from the respective calculus.
In particular, we want to obtain useful characterizations of these classes of  
infinite \lambdaterms.
Quite clearly, any such infinite \lambdaterm\ must exhibit an, in some sense, repetitive structure
that reflects the cyclic dependencies present in the finite description.
This is because these dependencies are only `rolled out', and so are preserved, by a typically infinite, stepwise unfolding process. 

For infinite terms over a first-order signature there is a well-known concept
of repetitive structure, namely regularity. 
An infinite term is called `regular' if it has only a finite number of different subterms. 
Such infinite terms correspond to trees over ranked alphabets that are regular \cite{cour:1983}. 
Like regular trees also regular terms can be expressed finitely 
by systems of recursion equations \cite{cour:1983},
by `rational expressions' \cite[Def.\hspace*{1pt}4.5.3]{cour:1983} which correspond to $\mu$\nb-terms (see e.g.\ \cite{endr:grab:klop:oost:2011}),
or by terms using \letrecbindings.
In this context finite expressions denote infinite terms either
via a mathematical definition (a fixed-point construction, or induction on paths) 
or as the limit of a rewrite sequence consisting of unfolding steps.
Regularity of infinite terms coincides, furthermore,
with expressibility by finite terms enriched with either
of the binding constructs $\smu$ or $\stxtletrec$. 
%
It is namely well-known that both representations are equally expressive with respect to denoting infinite terms,
because a representation using $\stxtletrec$'{s} can also be transformed into one using $\smu$'s
while preserving the infinite unfolding. 

For infinite \lambdaterms, however, the situation is different:
A definition of regularity is less clear due to the presence of variable binding.
And there are infinite \lambdaterms\ that are regular in an intuitive sense,
yet apparently are not \lambdaletreccal- or \lambdamucalexpressible.
For example, the syntax trees of the infinite \lambdaterms\ $\aiter$ 
in Fig.~\ref{fig:ll:expressible} and $\biter$ in Fig.~\ref{fig:not:ll:expressible}
both exhibit a regular structure. 
But while $\aiter$ clearly is $\lambdamucal$- and $\lambdaletreccal$\nb-ex\-press\-ible
(by  
$\muabs{\arecvar}{\labs{\avar\bvar}{\lapp{\lapp{\arecvar}{\bvar}}{\avar}}}$
and 
$\letrecin{\arecvar = \labs{\avar\bvar}{\lapp{\lapp{\arecvar}{\bvar}}{\avar}}}{\arecvar}$,
respectively), 
this seems not to be the case for $\biter$:
the \lambdabindings\ in $\biter$ are infinitely entangled, 
which suggests that it cannot be the result of just an unfolding process. 
Therefore it appears that the intuitive notion of regularity
is too weak for capturing the properties of \lambdamucal- and of \lambdaletreccalexpressibility.
We note that actually these two properties coincide, 
because between $\stxtlambdamucal$\nb-terms and $\stxtlambdaletreccal$\nb-terms similar transformations
are possible as between representations with $\smu$ and with $\stxtletrec$ of infinite first-order terms
(but this will not be proved here). 

It is therefore desirable to obtain a precise, and conceptually satisfying, definition of regularity for infinite \lambdaterms\
that formalizes the intuitive notion, and that makes it possible to
prove that \lambdamucal-/\lambdaletreccalexpressible\ infinite \lambdaterms\
form only a proper subclass of the regular ones. 
Furthermore the question arises 
of whether the property of \lambdamucal-/\lambdaletreccalexpressibility\
can be captured by a stronger concept of regularity that is still natural in some sense.

We tackle both desiderata at the same time, and provide solutions,
but treat only the case of \lambdamucalexpressibility\ here. 
We introduce two concepts of regularity for infinite \lambdaterms.
For this, we devise two closely related rewrite systems
(infinitary Combinatory Reduction Systems)
that allow to `observe' infinite \lambdaterms\ by subjecting them 
to primitive decomposition steps and thereby obtaining `generated subterms'.  
Then regular, and strongly regular infinite \lambdaterms\
are defined as those 
that give rise to only a finite number of generated subterms 
in the respective decomposition system.
We establish the inclusion of the class of strongly regular
in the class of regular infinite \lambdaterms, and the fact 
that this is
a proper inclusion (by recognizing that the \lambdaterm~$\biter$ in Fig.~\ref{fig:not:ll:expressible}
is regular, but not strongly regular).
As our main result 
we show  that an infinite \lambdaterm\ is \lambdamucalexpressible\  (that is, expressible by a term in $\lambdamucal$)
if and only if it is strongly regular.
Here we say that a term $\almter$ in $\lambdamucal$ expresses an infinite \lambdaterm~$\citer$
if $\citer$ is the infinite unfolding of $\almter$.
An infinite unfolding is unique if it exists, and can be obtained 
as the limit of an infinite rewrite sequence of unfolding steps. 

This expressibility theorem is a special case of a result we reported in \cite{grab:roch:2012},
which states that strong regularity coincides with \lambdaletreccalexpressibility. 
That more general result settles a conjecture %
by Blom in \cite[Sect.\hspace*{2pt}1.2.4]{blom:2001}. 
Its proof is closely connected to the proof of the result on \lambdamucalexpressibility\ we give here,
which exhibits and highlights all the same features, but lacks 
the complexity that is inherent to the formal treatment of unfolding for terms in $\lambdaletreccal$. 

Additionally we give a result that explains the relationship between regularity and strong regularity
by means of the concept of `\bindcaptchain':
a regular infinite \lambdaterms\ is strongly regular if and only if it does not contain an infinite \bindcaptchain. 

This report is associated with the article \cite{grab:roch:2013:RTA} in the proceedings of RTA~2013.
It extends that article by providing more details on Section~\ref{sec:chains},
and it contains some changes in the exposition of the proof in Section~\ref{sec:express}.
Furthermore it contains some changes of notation%
  \footnote{For example, the symbol used for the version of the \lambdacalculus\ with abstraction prefixes 
            defined in Section~\ref{sec:regular}
            has been changed from 
            $\bs{(\lambda)}^{\hspace*{-1pt}\bs{\infty}}$ in \cite{grab:roch:2013:RTA} 
            to $\inflambdaprefixcal$ here.
            Similarly, the symbol used for a proof system 
            defined in Section~\ref{sec:proofs} 
            that is sound and complete for such terms
            has been changed from
            $\bs{(\hspace*{-0.7pt}\Lambda\hspace*{-0.7pt})^{\hspace*{-1pt}\infty}}$ in \cite{grab:roch:2013:RTA}
            to $\Lambdaprefixreginf$ here. 
            }   
as well as modifications and corrections of details. 
Also closely related is the report \cite{grab:roch:2012} about the more general case of expressibility in $\lambdaletreccal$.

\vspace{0.75ex}
\myparagraph{Overview}
  In Section~\ref{sec:regular} we introduce rewriting systems (infinitary \CRSs)
  for decomposing \lambdaterms\ into their generated subterms.
  By means of these systems we define regularity and strong regularity 
  for infinite \lambdaterms.
  In Section~\ref{sec:proofs}
  we provide sound and complete proof systems for these notions,
  that play a vital role for the proof of the main result in Section~\ref{sec:express} later.
  In Section~\ref{sec:chains}
  we develop the notion of \bindcaptchain\ in infinite \lambdaterms, 
  and show that strong regularity amounts to regularity plus the absence of infinite \bindcaptchains.
  In Section~\ref{sec:express} we establish 
  the correspondence between strong regularity and 
  \lambdamuexpressibility\ for infinite \lambdaterms.
  In the final Section~\ref{sec:conclusion} we place the results presented here
  in the context of our investigations about sharing in cyclic \lambdaterms.

\begin{figure}[tbp]
\setlength\tabcolsep{-0.5mm}
\vspace*{-3ex}

\hspace*{-2ex}
\begin{tabular}{ccccc}
\vcentered{\fig{simpleletrec-unf}}
& \vcentered{\fig{simpleletrec-chains}}
& \vcentered{\fig{simpleletrec-scopes}}
& \vcentered{\fig{simpleletrec-extscopes}}
& ~~~~\vcentered{\fig{simpleletrec-ltg}}
\\[1ex]
{syntax tree} & \parbox[t]{19ex}{binding--capturing\\ \centering chains}\hspace*{2ex} & {scopes} & {\extscope}s 
              & 
                \hspace*{2ex}
                \parbox[t]{30ex}{\centering
                                   $\sstregred$-generated subterms,\\
                                   \extscope{s} reflected on them}
\end{tabular}
\caption{\label{fig:ll:expressible}%
Strongly regular infinite \protect\lambdaterm~$\aiter$,
which can be expressed by the \protect\lambdamuterm~$\muabs{\arecvar}{\labs{\avar\bvar}{\lapp{\lapp{\arecvar}{\bvar}}{\avar}}}$.}
 \vspace*{-2ex}
\end{figure}

\section{Regular and strongly regular infinite \protect\lambdaterms} 
\label{sec:regular}

In this section we motivate the introduction of higher-order versions of regularity,
and subsequently introduce the concepts of regularity and strong regularity
for infinite \lambdaterms.

For higher-order infinite terms such as infinite \lambdaterms,
regularity has been used with as meaning
the existence of a first-order syntax tree with named variables that is regular
(e.g.\ in \cite{ario:klop:1997,ario:blom:1997}).
For example, the infinite \lambdaterms\ $\aiter$ and $\biter$ from Figures~\ref{fig:ll:expressible} and \ref{fig:not:ll:expressible} 
are regular in this sense.
However, such a definition of regularity has the drawback that 
it depends on a first-order representation 
(as syntax trees with named abstractions and variables)
that is not invariant under
$\alpha$\nb-con\-ver\-sion, the re\-na\-ming of bound variables.
Note that the syntax trees of $\aiter$ and $\biter$ 
have renaming variants that contain infinitely many variables, and that for this reason are not regular as first-order trees.
It is therefore desirable to obtain a definition of regularity that uses the condition for the first-order case
but adapts the notion of subterm to \lambdaterms, 
and that pertains to a formulation of infinite \lambdaterms\ as higher-order terms. 

Viable notions of subterm for \lambdaterms\ in a higher-order formalization require a stipulation 
on how to treat variable binding 
when stepping from a \lambdaabstraction\ $\labs{\cvar}{\cter}$ into its body~$\cter$.
For this purpose we enrich the syntax of \lambdaterms\ with a bracketed prefix of
abstractions (similar to a proof system for weak $\mu$\nb-equality in
\cite[Fig.\hspace*{1.5pt}12]{endr:grab:klop:oost:2011}),
and consider $\flabs{\cvar}{\cter}$ as a `generated subterm' of $\labs{\cvar}{\cter}$,
obtained by a \lambdaabstraction\ decomposition applied to
$\femptylabs{\labs{\cvar}{\cter}}$, where $()$ is the empty prefix. 
An expression $\flabs{\avari{1}\ldots\avari{n}}{\aiter}$ 
represents a partially decomposed \lambdaterm: 
the body $\aiter$ typically contains free occurrences of variables 
that in the original \lambdaterm\ were bound by \lambdaabstractions\
but have since been split off by decomposition steps.
The role of such abstractions has then been taken over by abstractions in the prefix $\sflabs{\avari{1}\ldots\avari{n}}$.
In this way expressions with abstraction prefixes are kept closed under decomposition steps. 

We formulate infinite \lambdaterms\ and their prefixed variants
as terms in \iCRSs\ (infinitary Combinatory Reduction Systems)
for which we draw on the literature.
By \emph{\iCRS\nb-terms} we mean $\alpha$\nb-equi\-va\-lence classes 
of \iCRS\nb-preterms that are defined by metric completion from finite \CRS\nb-terms
\cite{kete:simo:2011}. 
For denoting and manipulating infinite terms we use customary notation for finite terms.
In order to simplify our exposition we restrict to closed terms,
but at one stage (a proof system in Section~\ref{sec:express}) we allow constants in our terms. 

Note that we do not formalize $\beta$\nb-re\-duc\-tion 
since we are only concerned with a static analysis of infinite \lambdaterms\ 
and later with finite expressions that express them via unfolding. 

\begin{definition}[\iCRS-representation of $\inflambdacal$]
    \normalfont\label{def:sigs:lambdacal:lambdaletrec:CRS}  
  The \CRS\nb-signature for the \lambdacalculus\ $\lambdacal$ and the infinitary \lambdacalculus\ $\inflambdacal$
  consists of the set $\siglcCRS = \setexp{\slappCRS,\,\slabsCRS}$ 
  where $\slappCRS$ is a binary and $\slabsCRS$ a unary function symbol.
  By $\Ter{\inflambdacal}$   
  we denote the set of infinite closed \iCRS\nb-terms over $\siglcCRS$
  with the restriction that CRS-abstraction can only occur as an argument of an $\slabsCRS$-symbol.
  Note that we restrict attention to closed terms, 
  and that here and below we subsume finite \lambdaterms\ among the infinite ones.
\end{definition}

\begin{example}
  The \lambdaterm\ $\labs{\avar\bvar}{\lapp{\bvar}{\avar}}$ in \CRS-notation is
  $\labsCRS{\avar}{\labsCRS\bvar{\lappCRS{\bvar}\avar}}$.
\end{example}

\begin{definition}[\iCRS-representation of $\inflambdaprefixcal$]
  \label{def:sig:lambdaprefixcal:CRS}
  The \CRS-signature $\siglpcCRS$ for $\inflambdaprefixcal$, the version of $\inflambdacal$ 
  with bracketed abstractions, extends $\siglcCRS$ 
  by unary function symbols of arbitrary arity:
  $\siglpcCRS = \siglcCRS \cup \descsetexp{\sflabsCRS{n}}{n\in\nats}$.
  Prefixed \lambdaterm{s}
  $\flabsCRS{n}{\avari{1}}{\ldots\absCRS{\avari{n}}{\ater}}$ will informally be
  denoted by $\flabs{\avari{1}\ldots\avari{n}}{\ater}$, abbreviated as
  $\flabs{\vec{\avar}}{\ater}$, or $\femptylabs\ater$ in case of an empty prefix.
  By $\Ter{\inflambdaprefixcal}$ we denote the set of closed \iCRS-terms 
  over $\siglpcCRS$ of the form
  $\flabsCRS{n}{\avari{1}}{\ldots\absCRS{\avari{n}}{\ater}}$ 
  for some $n\in\nats$ and some term $\ater$ over the signature $\siglcCRS$ with
  possible free occurrences of $\avari{1}$, \ldots, $\avari{n}$, 
  and
  the restriction that a CRS-abstraction can only occur as an argument of an
  $\slabsCRS$-symbol.
\end{definition}

\begin{example}
  The term 
  $\flabsCRS{1}{\avar}{\labsCRS\bvar{\lappCRS{\bvar}\avar}}$
  in $\Ter{\inflambdaprefixcal}$
  can be written, in informal notation, as the prefixed \lambdaterm\
  $\flabs{\avar}{\labs{\bvar}{\lapp{\bvar}{\avar}}}$.
\end{example}

 

On these prefixed \lambdaterms, we define two rewrite strategies $\sregred$ and $\sstregred$
that deconstruct infinite \lambdaterms\ by steps that decompose
applications and \lambdaabstractions, and take place just below the marked abstractions.
They differ with respect to which vacuous prefix bindings they remove:
while $\sregred$\nb-steps drop such bindings always before 
steps over applications and $\lambda$\nb-abstractions,
$\sstregred$\nb-steps remove vacuous bindings only if they occur 
at the end of the abstraction prefix. 
These rewrite strategies will define respective notions of `generated subterm', 
and will give rise to 
                      two concepts of regularity: 
a \lambdaterm\ is called regular\discretionary{/}{}{/}strongly regular
if its set of $\sregred$-reachable\discretionary{/}{}{/}$\sstregred$-reachable generated subterms is finite.

\begin{definition}[decomposing $\inflambdaprefixcal$-terms with 
        rewrite strategies $\sregred$ and $\sstregred$] 
  \label{def:RegCRS:stRegCRS}
We consider the following \CRS-rules over
$\siglpcCRS$ in informal notation:%
  \footnote{E.g.\ explicit form of scheme 
   ($\srulep{\snlvarsucc}$):
    $\flabsCRS{n+1}{\avari{1}\ldots\avari{n+1}}{\cmetavar{\avari{1},\ldots,\avari{n}}}
      \red 
    \flabsCRS{n}{\avari{1}\ldots\avari{n}}{\cmetavar{\avari{1},\ldots,\avari{n}}}$.
            }
  \begin{align*}
    (\srulep{\slappdecompi{i}}): \hspace*{-0.5ex}
      & &
    \flabs{\avari{1}\ldots\avari{n}}{\lapp{\ateri{0}}{\aiteri{1}}}
      & {} \red
    \flabs{\avari{1}\ldots\avari{n}}{\aiteri{i}}  
      & & \hspace*{0ex} (i\in\{0,1\})
    \displaybreak[0]\\
    (\srulep{\slabsdecomp}): \hspace*{-0.5ex}
      & &
    \flabs{\avari{1}\ldots\avari{n}}{\labs{\avari{n+1}}{\aiteri{0}}}
      & {} \red
    \flabs{\avari{1}\ldots\avari{n+1}}{\aiteri{0}} 
    \displaybreak[0]\\
    (\srulep{\snlvarsucc}): \hspace*{-0.5ex}
      & &
    \flabs{\avari{1}\ldots\avari{n+1}}{\aiteri{0}}
      & {} \red
    \flabs{\avari{1}\ldots\avari{n}}{\aiteri{0}}  
      & &
      \hspace*{-5ex} (\text{if binding $\lambda\avari{n+1}$ is vacuous})
    \displaybreak[0]\\  
    (\srulep{\scompress}): \hspace*{-0.5ex}
      & &
    \flabs{\avari{1}\ldots\avari{n+1}}{\aiteri{0}}
      & {} \red
    \flabs{\avari{1}\ldots\avari{i-1}\avari{i+1}\ldots\avari{n+1}}{\aiteri{0}}   
      & & 
      \hspace*{-0ex} (\text{if bind.\ $\lambda\avari{i}$ is vacuous})
  \end{align*}
We call an occurrence $\aocc$ of a binding like a \lambdaabstraction~$\slabs{\cvar}$ or a \CRSabstraction\ $\sabsCRS{\cvar}$ 
 in a term $\citer$ \emph{vacuous}
 if $\citer$ does not contain a variable occurrence of $\cvar$ that is bound by $\aocc$.
 
The \iCRS\ with these rules induces an \ARS\ (abstract rewriting system) $\aARS$ on infinite terms over 
                                                                                                $\siglpcCRS$. 
By \decompARS\ we denote the sub-\ARS\ of $\aARS$
with its set of objects restricted to $\Ter{\inflambdaprefixcal}$.
Note that $\Ter{\inflambdaprefixcal}$ is closed under steps in \decompARS.
By $\slappdecompired{0}$, $\slappdecompired{1}$, $\slabsdecompred$,
$\scompressstregred$, $\scompressregred$ we denote the rewrite
relations induced by $\decompARS$\nb-steps with respect to rules
$\srulep{\slappdecompi{0}}$, $\srulep{\slappdecompi{1}}$,
$\srulep{\slabsdecomp}$, $\srulep{\snlvarsucc}$, $\srulep{\scompress}$.
We define $\RegCRS$ ($\stRegCRS$) as the sub-\ARS\ of \decompARS\ that arises from dropping steps that are:
\begin{itemize}
  \item due to $\srulep{\snlvarsucc}$ ($\srulep{\scompress}$), so that the prefix can be shortened only by $\srulep{\scompress}$-steps ($\srulep{\snlvarsucc}$-steps).
  \item 
        due to rules other than $\srulep{\scompress}$ ($\srulep{\snlvarsucc}$)
        but whose source is also a source of a $\srulep{\scompress}$-step ($\srulep{\snlvarsucc}$-step).
\end{itemize}
$\RegCRS$ ($\stRegCRS$)  is \emph{$\srulep{\scompress}$-eager} (\emph{$\srulep{\snlvarsucc}$}-eager) in the sense that on each path $\srulep{\scompress}$-steps ($\srulep{\snlvarsucc}$-steps) occur as soon as possible.
We denote by $\sregred$ ($\sstregred$) the rewrite strategy 
induced by $\RegCRS$ ($\stRegCRS$).%
\footnote{We use `rewrite strategy' for a relation on terms, and not for a sub-\ARS\ of a \CRS\nb-induced \ARS\ \cite{terese:2003}.}
\end{definition}

\begin{example}\label{ex:regred:stregred}
  Using the recursive equation
  $\aiter = \labs{\avar\bvar}{\lapp{\lapp{\aiter}{\bvar}}{\avar}}$
  as a description for the infinite \lambdaterm~$\aiter$ in Fig.$\,$\ref{fig:ll:expressible},
  we find that decomposition by $\sstregred$\nb-steps proceeds as follows,
  repetitively:
  \begin{center}
    $
    \femptylabs{\aiter} ~~~
    \flabs{\avar}{\labs{\bvar}{\lapp{\lapp{\aiter}{\bvar}}{\avar}}} ~~~
    \flabs{\avar\bvar}{\lapp{\lapp{\aiter}{\bvar}}{\avar}} ~~~
    \begin{array}{lll}
    \flabs{\avar\bvar}{{{\lapp{\aiter}{\bvar}}}} &
    \begin{array}{llll}
    \flabs{\avar\bvar}{\aiter} ~~~ &
    \flabs{\avar}{\aiter}  ~~~ &
    \femptylabs{\aiter} ~~~ &
    \ldots
    \\
    \flabs{\avar\bvar}{\bvar}
    \end{array}
    \\[2ex]
    \flabs{\avar\bvar}{\avar} & \;\;
    \flabs{\avar}{\avar}
    \end{array}
    $
  \end{center}
  (in a tree that branches to the right).
  Note that removal steps for vacuous bindings take place only at the end of the prefix.
  See Fig.$\,$\ref{fig:ll:expressible} right for 
  the reduction graph of $\femptylabs{\aiter}$ with displayed sorts of decomposition steps. 
  Although $\scompressstregred$\nb-steps also are $\scompressregred$\nb-steps, 
  this decomposition is not also one according to $\sregred\,$,
  because e.g.\ the step
  $\flabs{\avar\bvar}{{{\lapp{\aiter}{\bvar}}}} 
     \lappdecompired{1}
   \flabs{\avar\bvar}{\bvar}$  
  is not $\srulep{\scompress}$-eager. 
\end{example}

The rules $\srulep{\snlvarsucc}$ are related to the
de\mbox{ }Bruijn notation of \lambda-terms.
Consider
$\labs{x}{\lapp{(\labs{y}{\lapp{x}{x}})}{x}}$
which in de\mbox{ }Bruijn notation is 
$\labs{}{\lapp{(\labs{}{\lapp{1}{1}})}{0}}$
and when using Peano numerals
$\labs{}{\lapp{(\labs{}{\lapp{S(0)}{S(0)}})}{0}}$.
Now if the symbols $\snlvarsucc$ are allowed to appear `shared' and occur further up in the term as in
$\labs{}{\lapp{(\labs{}{S(\lapp{0}{0})})}{0}}$, then this term structure
corresponds to the decomposition with $\sstregred$. 

To understand the difference between $\sregred$ and $\sstregred$, consider
the notions of scope and \extscope, illustrated in
Figures~\ref{fig:ll:expressible} and \ref{fig:not:ll:expressible}. The scope of
an abstraction is the smallest connected portion of a syntax tree that contains
the abstraction itself as well as all of its bound variable occurrences.
And \extscope{s} extend scopes minimally so that the resulting areas appear properly nested.
For a precise definition we refer to
\cite[Sect.\hspace*{2pt}4]{grab:roch:2012}. 
As can be seen in Figures~\ref{fig:ll:expressible} and
\ref{fig:not:ll:expressible}, applications of $\srulep\scompress$
($\srulep\snlvarsucc$) coincide with the positions where scopes (\extscope{s})
are closed.

\begin{definition}[regular/strongly regular \lambdaterms, generated subterms]\label{def:gST:reg:streg}
Let $\ater\in\Ter\inflambdacal$. We define the sets $\gSTreg{\aiter}$ and $\gSTstreg{\aiter}$
of \emph{generated subterms} of $\aiter$ \emph{with respect to $\sregred$} and $\sstregred$:
\begin{center}
  $\gSTreg{\ater} \defdby \descsetexp{\bter\in\Ter\inflambdaprefixcal}{\femptylabs\ater \regmred \bter}$
    \hspace*{2ex}
  $\sgSTstreg(\ater) \defdby \descsetexp{\bter\in\Ter\inflambdaprefixcal}{\femptylabs\ater \stregmred \bter}$
\end{center}
We say that $\aiter$ is \emph{regular} (\emph{strongly regular})
if $\aiter$ has only finitely many generated subterms with respect to $\sregred$ 
(respectively, with respect to $\sstregred$).
\end{definition}

\begin{figure}[tbp]
\setlength\tabcolsep{0.5ex}
\vspace*{-2.5ex}

\hspace*{-7ex}  
\begin{tabular}{ccccc}
  \hspace*{5ex}
  \vcentered{\fig{pstricks/entangled.lgr}}
& \hspace*{1.75ex}
  \vcentered{\fig{entangled-chains}}
& \hspace*{1ex}
  \vcentered{\fig{entangled-scopes}}
& \hspace*{1.75ex}
  \vcentered{\fig{entangled-extscopes}}
& \hspace*{1.75ex}
  \vcentered{\fig{entangled-ltg}}
\\[1ex]
{syntax tree} & \parbox[t]{19ex}{binding--capturing\\ \centering chain} & {scopes} & {\extscope{s}} & \text{$\sregred$-generated subterms}
\end{tabular}
\caption{\label{fig:not:ll:expressible}%
  The regular infinite \protect\lambdaterm~$\biter$ that is not strongly regular,
  and not \protect\lambdamuexpressible.} 
\vspace*{-3ex}\mbox{}  
\end{figure}

\begin{example}\label{ex:def:gST:reg:streg}
  From the $\sstregred$\nb-de\-com\-po\-si\-tion in Example~\ref{ex:regred:stregred} and Fig.$\,$\ref{fig:ll:expressible}
  of the infinite \lambdaterm~$\aiter$ in Fig.$\,$\ref{fig:ll:expressible}
  it follows that 
  $\gSTstreg{\aiter}$ consists of 9 generated subterms.
  Hence $\aiter$ is strongly regular. 
   
  The situation is different for the infinite \lambdaterm~$\biter$ in Fig.\,\ref{fig:not:ll:expressible}.
  When represented as the term $\labs{\avar}{\funap{R}{\avar}}$
  together with the \CRS\nb-rule $\funap{R}{\sametavar} \red \labs{\bvar}{\lapp{\funap{R}{\bvar}}{\sametavar}}$,
  its $\sstregred$\nb-decomposition is:
  \begin{center}
    \scalebox{0.92}{
      $
      \femptylabs{\biter} ~~
      \flabs{\avar}{\funap{R}{\avar}} ~~
      \flabs{\avar\bvar}{\lapp{{\funap{R}{\bvar}}}{\avar}}
      \begin{array}{lll}
      \flabs{\avar\bvar}{{{\funap{R}{\bvar}}}} &
      \flabs{\avar\bvar\cvar}{\lapp{{\funap{R}{\cvar}}}{\bvar}}
                                               & \hspace*{-2ex}
      \begin{array}{lll}
      \flabs{\avar\bvar\cvar}{{{\funap{R}{\cvar}}}} &
      \flabs{\avar\bvar\cvar\dvar}{\lapp{{\funap{R}{\dvar}}}{\cvar}}
                                               &
      \ldots
      \\
      \flabs{\avar\bvar\avar}{\bvar} & \flabs{\avar\bvar}{\bvar}
      \end{array}
      \\
      \flabs{\avar\bvar}{\avar} & \flabs{\avar}{\avar}
      \end{array}
      $
      }
  \end{center} 
  Since here the prefixes grow unboundedly, $\biter$ 
  has infinitely many $\sstregred$\nb-generated subterms,  
%
%
and hence $\biter$ is not strongly regular. 
But its $\sregred$\nb-decomposition exhibits again a repetition
as can be seen from the reduction graph in Fig.\,\ref{fig:not:ll:expressible} on the right.  
Note that a vacuous binding from within a prefix is removed.
$\femptylabs{\biter}$ has 6 only different $\sregred$\nb-reducts.
Hence $\biter$ is regular.

For infinite \lambdaterms\ like 
$\lapp{\lapp{\lapp{(\labs{x_1}{x_1})}{(\labs{x_1}{\labs{x_2}{x_2}})}}{(\labs{x_1}{\labs{x_2}{\labs{x_3}{x_3}}})}}{\ldots}$
that do not have any regular pseudoterm syntax-trees,
both $\sstregred$-decomposition and $\sregred$\nb-decomposition yield infinitely many generated subterms,
and hence they are neither regular nor strongly regular. 
\end{example}

For a better understanding of the precise relationship between $\sregred$ and $\stregred$,
and eventually of the two concepts of generated subterm and of regularity, 
we gather a number of basic properties of these rewrite strategies and their constituents. 

\begin{proposition}\label{prop:rewprops:RegCRS:stRegCRS}
  The restrictions of the rewrite relations from Def.~\ref{def:RegCRS:stRegCRS}
  to $\Ter{\inflambdaprefixcal}$, the set of objects of $\RegARS$ and $\stRegARS$,
  have the following properties:%
  \begin{enumerate}[(i)]
    \item{}\label{prop:rewprops:RegCRS:stRegCRS:item:i}
       $\scompressregred$ is confluent, and terminating.
    \item{}\label{prop:rewprops:RegCRS:stRegCRS:item:iii} 
      $\scompressstregred \subseteq \scompressregred$. 
      Furthermore, $\scompressstregred$ is deterministic, hence confluent,
      and terminating.
    \item{}\label{prop:rewprops:RegCRS:stRegCRS:item:ii}
      $\scompressregred$ one-step commutes with $\slabsdecompred$, $\slappdecompired{0}$, $\slappdecompired{1}$, 
      and one-step sub-commutes with $\scompressstregred\,$;
      $\scompressregred$ postpones over $\slabsdecompred$, $\slappdecompired{0}$, $\slappdecompired{1}$ and $\scompressstregred$.
      Formulated symbolically, this means:\vspace{-1ex}
      \begin{align*}
        \sbinrelcomp{\scompressregconvred}{\slabsdecompred} 
          & \;\subseteq\; 
        \sbinrelcomp{\slabsdecompred}{\scompressregconvred}
          & 
        \sbinrelcomp{\scompressregconvred}{\slappdecompired{i}} 
          & \;\subseteq\; 
        \sbinrelcomp{\slappdecompired{i}}{\scompressregconvred}
          & 
        \sbinrelcomp{\scompressregconvred}{\scompressstregred} 
          & \;\subseteq\; 
        \sbinrelcomp{\scompressstregeqred}{\scompressregconveqred}    
        \\
        \sbinrelcomp{\scompressregred}{\slabsdecompred}
          & \;\subseteq\;
        \sbinrelcomp{\slabsdecompred}{\scompressregred}  
          & 
        \sbinrelcomp{\scompressregred}{\slappdecompired{i}}
          & \;\subseteq\;
        \sbinrelcomp{\slappdecompired{i}}{\scompressregred}
          &
        \sbinrelcomp{\scompressregred}{\scompressstregred} 
          & \;\subseteq\; 
        \sbinrelcomp{\scompressstregred}{\scompressregred}
      \end{align*}
    \item{}\label{prop:rewprops:RegCRS:stRegCRS:item:iv}
      Normal forms of $\sregred$ and $\sstregred$ are of the form
      $\flabs{\avar}{\avar}$, and
      $\flabs{\avari{1}\ldots\avari{n}}{\avari{n}}$, respectively. 
    \item{}\label{prop:rewprops:RegCRS:stRegCRS:item:v}
      $\sregred$ and $\sstregred$ are finitely branching,
      and, on finite terms, terminating.
  \end{enumerate}
\end{proposition}

\begin{proof}
  These properties, including those concerning commutation of steps, are easy to verify
  by analyzing the behavior of the rewrite rules in $\RegARS$ on terms of $\Ter{\inflambdaprefixcal}$.
%
\end{proof}
 
\begin{proposition}\label{prop:compress:prefix:RegARS}
  \begin{enumerate}[(i)]
    \item\label{prop:compress:prefix:RegARS:item:i}
      Let $\flabs{\vec{\avar}}{\aiter}$ be a term in $\Ter{\inflambdaprefixcal}$
      with $\length{\vec{\avar}} = n\in\nats$.
      The number of terms $\flabs{\vec{\bvar}}{\biter}$ in $\Ter{\inflambdaprefixcal}$
      with $\flabs{\vec{\bvar}}{\biter} \compressregmred \flabs{\vec{\avar}}{\aiter}$
      and $\length{\vec{\bvar}} = n+k\in\nats$ 
      is $\binom{n+k}{n}$.
    \protect\item\label{prop:compress:prefix:RegARS:item:ii}
      Let $\aset\subseteq\Ter{\inflambdaprefixcal}$ be a finite set, and $k\in\nats$.
      Then also the set of terms in $\Ter{\inflambdaprefixcal}$ that are the form $\flabs{\vec{\bvar}}{\biter}$ 
      with $\length{\vec{\bvar}}\le k$ and that have a $\scompressregmred$\nb-reduct in $\aset$
      is finite.                                            
  \end{enumerate}         
\end{proposition}

\vspace{-1ex}
We state a lemma about a close connection between 
$\sregred$- and $\sstregred$\nb-rewrite sequences.  
\vspace{-1ex}
 
\begin{lemma}\label{lem:projection:lifting:RegCRS:stRegCRS}
  \begin{enumerate}[(i)]
    \item{}\label{lem:projection:lifting:RegCRS:stRegCRS:item:projection}
      On $\Ter{\inflambdaprefixcal}$ it holds:
      $\sbinrelcomp{\scompressregconvmred}{\sstregred}
         \;\subseteq\;
       \sbinrelcomp{\sbinrelcomp{\scompressregnfred}{\sregeqred}}{\scompressregconvmred}\,$,
      where $\scompressregnfred$ denotes many-step $\scompressregred\,$-reduction to $\scompressregred\,$-normal form. 
      As a consequence of this and of
      $\sbinrelcomp{\scompressregnfred}{\sregeqred}
         \;\subseteq\;
       \sregmred\,$, 
      every finite or infinite rewrite sequence in $\Ter{\inflambdaprefixcal}$: 
      \begin{center}
        $
        \arewseq \funin
          \flabs{\vec{\avar}_0}{\aiteri{0}}
            \stregred
          \flabs{\vec{\avar}_1}{\aiteri{1}}
            \stregred
          \ldots
            \stregred  
         \flabs{\vec{\avar}_k}{\aiteri{k}}
            \stregred
         \ldots
         $
      \end{center}   
      projects over a sequence
      $\crewseq \funin \flabs{\vec{\avar}_0}{\aiteri{0}} \compressregmred \flabs{\vec{\avar}'_0}{\aiteri{0}}$
      to a rewrite sequence of the form:
      \begin{center}
        $
        \check{\arewseq} \funin
          \flabs{\vec{\avar}'_0}{\aiteri{0}}
            \regmred  
          \flabs{\vec{\avar}'_1}{\aiteri{1}}
            \regmred  
          \;\:\ldots\;\:
            \regmred  
          \flabs{\vec{\avar}'_k}{\aiteri{k}}
            \regmred  
          \ldots
          $
      \end{center}    
      in the sense that
      $\flabs{\vec{\avar}_k}{\aiteri{k}} \compressregmred \flabs{\vec{\avar}'_k}{\aiteri{k}}\,$
      for all $k\in\nats$ less or equal to the length of $\arewseq$.
    \item{}\label{lem:projection:lifting:RegCRS:stRegCRS:item:lifting} 
      On $\Ter{\inflambdaprefixcal}$ it holds: 
      $\sbinrelcomp{\scompressregmred}{\sregred}
         \;\subseteq\;
       \sbinrelcomp{\sbinrelcomp{\scompressstregnfred}{\sstregeqred}}{\scompressregmred}\,$.
      Due to this and 
      $\sbinrelcomp{\scompressstregnfred}{\sstregeqred} 
         \;\subseteq\;
        \sstregmred\,$,
      every rewrite sequence
        $  
        \arewseq \funin
          \flabs{\vec{\avar}'_0}{\aiteri{0}}
            \regred
          \flabs{\vec{\avar}'_1}{\aiteri{1}}
            \regred
          \ldots
            \regred  
         \flabs{\vec{\avar}'_k}{\aiteri{k}}
            \regred
         \ldots
         $
      in $\Ter{\inflambdaprefixcal}$
      lifts over a sequence
      $\crewseq \funin \flabs{\vec{\avar}_0}{\aiteri{0}} \compressregmred \flabs{\vec{\avar}'_0}{\aiteri{0}}$
      to a $\stregred$\nb-rewrite sequence of the form:
        $
        \Hat{\arewseq} \funin
          \flabs{\vec{\avar}_0}{\aiteri{0}}
            \stregmred 
          \flabs{\vec{\avar}_1}{\aiteri{1}}
            \stregmred 
          \;\:\ldots\;\:
            \stregmred 
          \flabs{\vec{\avar}_k}{\aiteri{k}}
            \stregmred 
          \ldots
          $
      in the sense that
      $\flabs{\vec{\avar}_k}{\aiteri{k}} \compressregmred \flabs{\vec{\avar}'_k}{\aiteri{k}}\,$
      for all $k\in\nats$ less or equal to the length of $\arewseq$.
  \end{enumerate}     
\end{lemma}

\begin{proof}
  The inclusion properties in (\ref{lem:projection:lifting:RegCRS:stRegCRS:item:projection})
  and (\ref{lem:projection:lifting:RegCRS:stRegCRS:item:lifting})
  can be shown by easy arguments with diagrams
  using the commutation properties in Proposition\,\ref{prop:rewprops:RegCRS:stRegCRS}, (\ref{prop:rewprops:RegCRS:stRegCRS:item:ii}),
  as well as (\ref{prop:rewprops:RegCRS:stRegCRS:item:i}) and (\ref{prop:rewprops:RegCRS:stRegCRS:item:iii}) from there.
\end{proof}

\vspace{-1ex}
Now we are able to establish that strong regularity implies regularity for infinite \lambdaterms. 

\vspace{-1ex}
\begin{proposition}\label{prop:def:reg:streg}
      Every strongly regular infinite \lambdaterm\ is also regular.
      Finite $\lambda$\nb-terms are both regular and strongly regular. 
\end{proposition}

\begin{proof} 
  Let $\aiter$ be a strongly regular infinite \lambdaterm. Therefore $\gSTstreg{\aiter}$ is finite. 
  Since every $\sregred\hspace*{0.5pt}$-re\-write-se\-quence from $\femptylabs{\aiter}$ lifts
  to a $\sstregred\,$-re\-write-se\-quence from $\femptylabs{\aiter}$ over $\scompressregmred$\nb-com\-pres\-sion
  due to Lemma~\ref{lem:projection:lifting:RegCRS:stRegCRS},~(\ref{lem:projection:lifting:RegCRS:stRegCRS:item:lifting}),
  every term in $\gSTreg{\aiter}$ is the $\scompressregmred$\nb--com\-pres\-sion of a term in $\gSTstreg{\aiter}$. 
  Then it follows by Proposition\,\ref{prop:compress:prefix:RegARS}, (\ref{prop:compress:prefix:RegARS:item:i}),
  that also $\gSTreg{\aiter}$ is finite. Hence $\aiter$ is also regular.
  
  Let $\aiter$ be a finite \lambdaterm. Due to to Proposition\,\ref{prop:rewprops:RegCRS:stRegCRS}, (\ref{prop:rewprops:RegCRS:stRegCRS:item:v}),
  K\H{o}nig's Lemma can be applied to the reduction graph of $\femptylabs{\aiter}$ with respect to $\sstregmred$ 
  to yield that $\aiter$ has only finitely many generated subterms with respect to $\sstregmred$. Hence $\aiter$ is strongly regular.
%
%
%
\end{proof}

\section{Proving regularity and strong regularity} 
\label{sec:proofs}

As a preparation for the proof of the main expressibility result in Section~\ref{sec:express},
we introduce, in this section, proof systems for regularity and strong regularity of infinite \lambdaterms\
that formulate these notions in terms of derivability:
the systems $\Reginf$\hspace*{-1pt} and $\stReginf$\hspace*{-1pt} with typically infinite derivations,
and the systems $\Reg$, $\stReg$, and $\stRegzero$ for provability by finite derivations.
A completed derivation of $\femptylabs{\biter}$ in $\Reginf$\hspace*{-1pt} (in $\stReginf$\hspace*{-1pt}) corresponds to
the `tree unfolding' of the $\sregred\,$-reduction graph (the $\sstregred\,$-re\-duc\-tion graph) of $\femptylabs{\biter}$, 
which is a tree that describes all $\sregred\,$-(resp.\ $\stregred\,$-)re\-write sequences from $\femptylabs{\biter}$.
Closed derivations of $\femptylabs{\biter}$ in $\Reg$ (in $\stReg$, or $\stRegzero$)
correspond to finite unfoldings of the $\sregred\,$-reduction graph (the $\sstregred\,$-re\-duc\-tion graph)
into a graph with only vertical sharing. 

We start by introducing proof systems for well-formed prefixed terms, that is, terms from the set $\Ter{\inflambdaprefixcal}$).

\begin{definition}[proof systems $\Lambdaprefixreginf$, $\Lambdaprefixstreginf$ for well-formed $\inflambdacal$\nb-terms] 
    \label{def:Lambdaprefixreg:Lambdaprefixstreg}
  The proof systems defined here act on \CRS-terms over signature $\siglpcCRS$ 
  as formulas, and are Hilbert-style systems
  for finite or infinite prooftrees (of depth $\le\omega$). 
  The system $\Lambdaprefixstreginf$ 
  has the axioms ($\bvarax$) and the rules ($\lappcomp$), ($\labscomp$), and ($\Vacstreg$) 
  in Fig.\,\ref{fig:stReg:Lambdaprefixstreg:stRegzero}.
  The system $\Lambdaprefixreginf$ arises from $\Lambdaprefixstreginf$ 
  by replacing the axioms ($\bvarax$) and the rule ($\Vacstreg$)
  with the axioms ($\bvarax$) and the rule ($\Vacreg$) in Fig.\,\ref{fig:Reg:Lambdaprefixreg}, respectively.
  
  A finite or infinite derivation $\infDeriv$ in $\Lambdaprefixreginf$ (in $\Lambdaprefixstreginf$)
  is called \emph{closed} 
  if all terms in leafs of $\infDeriv$ are axioms.
  Derivability of a term $\flabs{\vec{\avar}}{\aiter}$ in $\Lambdaprefixreginf$ (in $\Lambdaprefixstreginf$),
  denoted symbolically by $\derivablein{\Lambdaprefixreginf}{\flabs{\vec{\avar}}{\aiter}}$ (resp.\ by $\derivablein{\Lambdaprefixstreginf}{\flabs{\vec{\avar}}{\aiter}}$), 
  means the existence of a closed derivation with conclusion $\flabs{\vec{\avar}}{\aiter}$. 
\end{definition}

We say that a proof system $\aproofsys$ is \emph{sound} (\emph{complete}) for a property $\saprop$ 
of infinite \lambdaterms\
if $\derivablein{\aproofsys}{\femptylabs{\aiter}}$ implies $\aprop{\aiter}$
(if $\aprop{\aiter}$ implies $\derivablein{\aproofsys}{\femptylabs{\aiter}}$)
for all infinite \lambdaterms~$\aiter\in\Ter{\inflambdacal}$. 


\begin{proposition}\label{prop:Lambdaprefixreg:Lambdaprefixstreg}
  $\Lambdaprefixreginf$ and $\Lambdaprefixstreginf$ are sound and complete 
  for 
  all infinite \lambdaterms.
  What is more, these systems are also sound and complete 
  for all prefixed infinite \lambdaterms\ among all terms over signature $\siglpcCRS$:
  for all terms $\diter\in\Ter{\siglpcCRS}$ it holds that
  $\derivablein{\Lambdaprefixreginf}{\diter}$
  if and only if
  $\diter\in\Ter{\inflambdaprefixcal}$
  (and hence $\diter$ is of the form $\flabs{\vec{\dvar}}{\diteracc}$).
\end{proposition}

\begin{proof}
  For completeness of these systems note
  that every prefixed term $\flabs{\vec{\bvar}}{\biter}\in\Ter{\inflambdaprefixcal}$ with $\biter$ not a variable
  is the conclusion of an instance of a rule in these systems. 
\end{proof}

Next we define proof systems for proving
regularity and strong regularity of infinite \lambdaterms\
by means of typically infinite derivations.

\begin{definition}[proof systems $\Reginf$, $\stReginf$]
  \label{def:Reginf:stReginf}
  The proof systems $\Reginf$ and $\stReginf$ 
  have the same axioms and rules as $\Lambdaprefixreginf$ and $\Lambdaprefixstreginf$, respectively,
  but they restrict the notion of derivability. 
  A derivation $\Deriv$ in $\Lambdaprefixreginf$ (in $\Lambdaprefixstreginf$)
  is called \emph{admissible in $\Reginf$} (\emph{in $\stReginf$})
  if it contains only finitely many different terms,
  and if 
  it is \emph{\Vacregeager} (\emph{\Vacstregeager}), that is, if
  no conclusion of an instance of ($\lappcomp$) or ($\labscomp$) in $\Deriv$
  is the source of a $\scompressregred$\nb-step (a $\scompressstregred$\nb-step). 
  Derivability in $\Reginf$ (in $\stReginf$)
  means the existence of a closed admissible derivation.
\end{definition}

For $\Reginf$ and $\stReginf$ we easily obtain the following soundness and completeness results. 
 
\begin{figure}[t]
\vspace*{-2ex}  
  
\begin{center}\framebox{\begin{minipage}{350pt}\begin{center}
  \mbox{}
  \\[1ex]
  \mbox{ 
    \AxiomC{$\phantom{\flabs{\vec{\avar}\bvar}{\bvar}}$}
    \RightLabel{\bvarax}
    \UnaryInfC{$\flabs{\vec{\avar}\bvar}{\bvar}$}
    \DisplayProof
  } 
  \hspace*{3.5ex}    
  \mbox{
    \AxiomC{$ \flabs{\vec{x}\bvar}{\aiteri{0}} $}
    \RightLabel{$\labscomp$}
    \UnaryInfC{$ \flabs{\vec{\avar}}{\labs{\bvar}{\aiteri{0}}} $}
    \DisplayProof
        }
  \hspace*{3.5ex}
  \mbox{
    \AxiomC{$ \flabs{\vec{\avar}}{\aiteri{0}}$}
    \AxiomC{$ \flabs{\vec{\avar}}{\aiteri{1}}$}
    \RightLabel{$\lappcomp$}
    \BinaryInfC{$ \flabs{\vec{\avar}}{\lapp{\aiteri{0}}{\aiteri{1}}} $}
    \DisplayProof 
    }     
   \\[2.5ex]
  \mbox{
    \AxiomC{$ \flabs{\avari{1}\ldots\avari{n-1}}{\aiter} $}
    \RightLabel{\Vacstreg\ $\;$
                \parbox[c]{80pt}{\small (if the binding\\$\slabs{\avari{n}}$ is vacuous)
                                 }%
                }
    \UnaryInfC{$ \flabs{\avari{1}\ldots\avari{n}}{\aiter} $}
    \DisplayProof
        }  
  \hspace*{-2ex}
  \mbox{
    \AxiomC{$ [\flabs{\vec{\avar}}{\aiter}]^{\amarker}$}
    \noLine
    \UnaryInfC{$\Derivi{0}$}
    \noLine
    \UnaryInfC{$ \flabs{\vec{\avar}}{\aiter} $}
    \RightLabel{$\sFIX,\amarker$ $\,$
                \parbox{49pt}{\small (if $\depth{\Derivi{0}} \ge 1$)}
                }                  
    \UnaryInfC{$ \flabs{\vec{\avar}}{\aiter} $}
    \DisplayProof
    }
\end{center}\end{minipage}}\end{center} 
  \vspace*{-2.5ex}  
  \caption{\label{fig:stReg:Lambdaprefixstreg:stRegzero}%
           The proof system $\stReg$ for strongly regular $\lambda$\nb-terms.
           In the variant system $\stRegzero$ of $\stReg$,
           instances of \FIX\ are subject to the additional side-condition: 
           for all $\flabs{\vec{y}}{\bter}$ on threads in $\Derivi{0}$ from
           open marked assumptions $(\flabs{\vec{\avar}}{\ater})^u$ downwards
           it holds that $\length{\vec{y}} \ge \length{\vec{x}}$. 
           The systems $\Lambdaprefixstreginf$ and $\stReginf$ do not contain the rule $\sFIX$.
           Derivations in $\stReg$, $\stRegzero$, and $\Reginf$ must be \Vacstregeager.
           }
\end{figure}

\begin{figure}[t]
\begin{center}  
  \framebox{
\begin{minipage}{330pt}
\begin{center}
  \mbox{}
  \\[0.75ex]
  \mbox{ 
    \AxiomC{$\phantom{\flabs{\bvar}{\bvar}}$}
    \RightLabel{\bvarax}
    \UnaryInfC{$\flabs{\bvar}{\bvar}$}
    \DisplayProof
        } 
  \hspace*{4.5ex}      
  \mbox{
    \AxiomC{$ \flabs{\avari{1}\ldots\avari{i-1}\avari{i+1}\ldots\avari{n}}{\aiter} $}
    \RightLabel{\Vacreg\ $\;$
                \parbox{65pt}{\small (if the binding\\[-0.5ex]\hspace*{\fill} $\lambda\avari{i}$ is vacuous)}
                }
    \UnaryInfC{$ \flabs{\avari{1}\ldots\avari{n}}{\aiter} $} 
    \DisplayProof
        }  
\end{center}
\end{minipage}
            }
\end{center} 
  \vspace*{-2.5ex}  
  \caption{\label{fig:Reg:Lambdaprefixreg}%
           The proof system $\Reg$ for regular $\lambda$\nb-terms 
           arises from $\stReg$ 
           through replacing the rule ($\Vacstreg$) by the rule ($\Vacreg$),
           and the axiom scheme ($\bvarax$) by the more restricted version here.
           The systems $\Lambdaprefixreginf$ and $\Reginf$ do not contain the rule ($\sFIX$). 
           Derivations in $\Reg$ and $\Reginf$ must be \Vacregeager.
           }
  \vspace*{-2ex}
\end{figure}

\begin{proposition}\label{prop:Reginf:stReginf}
  \begin{enumerate}[(i)]
    \item{}\label{prop:Reginf:stReginf:item:Reginf} 
      $\Reginf$ is sound and complete for regularity of infinite \lambdaterms.  
    \item{}\label{prop:Reginf:stReginf:item:stReginf}
      $\stReginf$ is sound and complete for strong regularity of infinite \lambdaterms.
  \end{enumerate}
\end{proposition}

\begin{proof}
  We argue only for (\ref{prop:Reginf:stReginf:item:stReginf}), since (\ref{prop:Reginf:stReginf:item:Reginf}) can be seen analogously.
  Every \Vacstregeager\ derivation $\infDeriv$ in $\Lambdaprefixstreginf$ with conclusion $\femptylabs{\aiter}$
  assembles the maximal $\sstregred$\nb-rewrite sequences from $\femptylabs{\aiter}$ in the following sense: 
  the steps of every such rewrite sequence
  correspond to the steps through $\infDeriv$ along a thread from the conclusion upwards.
  Therefore if $\infDeriv$ is an admissible derivation in $\stReginf$, and hence contains only finitely many terms,
  then $\gSTstreg{\aiter}$ is finite. Since every term $\femptylabs{\aiter}$ in $\Ter{\lambdaprefixcal}$ 
  has a \Vacstregeager\ derivation in $\Lambdaprefixstreginf$,
  the converse holds as well.
\end{proof}

Finally we introduce proof systems for proving regularity and strong regularity of infinite \lambdaterms\
by means of finite derivations. Derivations in these systems are able to detect the cyclic structure 
of a regular or strongly regular \lambdaterm, and correspondingly, cyclicity in 
$\sregred$- and $\sstregred$\nb-re\-write sequences that decompose the term.
These proof systems are reminiscent of coinductively motivated proof systems
such as the ones for recursive type equality and subtyping by Brandt and Henglein \cite{bran:heng:1998}
(proof-theoretic connections with more traditional proof systems have been studied in \cite{grab:2005a}).

\begin{definition}[proof systems $\Reg$, $\stReg$, and $\stRegzero$]
  The natural-deduction style proof system $\stReg$ 
  has the axioms and rules in Fig.\,\ref{fig:stReg:Lambdaprefixstreg:stRegzero}.
  Its variant $\stRegzero$ demands an additional side-condition on instances of the rule ($\sFIX$) as described there.
  The system $\Reg$ arises from $\stReg$ by dropping the rule ($\Vacstreg$),
  and restricting the axioms to the axioms ($\bvarax$) in Fig.\,\ref{fig:Reg:Lambdaprefixreg}.
  
  A derivation in one of these systems is called \emph{closed}
  if it does not contain any undischarged marker assumptions 
  (discharging assumptions is indicated by assigning the appertaining assumption markers
   to instances of $\sFIX$, see Fig.\,\ref{fig:stReg:Lambdaprefixstreg:stRegzero}).  
  Derivability in $\Reg$ (in $\stReg$ or in $\stRegzero$)
  means the existence of a closed, \Vacregeager\ (\Vacstregeager), finite derivation.
\end{definition}

The proposition below explains that the side-condition `$\depth{\Derivi{0}}\ge 1$' on subderivations of FIX-in\-stan\-ces
guarantees a `guardedness' property for threads in derivations in these systems. 
 
\begin{proposition}\label{prop:guardedness:Reg:stReg:stRegzero}
  Let $\Deriv$ be a derivation in $\Reg$, $\stReg$, or $\stRegzero$.
  Then for every instance $\ainst$ of the rule \FIX\ in $\Deriv$ it holds: 
  every thread from $\ainst$ upwards to a marked assumption that is discharged at $\ainst$
  passes at least one instance of a rule {\normalfont ($\labscomp$)} or {\normalfont ($\lappcomp$)}.
\end{proposition}

\begin{proof}
  Let $\Deriv$ be a derivation in $\Reg$, as the argument is analogous for $\stReg$ and $\stRegzero$.
  Let $\ainst$ be an instance of \FIX\ in $\Deriv$,
  and $\apath$ a thread from 
  the conclusion $\flabs{\vec{\bvar}}{\biter}$ 
                 of $\ainst$
  to a marked assumption $(\flabs{\vec{\bvar}}{\biter})^{\amarker}$ 
                         that is discharged at $\ainst$. 
  Then due to the side-condition on the topmost instance $\binst$ of \FIX\ passed on $\apath$ there is at least one
  instance of a rule ($\labscomp$), ($\lappcomp$), or ($\Vacreg$) passed on $\apath$
  above $\binst$. We are done unless that is an instance of ($\Vacreg$). But then there must\enlargethispage{3ex}
  also be an instance of ($\labscomp$) on $\apath$, since ($\Vacreg$) decreases the prefix length,
  only ($\labscomp$) increases it, and the prefix lengths in the formula at the start and at the end of $\apath$ are the same. 
\end{proof}

\pagebreak[4]
\begin{example}\label{ex:Reg:stReg}
\begin{enumerate}[(i)]
  \item{}\label{ex:Reg:stReg:item:stReg}
    The following are two derivations in $\stReg$ of different efficiency of
    the infinite \lambdaterm~$\ater$ from Fig.\,\ref{fig:ll:expressible} when
    represented by the recursive equation
    $\ater = \labs{\avar\bvar}{\lapp{\lapp{\ater}{\bvar}}{\avar}}\,$:
    \vspace{-1ex}
    \begin{equation*}
      \hspace*{-6.5ex}
      \begin{aligned}[c]
      \scalebox{0.9}{  
      \AxiomC{$ (\femptylabs{\ater})^{\amarker} $}
      \RightLabel{$\Vacstreg$}
      \UnaryInfC{$ \flabs{\avar}{\ater} $}
      \RightLabel{$\Vacstreg$}
      \UnaryInfC{$ \flabs{\avar\bvar}{\ater} $}
      \AxiomC{\mbox{}}
      \RightLabel{$\bvarax$}
      \UnaryInfC{$ \flabs{\avar\bvar}{\bvar} $}
      \RightLabel{$\lappcomp$}
      \BinaryInfC{$ \flabs{\avar\bvar}{\lapp{\ater}{\bvar}} $}
      \AxiomC{\mbox{}}
      \RightLabel{$\bvarax$}
      \UnaryInfC{$ \flabs{\avar}{\avar} $}
      \RightLabel{$\Vacstreg$}
      \UnaryInfC{$ \flabs{\avar\bvar}{\avar} $}
      \RightLabel{$\lappcomp$}
      \BinaryInfC{$ \flabs{\avar\bvar}{\lapp{\lapp{\ater}{\bvar}}{\avar}} $}
      \RightLabel{$\labscomp$}
      \UnaryInfC{$ \flabs{\avar}{\labs{\bvar}{\lapp{\lapp{\ater}{\bvar}}{\avar}}} $}
      \RightLabel{$\labscomp$}
      \UnaryInfC{$ \femptylabs{}{\labs{\avar\bvar}{\lapp{\lapp{\ater}{\bvar}}{\avar}}} $}
      \RightLabel{$\sFIX,\amarker$}
      \UnaryInfC{$ \femptylabs{}{\ater} $}
      \DisplayProof
        }
      \end{aligned}
      \hspace*{1ex}
      \begin{aligned}[c]
      \scalebox{0.9}{  
      \AxiomC{$ (\flabs{\avar}{\labs{\bvar}{\lapp{\lapp{\ater}{\bvar}}{\avar}}})^{\amarker} $}
      \RightLabel{$\labscomp$}
      \UnaryInfC{$ \femptylabs{\ater} $}
      \RightLabel{$\Vacstreg$}
      \UnaryInfC{$ \flabs{\avar}{\ater} $}
      \RightLabel{$\Vacstreg$}
      \UnaryInfC{$ \flabs{\avar\bvar}{\ater} $}
      \AxiomC{\mbox{}}
      \RightLabel{$\bvarax$}
      \UnaryInfC{$ \flabs{\avar\bvar}{\bvar} $}
      \RightLabel{$\lappcomp$}
      \BinaryInfC{$ \flabs{\avar\bvar}{\lapp{\ater}{\bvar}} $}
      \AxiomC{\mbox{}}
      \RightLabel{$\bvarax$}
      \UnaryInfC{$ \flabs{\avar}{\avar} $}
      \RightLabel{$\Vacstreg$}
      \UnaryInfC{$ \flabs{\avar\bvar}{\avar} $}
      \RightLabel{$\lappcomp$}
      \BinaryInfC{$ \flabs{\avar\bvar}{\lapp{\lapp{\ater}{\bvar}}{\avar}} $}
      \RightLabel{$\labscomp$}
      \UnaryInfC{$ \flabs{\avar}{\labs{\bvar}{\lapp{\lapp{\ater}{\bvar}}{\avar}}} $}
      \RightLabel{$\sFIX,\amarker$}
      \UnaryInfC{$ \flabs{\avar}{\labs{\bvar}{\lapp{\lapp{\ater}{\bvar}}{\avar}}} $}
      \RightLabel{$\labscomp$}
      \UnaryInfC{$ \femptylabs{}{\ater} $}
      \DisplayProof
             }
      \end{aligned}
    \end{equation*}
    Note that only the left derivation is one in $\stRegzero$, because
    the right one contains a term with shorter prefix than the discharged assumption on a thread to the instance of $\sFIX$. 
  \item{}\label{ex:Reg:stReg:item:Reg}
    The infinite \lambdaterm\ from Fig.\,\ref{fig:not:ll:expressible},
    denoted by the term $\femptylabs{\labs{\avar}{\funap{R}{\avar}}}$ 
    and generated by the \CRS\nb-rule $\funap{R}{\sametavar} \red \labs{\avar}{\lapp{\funap{R}{\avar}}{\sametavar}}$
    is derivable in $\Reg$ by the closed derivation on the left, but it is not derivable in $\stReg\,$:
\vspace{-2ex}
\begin{equation*} 
  \hspace*{-2ex}
  \begin{aligned}[c]
    \scalebox{0.9}{
      \AxiomC{$  (\overbrace{\flabs{\bvar}{\funap{R}{\bvar}}}^{{} = \flabs{\avar}{\funap{R}{\avar}}})^{\amarker}  $}
      \RightLabel{$\Vacreg$}
      \UnaryInfC{$ \flabs{\avar\bvar}{\funap{R}{\bvar}} $}
      \AxiomC{\mbox{}}
      \RightLabel{$\bvarax$}
      \UnaryInfC{$ \flabs{\avar}{\avar} $}
      \RightLabel{$\Vacreg$}
      \UnaryInfC{$ \flabs{\avar\bvar}{\avar}$}
      \RightLabel{$\lappcomp$}
      \BinaryInfC{$ \flabs{\avar\bvar}{\lapp{\funap{R}{\bvar}}{\avar}} $}
      \RightLabel{$\labscomp$}
      \UnaryInfC{$\flabs{\avar}{\labs{\bvar}{\lapp{\funap{R}{\bvar}}{\avar}}}$}
      \RightLabel{\sFIX, $\amarker$}
      \UnaryInfC{$ \flabs{\avar}{\funap{R}{\avar}} $}
      \RightLabel{$\labscomp$}
      \UnaryInfC{$ \femptylabs{\labs{\avar}{\funap{R}{\avar}}} $}
      \DisplayProof
          }
  \end{aligned}
    \hspace*{1ex}
  \begin{aligned}[c]
    \scalebox{0.9}{
      \AxiomC{$\vdots $}
      \noLine
      \UnaryInfC{$ \flabs{\avar\bvar\cvar\dvar\evar}{\lapp{\funap{R}{\evar}}{\dvar}} $}
      \RightLabel{$\labscomp$}
      \UnaryInfC{$ \flabs{\avar\bvar\cvar\dvar}{\labs{\evar}{\lapp{\funap{R}{\evar}}{\dvar}}} $}
      \AxiomC{\mbox{}}
      \RightLabel{$\bvarax$}
      \UnaryInfC{$ \flabs{\avar\bvar\cvar}{\cvar} $}
      \RightLabel{$\Vacstreg$}
      \UnaryInfC{$ \flabs{\avar\bvar\cvar\dvar}{\cvar}$}  
      \RightLabel{$\lappcomp$}
      \BinaryInfC{$ \flabs{\avar\bvar\cvar\dvar}{\lapp{\funap{R}{\dvar}}{\cvar}} $}
      \RightLabel{$\labscomp$}
      \UnaryInfC{$ \flabs{\avar\bvar\cvar}{\labs{\dvar}{\lapp{\funap{R}{\dvar}}{\cvar}}} $}
      \AxiomC{\mbox{}}
      \RightLabel{$\bvarax$}
      \UnaryInfC{$ \flabs{\avar\bvar}{\bvar} $}
      \RightLabel{$\Vacstreg$}
      \UnaryInfC{$ \flabs{\avar\bvar\cvar}{\bvar}$}  
      \RightLabel{$\lappcomp$}
      \BinaryInfC{$ \flabs{\avar\bvar\cvar}{\lapp{\funap{R}{\cvar}}{\bvar}} $}
      \RightLabel{$\labscomp$}
      \UnaryInfC{$ \flabs{\avar\bvar}{ \labs{\cvar}{\lapp{\funap{R}{\cvar}}{\bvar}} } $}
      \AxiomC{\mbox{}}
      \RightLabel{$\bvarax$}
      \UnaryInfC{$ \flabs{\avar}{\avar} $}
      \RightLabel{$\Vacstreg$}
      \UnaryInfC{$ \flabs{\avar\bvar}{\avar}$}
      \RightLabel{$\lappcomp$}
      \BinaryInfC{$ \flabs{\avar\bvar}{\lapp{\funap{R}{\bvar}}{\avar}} $}
      \RightLabel{$\labscomp$}
      \UnaryInfC{$ \flabs{\avar}{\labs{\bvar}{\lapp{\funap{R}{\bvar}}{\avar}}}$}
      \RightLabel{$\labscomp$}
      \UnaryInfC{$ \femptylabs{\labs{\avar}{\funap{R}{\avar}}}$}
      \DisplayProof
          }
  \end{aligned}
\end{equation*}
    The latter follows from the infinite prooftree on the right,
    the result of a bottom-up proof search in $\stReg$, which is a derivation in 
    $\Lambdaprefixstreginf$ but not in $\stReginf$, since, as it does not contain repetitions,
    the rule $\sFIX$ 
                     cannot be used to cut off repetitive subderivations.
  \end{enumerate}
\end{example}

Finally, we can link derivability in $\Reg$ and $\stReg$ to regularity and strong regularity.

\begin{theorem}
   \label{thm:Reg:stReg:stRegzero}
  \begin{enumerate}[(i)]
    \item{}\label{thm:Reg:stReg:stRegzero:item:Reg} 
      $\Reg$ is sound and complete for regularity of infinite \lambdaterms.
    \item\label{thm:Reg:stReg:stRegzero:item:stReg:stRegzero} 
      $\stReg$ and $\stRegzero$ are sound and complete for strong regularity of infinite \lambdaterms.
  \end{enumerate}
\end{theorem}


\begin{proof}
  For (\ref{thm:Reg:stReg:stRegzero:item:Reg}),
  in view of Proposition\,\ref{prop:Reginf:stReginf}, (\ref{prop:Reginf:stReginf:item:Reginf}), it suffices 
  to be able to transform closed, admissible derivations in $\Reginf$ into closed derivations in $\Reg$, and vice versa.
  Every closed derivation $\Deriv$ in $\Reg$ can be unfolded by a stepwise, typically infinite process into a closed derivation in $\Lambdaprefixreginf$:
  in every step the subderivation of a bottommost instance $\ainst$ of $\sFIX$
  is transferred to above each of the marked assumptions that are discharged at $\ainst$, and the original instance of $\sFIX$ is removed. 
  If this process is infinite, then due to Proposition\,\ref{prop:guardedness:Reg:stReg:stRegzero}
  it always eventually increases the size of
  the part of the derivation below the bottommost occurrences of $\sFIX$. 
  Hence in the limit it produces a closed, \Vacregeager\ derivation in $\Lambdaprefixreginf$
  that contains only finitely many terms (only those in $\Deriv$), and thus is admissible in $\Reginf\hspace*{-2pt}$.
  Conversely, every admissible, closed derivation $\infDeriv$ in $\Reginf$ can be `folded'
  into a finite closed derivation in $\Reg$ by introducing $\sFIX$\nb-instances to cut off the derivation
  above the upper occurrence of a repetition. This yields a finite derivation 
  since due to admissibility of $\infDeriv$ in $\Reginf$
  every sufficiently long thread contains a repetition, and then K\H{o}nig's Lemma can be applied. 
  
  For $\stReg$ in (\ref{thm:Reg:stReg:stRegzero:item:stReg:stRegzero}) it can be argued analogously, 
  using Proposition$\,$\ref{prop:Reginf:stReginf}, (\ref{prop:Reginf:stReginf:item:stReginf}),
  and unfolding/folding between closed derivations in $\stReg$ and closed, admissible derivations in $\stReginf$.
  Soundness of $\stRegzero$ follows from soundness of $\stReg$.
  For completeness of $\stRegzero$, note that every closed, admissible derivation $\infDeriv$ in $\stReginf$ can be `folded'
  into a closed derivation of $\stRegzero$ by using a stricter version of repetition of terms:
  distinct occurrences of a term $\flabs{\vec{\bvar}}{\biter}$ on a thread of a prooftree form such a repetition only
  if all formulas in between have an equally long or longer abstraction prefix. 
  Since $\infDeriv$ is admissible, on every infinite thread $\athread$ of $\infDeriv$ there must occur 
  such a stricter form of repetition, namely of a term with the shortest abstraction prefix among the terms
  that occur infinitely often on $\athread$. 
\end{proof}

\section{Binding--Capturing Chains}
\label{sec:chains}

In this section we develop a characterization 
of strongly regular infinite \lambdaterms\ 
through a property of their term structure,
concerning `\bindcaptchains' on positions of the term. 
While not needed for obtaining the result concerning \lambdaletrecexpressibility\
in Section~\ref{sec:express}, 
we think that this characterization is of independent interest.

\Bindcaptchains\ originate from the notion of `gripping' due to {M}elli\`es \cite{mell:96}, 
and from techniques concerning the notion of `holding' of redexes
developed by van~Oostrom \cite{oost:97}.
In \cite{endr:grab:klop:oost:2011} they have been used
to study $\alpha$\nb-conversion-avoi\-ding $\mu$\nb-unfolding. 
 
Technically, \bindcaptchains\ are alternations of two kinds of links
between positions of variable occurrences and \lambdaabstractions\ (called binders below) in a \lambdaterm: 
`binding links' from a \lambdaabstraction\ downward to the variable occurrences it binds,
and
`capturing links' from a variable occurrence upward to \lambdaabstractions\ 
that do not bind it, but are situated on the upward path to its binding \lambdaabstraction. 
We formalize these links by binding and capturing relations,
which are then used to define \bindcaptchains. 





\begin{definition}[binding, capturing]\label{def:bind:iscapturedby}
  Let $\aiter\in\Ter{\inflambdacal}$. 
  On the set $\Positions{\aiter}$ of positions of $\aiter$
  (for positions in iCRS-terms, see \cite{kete:simo:2011})
  we define two binary relations:
  the \emph{binding relation} $\sbinds$, and
  the \emph{capturing relation} $\siscapturedby$. 
  Since these relations are specific to the term $\aiter$,
  they could be denoted by $\sbindsiter{\aiter}$, and $\siscapturedbyiter{\aiter}$, respectively.
  However, we will generally avoid this subscript notation,
  assuming that the underlying infinite \lambdaterm\ will always be clear from the context. 
  For defining $\sbinds$ and $\siscapturedby$ for $\aiter$, let $\apos,\bpos\in\Positions{\aiter}$. 
  
  $\apos \binds \bpos$ 
  (in words: a binder (\lambdaabstraction) at $\apos$ \emph{binds} a variable occurrence at $\bpos$)
  holds
  if $\apos$ is a binder position, and $\bpos$ a variable position in $\aiter$,
     and the binder at position $\apos$ binds the variable occurrence at position $\bpos$.
  
  $\bpos \iscapturedby \apos$ 
  (in words: a variable occurrence at $\bpos$ \emph{is captured by} a binder at $\apos$),
  and conversely
  $\apos \captures \bpos$
  (the binder at $\apos$ \emph{captures} a variable occurrence at $\bpos$),
  hold
  if $\bpos$ is a variable position and $\apos < \bpos$ a binder position in $\aiter$,
      and there is no binder position $\bposi{0}$ in $\aiter$ with $\apos \le \bposi{0}$ and $\bposi{0} \binds \bpos$.
\end{definition}

\begin{definition}[\bindcaptchain]\label{def:bind:capt:chain}
  Let $\aiter\in\Ter{\inflambdacal}$. 
  A finite or infinite sequence $\enumsequence{\aposi{0},\bposi{1},\aposi{1},\bposi{2},\aposi{2},\ldots}$
  in $\setexp{0,1}^*$ is called a \emph{\bindcaptchain\ in $\aiter$}
  if $\aposi{0},\bposi{1},\aposi{1},\bposi{2},\aposi{2},\ldots\in\Positions{\aiter}$,  
  and if these positions are linked alternatingly via binding and capturing: 
  $\aposi{1} \binds \bposi{2} \iscapturedby \aposi{2} 
             \binds \bposi{3} \iscapturedby \aposi{3}
             \binds \ldots$,
  starting with a binding and ending with a capturing.
  The \emph{length} of such a \bindcaptchain\ is the number of `is captured by' links. 
\end{definition}

See Figs.~\ref{fig:ll:expressible} and \ref{fig:not:ll:expressible}
for illustrations of \bindcaptchains\ in terms we have encountered.
Note that
binding--capturing chains occur whenever scopes overlap, or in other words
when nesting \extscope{s} occur. Every \bindcaptchain\ is fully contained
within a \extscope.  

Now we introduce a \positionannotated\ variant $\stRegposCRS$ of $\stRegCRS$ in order to relate
\bindcaptchains\ to rewrite sequences in $\stRegARS$. 
The idea is that if a \lambdaterm~$\aiter$ has 
a ge\-ne\-ra\-ted subterm $\flabs{\bvari{1}\ldots\bvari{n}}{\biter}$ in $\stRegARS$,
then $\flabspos{\bvari{1}\ldots\bvari{n}}{\aposi{1},\ldots,\aposi{n}}{\bpos}{\biter}$ is a generated subterm in $\stRegposCRS$,
where $\aposi{1},\ldots,\aposi{n}$ are the positions in $\aiter$
from which the bindings $\slabs{\bvari{1}\ldots\bvari{n}}$ in the abstraction prefix descend,
and $\bpos$ is the position in $\aiter$ of the body $\biter$ of the generated subterm.

\begin{definition}[\iCRS-representation of $\inflambdaprefixposcal$, terms in $\Ter{\inflambdaprefixposcal}$]
  \label{def:sig:lambdaprefixposcal:CRS}
  The \CRS-sig\-na\-ture for $\inflambdaprefixposcal$, 
  the  \lambdacalculus\ with \posannotated\ abstraction prefixes is given~by:
  \begin{center}
    $\siglpposcCRS = \siglcCRS 
                       \cup 
                     \descsetexpnormalsize{ \sflabsposCRS{\aposi{1},\ldots,\aposi{n}}{\bpos} }{ n\in\nats,\, \aposi{1},\ldots,\aposi{n},\, \bpos\in\setexp{0,1}^* }$ 
  \end{center}                      
  where all of the function symbols $\sflabsposCRS{\aposi{1},\ldots,\aposi{n}}{\bpos}$ are unary.
  Terms that are of the specific form 
  $\flabsposCRS{\avari{1}}{\aposi{1},\ldots,\aposi{n}}{\bpos}{\ldots\absCRS{\avari{n}}{\ater}}$
  will be denoted in informal notation as  
  $\flabspos{\avari{1}\ldots\avari{n}}{\aposi{1},\ldots,\aposi{n}}{\bpos}{\ater}$, 
  which can be abbreviatee to $\flabspos{\vec{\avar}}{\vec{\apos}}{\bpos}{\ater}$, or to $\femptylabs\ater$ in case of an empty prefix.
  By $\Ter{\inflambdaprefixposcal}$ we denote the set of closed \iCRS-terms 
  over $\siglpposcCRS$ of the form
  $\flabsposCRS{\avari{1}}{\aposi{1},\ldots,\aposi{n}}{\bpos}{\ldots\absCRS{\avari{n}}{\ater}}$
  for some $n\in\nats$,
  $\aposi{1},\ldots,\aposi{n},\bpos\in\setexp{0,1}^*$
  and some term $\labsCRS{\avari{1}}{\ldots\labsCRS{\avari{n}}{\ater}}$ over the signature $\siglcCRS$ with
  the restriction that a CRS-abstraction in $\aiter$ can only occur as an argument of a
  function symbol $\slabsCRS$. 
\end{definition}

\begin{definition}[\posannotated\ variant {\normalfont $\stRegposCRS$}]\label{def:stRegposCRS}
  On $\Ter{\inflambdaprefixcal}$ we consider the following rewrite rules in informal notation:
\vspace{-1ex}
\begin{align*}
    (\rulepos{\slappdecompi{i}}):
      & & \hspace*{-1.5ex}
    \flabspos{\avari{1}\ldots\avari{n}}{\aposi{1},\ldots,\aposi{n}}{\bpos}{\lapp{\aiteri{0}}{\aiteri{1}}}
      & {} \red
    \flabspos{\avari{1}\ldots\avari{n}}{\aposi{1},\ldots,\aposi{n}}{\bpos i}{\aiteri{i}}
      & & \hspace*{-1.5ex} \text{(for each $i\in\setexp{0,1}$)}
    \\
    (\rulepos{\slabsdecomp}):
      & & \hspace*{-1.5ex}
    \flabspos{\avari{1}\ldots\avari{n}}{\aposi{1},\ldots,\aposi{n}}{\bpos}{\labs{\bvar}{\aiteri{0}}}
      & {} \red
    \flabspos{\avari{1}\ldots\avari{n}\bvar}{\aposi{1},\ldots,\aposi{n},\bpos}{\bpos 00}{\aiteri{0}}
    \displaybreak[0]\\
    (\rulepos{\snlvarsucc}):
      & & \hspace*{-1.5ex}
    \flabspos{\avari{1}\ldots\avari{n+1}}{\aposi{1},\ldots,\aposi{n+1}}{\bpos}{\aiteri{0}}
      & {} \red
    \flabspos{\avari{1}\ldots\avari{n}}{\aposi{1},\ldots,\aposi{n}}{\bpos}{\aiteri{0}}
      & & \hspace*{-1.5ex}
    \hspace*{0ex} (\text{if bind.\ $\slabs{\avari{n+1}}$ is vacuous})
\end{align*}
  The change of the term-body position in a $\lambda$\nb-decomposition step is
  motivated by the underlying \CRS\nb-formalization of terms in $\inflambdaprefixcal$: when a
  subexpression $\labsCRS{\bvar}{\aiteri{0}}$ of a term in $\Ter{\inflambdaprefixposcal}$ 
  that represents a $\lambda$\nb-abstraction starts at position $\bpos$,
  its binding is declared at position $\bpos 0$, and its  body $\aiteri{0}$ starts at position $\bpos 00$.

  By $\stRegposARS$ we denote the
  abstract rewriting systems induced,
  similar to the definition of $\stRegARS$ in Def.~\ref{def:RegCRS:stRegCRS} earlier,
  by the rules above on \posannotated\ terms in $\Ter{\inflambdaprefixcal}$.
  
  Also analogously to Def.~\ref{def:RegCRS:stRegCRS},
  by $\stregred$ we denote the \emph{$\rulepos{\snlvarsucc}$}-eager 
  rewrite strategy for $\stRegposARS$.
%
\end{definition}

The lemma below gathers basic properties of the rewrite relation $\stregred$
on \posannotated\ prefixed \lambdaterms, and statements about
the form of possible $\stregred$\nb-rewrite sequences.

\begin{lemma}\label{lem:def:stRegposCRS}
  \begin{enumerate}[(i)]
    \item{}\label{lem:def:stRegposCRS:item:i}
      If
      $\,
       \femptylabspos{\tuple{\niks}}{\rootpos}{\aiter}
         \stregmred
       \flabspos{\vec{\avar}}{\vec{\apos}}{\bpos}{\biter}  
       \,$,
      then there is $n\in\nats$
      such that $\vec{\avar} = \tuple{\avari{1}\ldots\avari{n}}$,
                $\vec{\apos} = \tuple{\aposi{1},\ldots,\aposi{n}}$,
      $\aposi{1},\ldots,\aposi{n},\bpos\in\Positions{\aiter}$, and    
      $\aposi{1} < \aposi{2} < \ldots < \aposi{n} < \bpos$.
      
    \item{}\label{lem:def:stRegposCRS:item:ii}
%
      If
      $\flabspos{\bvari{1}\ldots\bvari{n}}{\bposi{1},\ldots,\bposi{n}}{\bpos}{\biteri{1}}
         \stregconvmred
       \flabspos{\avari{1}\ldots\avari{n}}{\aposi{1},\ldots,\aposi{n}}{\bposi{0}}{\biter}
         \stregmred
       \flabspos{\cvari{1}\ldots\cvari{m}}{\cposi{1},\ldots,\cposi{m}}{\bpos}{\biteri{2}}$
      holds for some $\bpos\in\setexp{0,1}^*$,
      then 
      $\flabspos{\bvari{1}\ldots\bvari{n}}{\bposi{1},\ldots,\bposi{n}}{\bpos}{\biteri{1}}
         =  
       \flabspos{\cvari{1}\ldots\cvari{m}}{\cposi{1},\ldots,\cposi{m}}{\bpos}{\biteri{2}}$ 
      follows, and hence also
      $n = m$, and
      $\bposi{1} = \cposi{1}$, \ldots, $\bposi{n} = \cposi{m}$.

    \item{}\label{lem:def:stRegposCRS:item:iii}
     If
      $
       \arewseq \funin 
       \flabspos{\avari{1}\ldots\avari{n}}{\aposi{1},\ldots,\aposi{n}}{\bpos}{\biter}
         \stregmred
       \flabspos{\avari{1}\ldots\avari{n'}}{\aposacci{1},\ldots,\aposacci{n'}}{\bposacc}{\biter'}  
       $
      is a rewrite sequence in $\stRegposCRS$, 
      and 
      $n_0 = \max \descsetexp{i}{1\le i\le n,\,\text{$\aposi{1} = \aposacci{1}$, \ldots, $\aposi{i} = \aposacci{i}$}}
           (\le \min \setexp{n,n'})$
      then $\arewseq$ is of the form:
      \begin{center}
      $\,
       \arewseq \funin \;
       \flabspos{\avari{1}\ldots\avari{n}}{\aposi{1},\ldots,\aposi{n}}{\bpos}{\biter}
         \stregmred
       \flabspos{\avari{1}\ldots\avari{n_0}}{\aposi{1},\ldots,\aposi{{n_0}}}{\bposi{0}}{\citer}
         \stregmred
       \flabspos{\avari{1}\ldots\avari{n'}}{\aposacci{1},\ldots,\aposacci{n'}}{\bposacc}{\biter'}  
       \,$
      \end{center}
      for
      some $\bposi{0}$ such that
      $\bpos \le \bposi{0} < \aposacci{{n_0}+1} < \ldots < \aposacci{n'}\,$.  
      
    \item{}\label{lem:def:stRegposCRS:item:iv}
      For every rewrite sequence
      $\,
       \arewseq \funin \;
       \flabspos{\avari{1}\ldots\avari{n}}{\aposi{1},\ldots,\aposi{n}}{\bpos}{\biter}
         \stregmred
       \flabspos{\avari{1}\ldots\avari{n+m+1}}{\aposi{1},\ldots,\aposi{n+m+1}}{\bposacc}{\citer}  
       \,$
      is of the form:
      \begin{align*}
        \arewseq \funin \;
          \flabspos{\avari{1}\ldots\avari{n}}{\aposi{1},\ldots,\aposi{n}}{\bpos}{\biter}
            & \stregmred
          \flabspos{\avari{1}\ldots\avari{n+m}}{\aposi{1},\ldots,\aposi{n+m}}{\aposi{n+m+1}}{\labs{\avari{n+m+1}}{\citeri{0}}}
          \displaybreak[0]\\
            & \stregred
          \flabspos{\avari{1}\ldots\avari{n+m+1}}{\aposi{1},\ldots,\aposi{n+m+1}}{\aposi{n+m+1}00}{\citeri{0}} 
          \displaybreak[0]\\
            & \stregmred 
          \flabspos{\avari{1}\ldots\avari{n+m+1}}{\aposi{1},\ldots,\aposi{n+m+1}}{\bposacc}{\citer}    
      \end{align*}
      where $\avari{n+m}$ occurs free in $\labs{\avari{n+m+1}}{\citeri{0}}$, and hence also in $\citeri{0}$.
  \end{enumerate}
\end{lemma}

\begin{proof}
  Each of the statements (\ref{lem:def:stRegposCRS:item:i}), (\ref{lem:def:stRegposCRS:item:iii}), and (\ref{lem:def:stRegposCRS:item:iv})
  can be shown by induction on the length of the $\sstregred$\nb-rewrite sequence in the assumption,
  distinguishing the cases of the rule applications in the last step.
  For the part of statement (\ref{lem:def:stRegposCRS:item:i}),
  that a rewrite sequence
      $\,
       \femptylabspos{\tuple{\niks}}{\rootpos}{\aiter}
         \stregmred
       \flabspos{\vec{\avar}}{\vec{\apos}}{\bpos}{\biter}  
       \,$
  implies
      $\aposi{1},\ldots,\aposi{n},\bpos\in\Positions{\aiter}$,
  it suffices to show that $\bpos\in\Positions{\aiter}$,
  because, due to the form of the rules in $\stRegposCRS$, 
  a position is added to the list in the subscript of the abstraction prefix
  (by an application of the rule $\rulepos{\slabsdecomp}$)
  only if it has already been encountered as the position in the subscript of an abstraction prefix. 
  This remaining statement can be established as follows:
  every $\stregred$\nb-re\-write sequence from $\femptylabspos{\tuple{\niks}}{\rootpos}{\aiter}$
  explores $\aiter$ (or a pre-term representation of $\aiter$) from the root position downwards, 
  thereby keeping track of the current position in the superscript $\bpos$ of the abstraction prefix. 
  
  Statement (\ref{lem:def:stRegposCRS:item:ii}) can be shown similarly by induction on
  the sum (or the minimum) of the lengths of the two $\sstregred$\nb-rewrite sequence in the assumption.
\end{proof}

The proposition below formulates the statement that
$\stregred$\nb-rewrite sequences on terms in $\Ter\inflambdaprefixcal$
are related to $\stregred$\nb-rewrite sequences on \posannotated\ terms
via lifting (adding annotations) and projecting (dropping annotations).

\begin{proposition}\label{prop:position:lifting:projecting}
  \begin{enumerate}[(i)]
    \item{}\label{prop:position:lifting:projecting:item:lifting}  
      \emph{Lifting}:
      Every rewrite sequence 
        $
        \arewseq \;\funin\;\;
          \flabs{\vec{\avar}_0}{\aiteri{0}}
            \stregred
          \flabs{\vec{\avar}_1}{\aiteri{1}}
            \stregred
          \ldots
            \stregred
          \flabs{\vec{\avar}_n}{\aiteri{n}}
          $
      in $\stRegARS$
      can be lifted, 
      by adding given $\bposi{0}\in\positions$ and $\vec{\apos}_0\in\vecpositions$
      with $\length{\vec{\apos}_0} = \length{\vec{\avar}_0}$,
      and appropriate further position annotations
      $\bposi{1}, \ldots, \bposi{n}\in\positions$
      and $\vec{\apos}_1, \ldots, \vec{\apos}_n\in\vec{\positions}$,
      to the terms of $\arewseq$, yielding a rewrite sequence 
        $
        \arewseq^{\spos} \funin
          \flabspos{\vec{\avar}_0}{\vec{\apos}_0}{\bpos_0}{\aiteri{0}}
            \stregred
          \flabspos{\vec{\avar}_1}{\vec{\apos}_1}{\bpos_1}{\aiteri{1}}
            \stregred
          \ldots
            \stregred
          \flabspos{\vec{\avar}_n}{\vec{\apos}_n}{\bpos_n}{\aiteri{n}}
          $
      in $\stRegposARS$.
    \item{}\label{prop:position:lifting:projecting:item:projecting}    
      \emph{Projection}:
      The result of dropping the position annotations in the prefix 
      in a rewrite sequence in $\stRegposARS$ is a rewrite sequence in $\stRegARS$.
  \end{enumerate}      
\end{proposition}

\begin{proof}
  Statements~(i) and~(ii)
  can be shown by straightforward induction 
  on the length of rewrite sequences in $\stRegARS$ and in $\stRegposARS$, respectively. 
\end{proof}

\begin{figure}[t]
\vspace*{-2ex}  
  
\begin{center}\framebox{\begin{minipage}{375pt}\begin{center}
  \mbox{}
  \\[1ex]
  \mbox{ 
    \AxiomC{$\phantom{\flabspos{\avari{1}\ldots,\avari{n}\bvar}{\aposi{1},\ldots,\aposi{n},\cpos}{\bpos}{\:\bvar}}$}
    \RightLabel{\bvarax}
    \UnaryInfC{$\flabspos{\avari{1}\ldots\avari{n}\bvar}{\aposi{1},\ldots,\aposi{n},\cpos}{\bpos}{\:\bvar}$}
    \DisplayProof
  } 
  \hspace*{2.5ex} 
  \mbox{
    \AxiomC{$ \flabspos{\avari{1}\ldots\avari{n-1}}{\aposi{1},\ldots,\aposi{n-1}}{\bpos}{\aiter} $}
    \RightLabel{\Vacstreg\ $\;$
                \parbox[c]{80pt}{\small (if the binding\\$\slabs{\avari{n}}$ is vacuous)
                                 }%
                }
    \UnaryInfC{$ \flabspos{\avari{1}\ldots\avari{n}}{\aposi{1},\ldots,\aposi{n}}{\bpos}{\aiter} $}
    \DisplayProof
        }      
   \\[3ex]   
  \mbox{
    \AxiomC{$ \flabspos{\avari{1}\ldots\avari{n}\bvar}{\aposi{1},\ldots,\aposi{n},\bpos}{\bpos 00}{\aiteri{0}} $}
    \RightLabel{$\labscomp$}
    \UnaryInfC{$ \flabspos{\avari{1}\ldots\avari{n}}{\aposi{1},\ldots,\aposi{n}}{\bpos}{\labs{\bvar}{\aiteri{0}}} $}
    \DisplayProof
        }
  \hspace*{5.5ex}
  \mbox{
    \AxiomC{$ \flabspos{\vec{\avar}}{\vec{\apos}}{\bpos 0}{\:\aiteri{0}}$}
    \AxiomC{$ \flabspos{\vec{\avar}}{\vec{\apos}}{\bpos 1}{\:\aiteri{1}}$}
    \RightLabel{$\lappcomp$}
    \BinaryInfC{$ \flabspos{\vec{\avar}}{\vec{\apos}}{\bpos}{\:\lapp{\aiteri{0}}{\aiteri{1}}} $}
    \DisplayProof 
    }
  \\[0.5ex]\mbox{} 
\end{center}\end{minipage}}\end{center} 
  \vspace*{-2.5ex}  
  \caption{\label{fig:Lambdaprefixstregposinf}%
           The proof system $\Lambdaprefixstregposinf$ 
           }
\end{figure}


As a consequence of \Vacstregeager ness of $\stregred$,
and the position-change recorded in steps of $\stregred$,
reduction graphs with respect to $\stregred$ in $\stRegposCRS$ have the property 
to be trees. As such they can be captured directly by prooftrees (with certain properties) in a proof system $\Lambdaprefixstregposinf$
that is defined below.
Recall that $\sstregred\,$-re\-duc\-tion graphs in $\stRegCRS$ do not have tree shape in general,
but that their tree unfoldings can be captured by completed derivations in the proof system $\stReginf$ from Section~\ref{sec:proofs}.

\begin{definition}[proof system $\Lambdaprefixstregposinf$]
  \label{def:Lambdaprefixstregposinf}
  The proof system $\Lambdaprefixstregposinf$ 
  acts on \CRS-terms over the signature $\siglpposcCRS$
  as formulas, is a Hilbert-style system 
  for finite or infinite prooftrees (of depth $\le\omega$),
  and it has the axioms and rules displayed in Fig.~\ref{fig:Lambdaprefixstregposinf}.  
  
  A derivation $\infDeriv$ in $\Lambdaprefixstreginf$ is called \emph{\Vacstregeager}
  if no conclusion of an instance of ($\lappcomp$) or ($\labscomp$) in $\Deriv$
  is the source of a $\scompressstregred$\nb-step. 
  A finite or infinite derivation $\infDeriv$ in $\Lambdaprefixstregposinf$ 
  is called \emph{closed} 
  if all terms in leafs of $\infDeriv$ are axioms.
  By derivability of a term $\flabspos{\vec{\avar}}{\vec{\apos}}{\bpos}{\,\aiter}$ in $\Lambdaprefixstregposinf$,
  which is denoted symbolically by $\derivablein{\Lambdaprefixstregposinf}{\flabspos{\vec{\avar}}{\vec{\apos}}{\bpos}{\,\aiter}}$, 
  we mean the existence of a closed derivation with conclusion $\flabspos{\vec{\avar}}{\vec{\apos}}{\bpos}{\,\aiter}$.
\end{definition}

\begin{proposition}
  Let \lambdaterm~$\aiter\in\Ter{\inflambdacal}$ be an infinite \lambdaterm.
  Then there is a unique \Vacstregeager, and closed derivation $\infDeriv$ with conclusion $\femptylabs{\aiter}$ in \Lambdaprefixstregposinf.
  $\infDeriv$ corresponds directly to the $\stregred$\nb-deri\-va\-tion graph of $\femptylabs{\aiter}$ in $\stRegposCRS$,
  which has the form of a tree.
\end{proposition}

\begin{proposition}\label{prop:Lambdaprefixstregposinf}
  The proof system $\Lambdaprefixstregposinf$ is sound and complete 
  for 
  all po\-si\-tion-anno\-ta\-ted prefixed infinite \lambdaterms\ in $\Ter{\inflambdaprefixposcal}$,
  that is: 
  for all terms $\diter\in\Ter{\siglpposcCRS}$ it holds that
  $\derivablein{\Lambdaprefixstregposinf}{\diter}$
  if and only if
  $\diter\in\Ter{\inflambdaprefixposcal}$
  (thus $\diter$ is of the form
   $\flabspos{\avari{1}\ldots\avari{n}}{\aposi{1},\ldots,\aposi{n}}{\bpos}{\diteracc}$
   for $n\in\nats$, $\aposi{1},\ldots,\aposi{n},\bpos\in\setexp{0,1}^*$, 
   and $\flabs{\avari{1}\ldots\avari{n}}{\diteracc}\in\Ter{\inflambdaprefixcal}$).
\end{proposition}

The following proposition relates positions $\bpos$ in an infinite \lambdaterm\ $\aiter$
with $\stregred$\nb-re\-write sequences from $\femptylabspos{\tuple{\niks}}{\rootpos}{\aiter}$ in $\stRegposCRS$
that `access' $\bpos$ in $\aiter$, and that in doing so eventually produce
the `\posannotated\ generated subterm of $\aiter$ at position $\bpos$'.

\begin{proposition}\label{prop:positions:vs:stregred:rewseqs}
  Let $\aiter\in\Ter{\inflambdacal}$ be an infinite \lambdaterm. 
  Then for all $\bpos\in\setexp{0,1}^*$ it holds:
  \begin{enumerate}[(i)]
    \item{}\label{prop:positions:vs:stregred:rewseqs:item:i}
      If $\bpos\in\Positions{\aiter}$
      then there is a unique $\stregred$\nb-rewrite sequence in $\stRegposCRS$ of the form
      $\femptylabspos{\tuple{\niks}}{\rootpos}{\aiter}
         \stregmred
       \flabspos{\avari{1}\ldots\avari{n}}{\aposi{1},\ldots,\aposi{n}}{\bpos}{\biter}$
      that proceeds via terms of the form 
      $\flabspos{\vec{\avar}_i}{\vec{\apos}_i}{\bposi{i}}{\biteri{i}}$
      where the $\bposi{i}$ are contained in, and exhaust, the set $\descsetexp{\bposacc}{\bposacc\le\bpos}$
      such that furthermore:
      \begin{enumerate}[(a)]
        \item{}\label{prop:positions:vs:stregred:rewseqs:item:i:subitem:a}
          $\biter$ corresponds to the remaining body of the \lambdaterm~$\aiter$ at and below $\bpos\,$,
          with the variables free in $\biter$ being bound in the abstraction prefix of 
          $\flabspos{\avari{1}\ldots\avari{n}}{\aposi{1},\ldots,\aposi{n}}{\bpos}{\biter}$.
        \item{}\label{prop:positions:vs:stregred:rewseqs:item:i:subitem:b}
          the \posannotated\ generated subterm 
          $\flabspos{\avari{1}\ldots\avari{n}}{\aposi{1},\ldots,\aposi{n}}{\bpos}{\biter}$
          of $\aiter$ contains 
          the information on at which \lambdabinding\ positions of $\aiter$
          the free variables $\avari{1}, \ldots, \avari{n}$ of $\biter$  
          have been bound originally in $\aiter$,
          namely: 
          a free occurrence of $\avari{i}$ in $\biter$, where $1\le i\le n$, 
          descends from a variable position below $\bpos$ in $\aiter$ that is bound by a \lambdabinding\ at position $\aposi{i}$ above $\bpos$ in $\aiter$. 
      \end{enumerate}
    \item{}\label{prop:positions:vs:stregred:rewseqs:item:ii}   
      If $\femptylabspos{\tuple{\niks}}{\rootpos}{\aiter}
         \stregmred
       \flabspos{\avari{1}\ldots\avari{n}}{\aposi{1},\ldots,\aposi{n}}{\bpos}{\biter}$ holds,
      then $\bpos\in\Positions{\aiter}$ follows
      (and hence further statements described in item~(\ref{prop:positions:vs:stregred:rewseqs:item:i}) hold as well). 
  \end{enumerate} 
\end{proposition}

\begin{proof}[Proof (Hint)]
  The two items of the proposition can be established
  by induction on the length of $\stregred$\nb-re\-write sequences in $\stRegposCRS$,
  and by induction on the length of positions, respectively. 
  Thoroughly formal proofs of the statements here have to be based
  on the definition of \iCRS\nb-terms as $\alpha$\nb-equivalence classes 
  of \iCRS\nb-preterms and the definition of positions in \iCRS\nb-preterms.
\end{proof}


As an easy consequence, we obtain the following proposition.

\begin{proposition}\label{prop:stRegposCRS}
  For all $\aiter\in\Ter{\inflambdacal}$ and positions $\bpos\in\Positions{\aiter}$
  it holds:
  \begin{enumerate}[(i)]
    \item{}\label{prop:stRegposCRS:item:labs}
      if $\bpos$ is the position of an abstraction in $\aiter$,
      then there is a rewrite sequence in $\stRegposCRS$ of the form
      $\femptylabspos{\tuple{\niks}}{\rootpos}{\aiter}
         \stregmred
       \flabspos{\avari{1}\ldots\avari{n}}{\aposi{1},\ldots,\aposi{n}}{\bpos}{\labs{\avari{n+1}}{\biter}}$
      for some $n\ge 0$ and $\aposi{1},\ldots,\aposi{n}\in\Positions{\aiter}$.
   \item{}\label{prop:stRegposCRS:item:lapp}
     if $\bpos$ is the position of an application in $\aiter$,
     then there is a rewrite sequence in $\stRegposCRS$ of the form
     $\femptylabspos{\tuple{\niks}}{\rootpos}{\aiter}
        \stregmred
      \flabspos{\avari{1}\ldots\avari{n}}{\aposi{1},\ldots,\aposi{n}}{\bpos}{\lapp{\biteri{0}}{\biteri{1}}}$
     for some $n\ge 0$ and $\aposi{1},\ldots,\aposi{n}\in\Positions{\aiter}$.  
   \item{}\label{prop:stRegposCRS:item:var}
     if $\bpos$ is a variable position in $\aiter$,
      then there is a rewrite sequence in $\stRegposCRS$ of the form
      $\femptylabspos{\tuple{\niks}}{\rootpos}{\aiter}
         \stregmred
       \flabspos{\avari{1}\ldots\avari{n}}{\aposi{1},\ldots,\aposi{n}}{\bpos}{\avari{n}}  
       $  
      for some $n\ge 1$ and $\aposi{1},\ldots,\aposi{n}\in\Positions{\aiter}$. 
  \end{enumerate}
\end{proposition}


The next proposition describes the connection 
between the concepts of binding and capturing with
\posannotated\ $\stregred$\nb-re\-write sequences. 

\begin{proposition}\label{prop:bind:iscapturedby}
  For all $\aiter\in\Ter{\inflambdacal}$ and 
                                                       $\apos,\bpos\in\setexp{0,1}^*$ it holds:
  \vspace{-1,5ex}
  \begin{align*}
    \apos \binds \bpos
       \;\;&\Longleftrightarrow\;\;
    \text{there is a rewrite sequence  
          $\femptylabspos{\tuple{\niks}}{\rootpos}{\aiter}
             \stregmred
           \flabspos{\avari{1}\ldots\avari{n}}{\aposi{1},\ldots,\aposi{n}}{\bpos}{\avari{n}}$
          with $\apos  = \aposi{n}$}
%
%
    \\
    \apos \captures \bpos
       \;\;&\Longleftrightarrow\;\;
    \parbox[t]{310pt}{there is a rewrite sequence\\
                      $\femptylabspos{\tuple{\niks}}{\rootpos}{\aiter} 
                         \stregmred 
                       \flabspos{\avari{1}\ldots\avari{i}\ldots\avari{n}}{\aposi{1},\ldots,\aposi{i},\ldots,\aposi{n}}{\aposi{n}00}{\biter}
                         \stregmred  
                       \flabspos{\avari{1}\ldots\avari{i}}{\aposi{1},\ldots,\aposi{i}}{\bpos}{\avari{i}}$\\[0.35ex]
                      such that 
                                $i < n$, and $\apos = \aposi{n}$}
  \end{align*}
\end{proposition}

\begin{proof}[Proof (Hint)]
  The two statements of this proposition can be established along
  the statements of Proposition~\ref{prop:positions:vs:stregred:rewseqs}.
\end{proof}

\begin{lemma}\label{lem:def:stRegposCRS:item:v}
      If for some infinite \lambdaterm~$\aiter$ we have
      $\femptylabspos{\tuple{\niks}}{\rootpos}{\aiter}
         \stregmred
       \flabspos{\avari{1}\ldots\avari{n}}{\aposi{1},\ldots,\aposi{n}}{\bpos}{\biter}$ 
      with $\avari{1}, \ldots, \avari{n}$ distinct
      and 
      $i\in\setexp{1,\ldots,n}$ such that $\avari{i}$ occurs free in the body $\biter$,
      then there exists $\bpos'\in\setexp{0,1}^*$ such that 
      $\,
       \flabspos{\avari{1}\ldots\avari{n}}{\aposi{1},\ldots,\aposi{n}}{\bpos}{\biter}
         \stregmred
       \flabspos{\avari{1}\ldots\avari{i}}{\aposi{1},\ldots,\aposi{i}}{\bpos'}{\avari{i}}  
       \,$
      and $\bpos' > \bpos$. 
\end{lemma}

\begin{proof}[Proof (Hint)]
  The statement of the lemma can again be proved along
  the statement of Proposition~\ref{prop:positions:vs:stregred:rewseqs}. 
\end{proof}


The lemma below describes the connection between 
\bindcaptchains\ and \posannotated\ $\stregred$\nb-re\-write sequences.

\begin{lemma}[\bindcaptchains]\label{lem:bind:capt:chains:stRegpos}
  For all $\aiter\in\iTer{\lambdacal}$ it holds:
  \vspace*{-1.5ex}
  \begin{enumerate}[(i)]
    \item{}\label{lem:bind:capt:chains:stRegpos:item:i}
      If $\femptylabspos{\niks}{\rootpos}{\aiter}
            \stregmred
          \flabspos{\avari{1}\ldots\avari{n}\!}{\aposi{1},\ldots,\aposi{n}}{\bpos}{\!\bter}$,
      then $\aposi{1},\ldots,\aposi{n}\in\Positions{\aiter}$, 
      and if $n\ge 2$,      
      there are $\bposi{2},\ldots,\bposi{n}\in\Positions{\ater}$ such that
      $\aposi{1} \binds \bposi{2} \iscapturedby \aposi{2} \binds \ldots \binds \bposi{n} \iscapturedby \aposi{n}$.
      %
    \item{}\label{lem:bind:capt:chains:stRegpos:item:ii}
      If $\aposi{1} \binds \bposi{2} \iscapturedby \aposi{2} \binds \ldots \binds \bposi{n} \iscapturedby \aposi{n}$
      is a \bindcaptchain\ in $\aiter$,
      then there exist positions $\cposi{1},\ldots,\cposi{m}\in\Positions{\ater}$ with $m\ge n$ such that
      $\femptylabspos{\tuple{\niks}}{\rootpos}{\ater}
         \stregmred
       \flabspos{\avari{1}\ldots\avari{m}}{\cposi{1},\ldots,\cposi{m}}{\cposi{m}00}{\biter}$,
      $\aposi{1},\ldots,\aposi{n} \in \setexp{\cposi{1},\ldots,\cposi{m}}$,
      and $\aposi{1} < \aposi{2} < \ldots < \aposi{n} = \cposi{m} $.
  \end{enumerate}
\end{lemma}

\begin{proof}
  We first prove statement~(\ref{lem:bind:capt:chains:stRegpos:item:i}),
  by induction on $n\in\nats$.
  For $n=0$ nothing has to be shown.
  In case of $n=1$, for a given rewrite sequence
  $\femptylabspos{\niks}{\rootpos}{\aiter}
     \stregmred
   \flabspos{\avari{1}}{\aposi{1}}{\bpos}{\!\biter}$
  it has to be shown that $\aposi{1}\in\Positions{\aiter}$.
  This follows from Lemma~\ref{lem:def:stRegposCRS},~(\ref{lem:def:stRegposCRS:item:i}).
  For the induction step from $n$ to $n+1$ we let $n\ge 1$, and assume a rewrite sequence of the form
  $\arewseq \funin\;
          \femptylabspos{\niks}{\rootpos}{\aiter}
            \stregmred
          \flabspos{\avari{1}\ldots\avari{n+1}\!}{\aposi{1},\ldots,\aposi{n+1}}{\bpos}{\!\biter}$.
  By Lemma~\ref{lem:def:stRegposCRS},~(\ref{lem:def:stRegposCRS:item:iv}),
  $\arewseq$ is of the form:
  \begin{equation}\label{eq1:prf:lem:def:stRegposCRS}
  \left.  
  \begin{aligned}
    \arewseq \funin \;
      \femptylabspos{\niks}{\rootpos}{\ater}
        & \stregmred
      \flabspos{\avari{1}\ldots\avari{n}}{\aposi{1},\ldots,\aposi{n}}{\aposi{n+1}}{\labs{\avari{n+1}}{\citer}}
      \\
        & \stregred
      \flabspos{\avari{1}\ldots\avari{n}\avari{n+1}}{\aposi{1},\ldots,\aposi{n},\aposi{n+1}}{\aposi{n+1}00}{\citer} 
      \\
        & \stregmred 
      \flabspos{\avari{1}\ldots\avari{n}\avari{n+1}}{\aposi{1},\ldots,\aposi{n},\aposi{n+1}}{\bpos}{\biter}    
   \end{aligned}
   \qquad\right\}  
   \end{equation}
   By applying the induction hypothesis to the initial segment of $\arewseq$
   formed by the rewrite steps in the first line of \eqref{eq1:prf:lem:def:stRegposCRS},
   it follows that $\aposi{1}\ldots\aposi{n}\in\Positions{\aiter}$, and 
   that there exists a \bindcaptchain\
   $\aposi{1} \binds \bposi{2} \iscapturedby \aposi{2} \binds \ldots \binds \bposi{n} \iscapturedby \aposi{n}$ in $\aiter$,
   for some positions $\bposi{2},\ldots,\bposi{n}\in\Positions{\ater}$
   (if $n=1$, this \bindcaptchain\ has 
                                       length $0$).  
   Now note that $\avari{n}$ must occurs free in $\labs{\avari{n+1}}{\citer}$ and in $\citer$,
   because otherwise the $\sstregred$\nb-step displayed in \eqref{eq1:prf:lem:def:stRegposCRS}, which is a $\slabsdecompred$\nb-step,
   would not be $\srulep{\snlvarsucc}$-eager.
   Then Lemma~\ref{lem:def:stRegposCRS:item:v},
   implies that
   $\flabspos{\avari{1}\ldots\avari{n+1}}{\aposi{1},\ldots,\aposi{n+1}}{\aposi{n+1}00}{\citer} 
      \stregmred
    \flabspos{\avari{1}\ldots\avari{n}}{\aposi{1},\ldots,\aposi{n}}{\bposi{n+1}}{\avari{n}}$
   holds for some $\bposi{n+1} \ge \aposi{n+1}00$.
   Together with \eqref{eq1:prf:lem:def:stRegposCRS} it follows:
   \begin{equation*}
      \femptylabspos{\niks}{\rootpos}{\ater}
        \stregmred
      \flabspos{\avari{1}\ldots\avari{n+1}}{\aposi{1},\ldots,\aposi{n+1}}{\aposi{n+1}00}{\citer} 
        \stregmred
      \flabspos{\avari{1}\ldots\avari{n}}{\aposi{1},\ldots,\aposi{n}}{\bposi{n+1}}{\avari{n}} 
   \end{equation*}
   From this we obtain that $\bposi{n+1},\aposi{n+1}\in\Positions{\aiter}$ holds
   by Lemma~\ref{lem:def:stRegposCRS},~(\ref{lem:def:stRegposCRS:item:i}) ,
   and that
   $\aposi{n} \binds \bposi{n+1}$ and $\bposi{n+1} \iscapturedby \aposi{n+1}$
   hold by Proposition~\ref{prop:bind:iscapturedby}.
   With these links the already obtained \bindcaptchain\ can be extended to
   $\aposi{1} \binds \bposi{2} \iscapturedby \aposi{2} \binds \ldots \iscapturedby \aposi{n} \binds \bposi{n+1} \iscapturedby \aposi{n+1}$. 
   Also, we have seen that all positions in this \bindcaptchain\ are in $\Positions{\aiter}$.
   In this way we have established the induction step.
   
   Second, we prove statement~(\ref{lem:bind:capt:chains:stRegpos:item:ii}) of the lemma
   by induction on $n$, the number of binder positions in the assumed \bindcaptchain. 
   
   In the base case $n=0$ nothing needs to be shown. 
   For showing the case that $n=1$, let $\aposi{1}$ be a binder position in $\aiter$. 
   Then by using Proposition~\ref{prop:stRegposCRS},~(\ref{prop:stRegposCRS:item:labs}),
   and a $\sstregred$\nb-step over a \lambdaabstraction\ we obtain a rewrite sequence of the form:
   \begin{equation*}
     \femptylabspos{\tuple{\niks}}{\rootpos}{\aiter}
       \stregmred
     \flabspos{\avari{1}\ldots\avari{m-1}}{\cposi{1},\ldots,\cposi{m}}{\aposi{1}}{\labs{\avari{m}}{\biter}}
       \stregred
     \flabspos{\avari{1}\ldots\avari{m-1}\avari{m}}{\cposi{1},\ldots,\cposi{m-1},\cposi{m}}{\cposi{m}00}{\biter}  
   \end{equation*}
   for some $m\ge 1$ and $\cposi{1},\ldots,\cposi{m}\in\Positions{\aiter}$ with $\cposi{m} = \aposi{1}$.
   This shows the statement for $n=1$.
   
   For the induction step, we let $n\ge 1$ and
   $\aposi{1} \binds \bposi{2} \iscapturedby \aposi{2} \binds \ldots \iscapturedby \aposi{n} \binds \bposi{n+1} \iscapturedby \aposi{n+1}$
   be a \bindcaptchain\ in $\aiter$. 
   By applying the induction hypothesis to the \bindcaptchain\ formed by all but the last two links,
   we obtain $\cposi{1},\ldots,\cposi{m_0},\dpos\in\Positions{\ater}$ with $m_0\ge n$
   and a rewrite sequence:
   \begin{equation}\label{eq2:prf:lem:def:stRegposCRS}
     \begin{aligned}  
       \femptylabspos{\niks}{\rootpos}{\ater}
           \stregmred &
         \flabspos{\avari{1}\ldots\avari{m_0}}{\cposi{1},\ldots,\cposi{m_0}}{\cposi{m_0}00}{\biteri{0}}
       \\[-0.5ex]
         & 
         \text{such that
                $\aposi{1},\ldots,\aposi{n} \in \setexp{\cposi{1},\ldots,\cposi{m_0}}$
                and $\aposi{1} < \aposi{2} < \ldots < \aposi{n} = \cposi{m_0} $}
     \end{aligned}           
   \end{equation}
   Now since $\aposi{n} \binds \bposi{n+1}$ holds,
   it follows by Proposition~\ref{prop:bind:iscapturedby}
   that
   $\femptylabspos{\niks}{\rootpos}{\ater}
      \stregmred
    \flabspos{\bvari{1}\ldots\bvari{k}}{\dposi{1},\ldots,\dposi{k}}{\bposi{n+1}}{\bvari{k}}$
   for some $k\ge 1$ and $\dposi{1},\ldots,\dposi{k}$ such that
   $\dposi{k} = \aposi{n} = \cposi{m_0}$.
   By Lemma~\ref{lem:def:stRegposCRS},~(\ref{lem:def:stRegposCRS:item:iv}), it follows:
   \begin{equation*}
    \femptylabspos{\niks}{\rootpos}{\ater}
      \stregmred
    \flabspos{\bvari{1}\ldots\bvari{k}}{\dposi{1},\ldots,\dposi{k}}{\dposi{k}00}{\biteracci{0}}  
      \stregmred
    \flabspos{\bvari{1}\ldots\bvari{k}}{\dposi{1},\ldots,\dposi{k}}{\bposi{n+1}}{\bvari{k}}
   \end{equation*} 
   with some $\biteracci{0}$. 
   Then Lemma~\ref{lem:def:stRegposCRS},~(\ref{lem:def:stRegposCRS:item:ii}), entails
   $\flabspos{\bvari{1}\ldots\bvari{k}}{\dposi{1},\ldots,\dposi{k}}{\dposi{k}00}{\biteracci{0}}
      = 
    \flabspos{\avari{1}\ldots\avari{m_0}}{\cposi{1},\ldots,\cposi{m_0}}{\cposi{m_0}00}{\biteri{0}}$,
   $\flabspos{\bvari{1}\ldots\bvari{k}}{\dposi{1},\ldots,\dposi{k}}{\bposi{n+1}}{\bvari{k}}
      = 
    \flabspos{\avari{1}\ldots\avari{m_0}}{\cposi{1},\ldots,\cposi{m_0}}{\bposi{n+1}}{\avari{m_0}}$,
   and hence $k = m_0$.
   Thus we obtain:
   \begin{equation}\label{eq3:prf:lem:def:stRegposCRS}
     \femptylabspos{\niks}{\rootpos}{\ater}
       \stregred
     \flabspos{\avari{1}\ldots\avari{m_0}}{\cposi{1},\ldots,\cposi{m_0}}{\cposi{m_0}00}{\biteri{0}}
       \stregred
     \flabspos{\avari{1}\ldots\avari{m_0}}{\cposi{1},\ldots,\cposi{m_0}}{\bposi{n+1}}{\avari{m_0}}  
   \end{equation}    
   Due to the last link $\bposi{n+1} \iscapturedby \aposi{n+1}$ in the assumed \bindcaptchain\ 
   there exists, in view of Proposition~\ref{prop:bind:iscapturedby},
   a rewrite sequence
   $\femptylabspos{\niks}{\rootpos}{\ater}
       \stregmred
     \flabspos{\cvari{1}\ldots\cvari{l}}{\dposi{1},\ldots,\dposi{l}}{\dposi{l}00}{\citer}
       \stregmred
     \flabspos{\cvari{1}\ldots\cvari{l_0}}{\dposi{1},\ldots,\dposi{l_0}}{\bposi{n+1}}{\cvari{l_0}}$
   for some $\dposi{1},\ldots,\dposi{l}$ 
   with $\dposi{l} = \aposi{n+1}$, and $1 \le l_0 < l$. 
   By Lemma~\ref{lem:def:stRegposCRS},~(\ref{lem:def:stRegposCRS:item:ii}), it follows 
   from this rewrite sequence and the one in \eqref{eq3:prf:lem:def:stRegposCRS} that
   $\flabspos{\cvari{1}\ldots\cvari{l_0}}{\dposi{1},\ldots,\dposi{l_0}}{\bposi{n+1}}{\cvari{l_0}}
      =
    \flabspos{\avari{1}\ldots\avari{m_0}}{\cposi{1},\ldots,\cposi{m_0}}{\bposi{n+1}}{\avari{m_0}}$,
   $l_0 = m_0$, and
   $\flabspos{\cvari{1}\ldots\cvari{l}}{\dposi{1},\ldots,\dposi{l}}{\dposi{l}00}{\citer}
      =
    \flabspos{\cvari{1}\ldots\cvari{l}}{\cposi{1},\ldots,\cposi{m_0},\dposi{{l_0}+1},\ldots,\dposi{l}}{\dposi{l}00}{\citer}$.  
   Hence we obtain:   
   \begin{equation}\label{eq4:prf:lem:def:stRegposCRS}
     \femptylabspos{\niks}{\rootpos}{\ater}
       \stregmred
     \flabspos{\cvari{1}\ldots\cvari{l}}{\cposi{1},\ldots,\cposi{m_0},\dposi{{l_0}+1},\ldots,\dposi{l}}{\dposi{l}00}{\citer}
       \stregmred
     \flabspos{\avari{1}\ldots\avari{m_0}}{\cposi{1},\ldots,\cposi{m_0}}{\bposi{n+1}}{\avari{m_0}}    
   \end{equation}
   with $\dposi{l} = \aposi{n+1}$. 
   Now we let $m \defdby l$,
   $\cposi{m_0 +1} \defdby \dposi{m_0 + 1}$,
   \ldots,
   $\cposi{l} \defdby \dposi{l}$
   (note that thus $\cposi{l} = \dposi{l} = \aposi{n+1}$)
   and 
   $\flabspos{\avari{1}\ldots\avari{m}}{\cposi{1},\ldots,\cposi{m}}{\cposi{m}00}{\biter}
      \defdby
    \flabspos{\cvari{1}\ldots\cvari{l}}{\cposi{1},\ldots,\cposi{m_0},\dposi{{l_0}+1},\ldots,\dposi{l}}{\dposi{l}00}{\citer}$  
   to find that \eqref{eq4:prf:lem:def:stRegposCRS} then yields:  
   \begin{equation*}
     \femptylabspos{\niks}{\rootpos}{\ater}
       \stregmred
     \flabspos{\avari{1}\ldots\avari{m}}{\cposi{1},\ldots,\cposi{m}}{\cposi{m}00}{\biter}
   \end{equation*}
   with 
   $\aposi{1},\ldots,\aposi{n+1} \in \setexp{\cposi{1},\ldots,\cposi{m}}$
   due to 
   $\aposi{1},\ldots,\aposi{n} \in \setexp{\cposi{1},\ldots,\cposi{m_0}}$
   and
   $\aposi{n+1} = \dposi{l} = \cposi{m}$.
   Since by Lemma~\ref{lem:def:stRegposCRS}, (\ref{lem:def:stRegposCRS:item:i}), 
   $\cposi{1} < \cposi{2} < \ldots < \cposi{m}$ holds, 
   we have
   $\aposi{1} < \aposi{2} < \ldots < \aposi{n} = \cposi{m_0}$ 
   by the induction hypothesis,
   and know $m_0 < m$ and $\cposi{m} = \aposi{n+1}$, 
   we also obtain 
   $\aposi{1} < \aposi{2} < \ldots < \aposi{n} < \aposi{n+1} = \cposi{m}$.
   In this way we have successfully performed the induction step. 
\end{proof}

The following lemma relates the length of \bindcaptchains\ in an infinite \lambdaterm~$\aiter$
with the length of abstraction prefixes in generated subterms of $\aiter$ with respect to $\stregred$,
that is, the length of abstraction prefixes of terms that can be obtained from $\femptylabs{\aiter}$
by $\stregred$\nb-rewrite sequences in the not \posannotated\ system $\stRegCRS$.  

\begin{lemma}\label{lem:fin:bind:capt:chains}
  Let $\aiter\in\Ter{\inflambdacal}$. For all $n\in\nats$ it holds:
  $\aiter$ contains a \bindcaptchain\ of length $\max \setexp{n-1,0}$
  if and only if
  there is a rewrite sequence 
  $\femptylabs{\aiter} \stregmred \flabs{\avari{1}\ldots\avari{n}}{\biter}$
  (in $\stRegCRS$) for some $\flabs{\avari{1}\ldots\avari{n}}{\biter} \in\gSTstreg{\aiter}$. 
\end{lemma}

\begin{proof}
  Let $\aiter\in\Ter{\inflambdacal}$. 
  For showing the direction ``$\Leftarrow$'',
  we assume a rewrite sequence in $\stRegARS$ of the form
  $\femptylabs{\aiter} \stregmred \flabs{\avari{1}\ldots\avari{n}}{\biter}$.
  By Proposition~\ref{prop:position:lifting:projecting}, (\ref{prop:position:lifting:projecting:item:lifting}), 
  this rewrite sequence can be lifted to a rewrite sequence
  $\femptylabspos{\niks}{\rootpos}{\aiter} 
     \stregmred
   \flabspos{\avari{1}\ldots\avari{n}}{\aposi{1},\ldots,\aposi{n}}{\bpos}{\biter}$
  in $\stRegposARS$.
  Then it follows from Lemma~\ref{lem:bind:capt:chains:stRegpos}, (\ref{lem:bind:capt:chains:stRegpos:item:i}),
  that there exists a \bindcaptchain\ in $\aiter$ of length~$\max \setexp{n-1,0}$.
  
  Now we show the direction ``$\Rightarrow$''.
  For the case $n=0$ nothing has to be shown.
  Now we let $n\ge 1$, and
  suppose that
  $\aposi{1} \binds \bposi{2} \iscapturedby \aposi{2} \binds \ldots \iscapturedby \aposi{n}$
  is a \bindcaptchain\ in $\aiter$ of length $n-1$. 
  Then by an appeal to Lemma~\ref{lem:bind:capt:chains:stRegpos},~(\ref{lem:bind:capt:chains:stRegpos:item:ii}),
  we obtain a rewrite sequence $\arewseq$ in $\stRegposCRS$ of the form
  $\femptylabspos{\niks}{\rootpos}{\ater}
         \stregmred
   \flabspos{\avari{1}\ldots\avari{m}}{\cposi{1},\ldots,\cposi{m}}{\cposi{m}00}{\biter}$
  for positions $\cposi{1},\ldots,\cposi{m}\in\Positions{\ater}$ with $m\ge n$ such that
  $\aposi{1},\ldots,\aposi{n} \in \setexp{\cposi{1},\ldots,\cposi{m}}$,
  and $\aposi{1} < \aposi{2} < \ldots < \aposi{n} = \cposi{m} $.
  By possibly repeated application of 
  Lemma~\ref{lem:def:stRegposCRS},~(\ref{lem:def:stRegposCRS:item:iii}),
  we obtain that $\arewseq$ is actually of the form
  $\femptylabspos{\tuple{\niks}}{\rootpos}{\aiter}
     \stregmred
   \flabspos{\avari{1}\ldots\avari{n}}{\cposi{1},\ldots,\cposi{n}}{\cposi{n}00}{\biteri{0}}
     \stregmred
   \flabspos{\avari{1}\ldots\avari{m}}{\cposi{1},\ldots,\cposi{m}}{\cposi{m}00}{\biter}  
   $.
  By applying
  Proposition$\,$\ref{prop:position:lifting:projecting},$\,$(\ref{prop:position:lifting:projecting:item:projecting}),
  to the first seqment of $\arewseq$ displayed here
  we obtain the $\stregred$\nb-re\-write sequence
  $\femptylabs{\aiter}
     \stregmred
   \flabs{\avari{1}\ldots\avari{n}}{\biteri{0}}$
  in $\stRegCRS$ by projection through just dropping the position annotations.  
  By the definition of $\gSTstreg{\aiter}$ (see Definition~\ref{def:gST:reg:streg}),
  $\flabs{\avari{1}\ldots\avari{n}}{\biteri{0}} \in\gSTstreg{\aiter}$ follows. 
\end{proof}

The lemma below states a condition that guarantees an infinite \bindcaptchain\
in an infinite \lambdaterm: the existence of an infinite $\stregred$\nb-re\-write sequence
in which the length of the abstraction prefixes tends to infinity in the limit. 

\begin{lemma}[infinite \bindcaptchains]\label{lem:inf:bind:capt:chains}
  Let $\aiter$ be a \lambdaterm, and let $\arewseq$ be an infinite $\sstregred$\nb-re\-write sequence
  $ \femptylabs{\aiter} = \flabs{\vec{\avar}_0}{\aiteri{0}}
      \stregred
    \flabs{\vec{\avar}_1}{\aiteri{1}}
      \stregred
    \ldots $
  such that $\lim_{i\to\infty} \length{\vec{\avar}_i} = \infty$. 
  Then there exists an infinite \bindcaptchain\ in $\aiter$.
\end{lemma}

\begin{proof}
  Let $\aiter$ and $\arewseq$ be as in the assumption of the lemma.
  By Proposition\,\ref{prop:position:lifting:projecting}, (i), 
  $\arewseq$ can be lifted to a rewrite sequence with position annotations:
  \begin{center}
    $  
    \arewseq^{\spos} \;\funin\;\;
       \femptylabspos{\niks}{\rootpos}{\aiter} = \flabspos{\vec{\avar}_0}{\niks}{\rootpos}{\aiteri{0}}
         \stregred
       \flabspos{\vec{\avar}_1}{\vec{\apos}_1}{\bpos_1}{\aiteri{1}}
         \stregred
       \ldots
         \stregred
       \flabspos{\vec{\avar}_i}{\vec{\apos}_i}{\bpos_i}{\aiteri{i}}
         \stregred
       \ldots
       $
  \end{center}     
  where, for all $i\in\nats$,
  $\bposi{i}$ are positions and $\vec{\apos}_i = \tuple{\aposi{1},\ldots,\aposi{m_i}}$ vectors of positions,
  with $m_i\in\nats$.
 
  Due to $\lim_{i\to\infty} \length{\vec{\avar}_i} = \infty$
  it follows that $\liminf_{i\to\infty} \length{\vec{\avar}_i} = \infty$,
  and hence there exists a sequence $\sequence{i_j}{j\in\nats}$ 
  of increasing $0 = i_0 < i_1 < i_2 < \ldots$ natural numbers
  such that
  $0 = \length{\vec{\avar}_{i_0}}
     < \length{\vec{\avar}_{i_1}}
     < \length{\vec{\avar}_{i_2}}
     < \ldots$,
  and
  $\length{\vec{\avar}_{i_j}} \le \length{\vec{\avar}_k}$
  for all $j,k\in\nats$ with $k\ge i_j$.
  Now if the rewrite sequence $\arewseq^{\spos}$ is written with highlighted segments of the form:
  \begin{equation}\label{eq2:prf:lem:inf:bind:capt:chains}
  \begin{aligned} 
    \arewseq^{\spos} \;\funin\;\;
      \femptylabspos{\niks}{\rootpos}{\aiter} =
      \flabspos{\vec{\avar}_{i_0}}{\vec{\apos}_{i_0}}{\bpos_{i_0}}{\aiteri{i_0}}
        & \stregmred  
      \ldots
      \\
      \ldots  & \stregmred
      \flabspos{\vec{\avar}_{i_j}}{\vec{\apos}_{i_j}}{\bpos_{i_j}}{\aiteri{i_j}}
        \stregmred
      \flabspos{\vec{\avar}_{i_{j+1}}}{\vec{\apos}_{i_{j+1}}}{\bpos_{i_{j+1}}}{\aiteri{i_{j+1}}}
        \stregmred
      \ldots
  \end{aligned}   
  \end{equation}
  (where $j\in\nats$),
  then it follows, for all $j\in\nats$, 
  that 
       $\lengthnormalsize{\vec{\apos}_{i_j}} < \lengthnormalsize{\vec{\apos}_{i_{j+1}}}$,
  and that all terms of the sequence after
  $\flabspos{\vec{\avar}_{i_j}}{\vec{\apos}_{i_j}}{\bpos_{i_j}}{\aiteri{i_j}}$
  have an abstraction prefix of length greater or equal $\lengthnormalsize{\vec{\avar}_{i_j}}$.
  Due to the property of steps in $\stRegposCRS$ to remove position annotations only 
  when the corresponding abstraction variable is dropped from the prefix in an $\scompressstregred$-step,
  it follows that $\vec{\apos}_{i_j} < \vec{\apos}_{i_{j+1}}$ holds in the prefix order, for all $j\in\nats$, and hence 
  that $\vec{\apos}_{i_0} < \vec{\apos}_{i_1} < \vec{\apos}_{i_2} < \ldots $.
  
  Consequently, the position vectors $\vec{\apos}_{i_j}$ tend towards an infinite vector 
  $\tuple{\cposi{1},\cposi{2},\cposi{2},\cdots}$ of positions
  such that there are $n_1 < n_2 < \ldots$ in $\nats$
  with  $\vec{\apos}_{i_j} = \tuple{\cposi{1},\cposi{2},\cdots,\cposi{n_j}}$ 
  for all $j\in\nats$ with $j>0$.
  But then induction on the segmented structure of $\arewseq^{\spos}$ as indicated in \eqref{eq2:prf:lem:inf:bind:capt:chains},
  thereby using Lemma\,\ref{lem:bind:capt:chains:stRegpos}, (\ref{lem:bind:capt:chains:stRegpos:item:i}),
  establishes the existence of positions $\dposi{1},\dposi{2},\ldots$ such that:
  \begin{center} 
    $ 
    \cposi{1}
      \binds \dposi{1} \iscapturedby
    \cposi{2}
       \binds 
    \ldots
      \iscapturedby
    \cposi{n_j}
      \binds \dposi{n_j +1} \iscapturedby
    \cposi{n_j +1}
      \binds \dposi{n_j +2} \iscapturedby
    \ldots
       \binds \dposi{n_{j+1}} \iscapturedby
    \cposi{n_{j+1}}
      \binds            
    \ldots
     $
  \end{center}   
  holds, thereby yielding an infinite \bindcaptchain\ in $\aiter$.
\end{proof}

Now we formulate and prove the main theorem of this section,
which applies the concept of \bindcaptchain\
to pin down, among all infinite \lambdaterms\ that are regular,
those that are strongly regular.

\begin{theorem}\label{thm:streg:fin:bind:capt:chains}
  A regular \lambdaterm\ is strongly regular
    if and only if
  it contains only finite \bindcaptchains.
\end{theorem}


\begin{proof}
 Let $\aiter$ be an infinite \lambdaterm\ that is regular.
 
 For showing ``$\Rightarrow$'', suppose that $\aiter$ is also strongly regular.
 Then by the definition of strong regularity, $\gSTstreg{\aiter}$ is finite.
 Let $n$ be the length of the longest abstraction prefix of a term in $\gSTstreg{\aiter}$.
 Then 
 Lemma~\ref{lem:fin:bind:capt:chains} implies that the length of
 every \bindcaptchain\ in $\aiter$ is bounded by $\max \setexp{n-1,0}$.
 Hence $\aiter$ only contains finite \bindcaptchains.

 For the implication ``$\Leftarrow$'' we argue indirectly:
 assuming that $\aiter$ is not strongly regular, we show the existence of 
 an infinite \bindcaptchain\  in $\aiter$.
 
 So suppose that $\aiter$ is not strongly regular.
 Then $\gSTstreg{\aiter}$ is infinite,
 and so $\femptylabs{\aiter}$ has infinitely many $\sregred$\nb-reducts.
 Since the rewrite strategy $\sstregred$ has branching degree~$\le 2$
 (branching only happens at sources of $\slappdecompired{i}$\nb-steps),
 it follows by K\H{o}nig's Lemma%
    \footnote{Here we use the following version of \emph{K\H{o}nig's Lemma}: 
     Let $\agraph$ 
     be a rooted directed graph with root $r$. 
     Suppose that $\agraph$ has infinitely many vertices,
     that every vertex of $\agraph$ is reachable from $r$ via a directed path,
     and that every vertex has finite out-degree (finitely many successor vertices).
     Then there exists an infinite directed path in $\agraph$ that starts at $r$ and is simple (no repetitions of vertices).}
 that there is an infinite rewrite sequence:
 \begin{center}
   $
   \arewseq \;\funin\;\;
      \femptylabs{\aiter} = \flabs{\vec{\avar}_0}{\aiteri{0}}
        \stregred
      \flabs{\vec{\avar}_1}{\aiteri{1}}
        \stregred
      \ldots
        \stregred
      \flabs{\vec{\avar}_i}{\aiteri{i}}
        \stregred
      \ldots
      $
 \end{center}
 that passes through distinct terms.
 By Lemma~\ref{lem:projection:lifting:RegCRS:stRegCRS}, (\ref{lem:projection:lifting:RegCRS:stRegCRS:item:projection}),
 this rewrite sequence projects to: 
 \begin{center}
   $
   \Checkreg{\arewseq} \;\funin\;\;
      \femptylabs{\aiter} = \flabs{\vec{\avar}'_0}{\aiteri{0}}
        \regmred
      \flabs{\vec{\avar}'_1}{\aiteri{1}}
        \regmred
      \ldots
        \regmred
      \flabs{\vec{\avar}'_i}{\aiteri{i}}
        \regmred
      \ldots \punc{,}
  $      
  \vspace*{-1.5ex}
  \begin{equation}\label{eq3:prf:thm:streg:fin:bind:capt:chains}
    \hspace*{15ex}
    \text{so that for all $j\in\nats$:} \hspace*{5ex}
    \flabs{\vec{\avar}_j}{\aiteri{j}} \compressregmred \flabs{\vec{\avar}'_j}{\aiteri{j}}
  \end{equation}    
 \end{center}
 thereby respectively shortening the length of the abstraction prefix. 
 Since $\aiter$ is regular, $\gSTreg{\aiter}$ is finite,
 and hence only finitely many terms occur in $\Checkreg{\arewseq}$.
 Now we use this contrast with $\arewseq$ together with \eqref{eq3:prf:thm:streg:fin:bind:capt:chains}
 to show that the prefix lengths of terms in $\arewseq$ tend to infinity.  

 Suppose that $\lim_{i\to\infty} \length{\vec{\avar}_i} = \infty$ does not hold.
 Then there exists $l_0\in\nats$ such that $\length{\vec{\avar}_i} < l_0$ for infinitely many $i\in\nats$.
 Hence there is a sequence $i_0 < i_1 < i_2 < i_3 < \ldots$ in $\nats$
 such that:
 \begin{gather}
   \aset \defdby \descsetexp{\flabs{\vec{\avar}_{i_j}}{\aiteri{i_j}}}{j\in\nats}
     \text{ is infinite}
   \label{eq6:prf:thm:streg:fin:bind:capt:chains}
 \end{gather}
 (since the terms on $\arewseq$ are distinct),
 and $\length{\vec{\avar}_{i_j}} < l_0$ for all $\flabs{\vec{\avar}_{i_j}}{\aiteri{i_j}} \in S$.
 On the other hand:
 \begin{equation*}
   \bset \defdby
   \descsetexp{\flabs{\vec{\avar}'_{i_j}}{\aiteri{i_j}}}{j\in\nats}
     \; \subseteq \;
   \gSTreg{\aiter} \text{ is finite}
 \end{equation*}
 because $\aiter$ is regular. 
 However, since every term in $\aset$ has a $\scompressregmred$\nb-reduct in $\bset$ due to \eqref{eq3:prf:thm:streg:fin:bind:capt:chains},
 as well as an abstraction prefix of a length bounded by $l_0$,
 it follows by Proposition\,\ref{prop:compress:prefix:RegARS}, (\ref{prop:compress:prefix:RegARS:item:ii}),
 that $\aset$ also has to be finite, in contradiction with \eqref{eq6:prf:thm:streg:fin:bind:capt:chains}.
 Hence we conclude: $\lim_{i\to\infty} \length{\vec{\avar}_i} = \infty$.

 Now Lemma~\ref{lem:inf:bind:capt:chains} is applicable to $\arewseq$, and yields an infinite \bindcaptchain\ in $\aiter$.
\end{proof}

By adding the statement of Proposition~\ref{prop:def:reg:streg} we obtain the following accentuation
of this theorem.

\begin{corollary}\label{cor:thm:streg:fin:bind:capt:chains}
  An infinite \lambdaterm\ is strongly regular 
    if and only if 
  it is regular, 
  and contains only finite \bindcaptchains. 
\end{corollary}


\section{Expressibility by terms of the \protect\lambdacalculus\ with $\mu$}
\label{sec:express}

Having adapted (in Section~\ref{sec:regular}) the concept of regularity for infinite \lambdaterms\ in two ways,
we now obtain an expressibility result for one of these adaptations 
that is analogous to that in \cite{cour:1983}
for regular first-order trees with respect to rational expressions (or equivalently, $\mu$\nb-terms).
We show that an infinite \lambdaterm\ is strongly regular if and only if it is \lambdamucalexpressible. 

We first define terms of $\lambdamucal$, the unfolding rewrite relation, and \lambdamucalexpressibility.

\begin{definition}[\CRS-representation for $\lambdamucal$]
    \normalfont\label{def:lambdamucal}  
  The \CRS\nb-signature $\siglmcCRS = \siglcCRS \cup \setexp{\smuCRS}$
  for $\lambdamucal$ extends $\siglcCRS$ by a unary function symbol $\smuCRS$.
  By $\Ter{\lambdamucal}$ we denote the set of closed finite
  CRS\nb-terms over $\siglmcCRS$ with the restriction
  that CRS-abstraction occurs only as an argument of the symbols $\slabsCRS$ or $\smuCRS$. 
  By $\Ter{\lambdamuprefixcal}$ we denote the analogously defined set of terms over the signature $\siglpcCRS\cup\setexp{\smuCRS}$.
  We consider the \emph{$\mu$\nb-un\-fol\-ding rule} in informal and formal notation:
  \begin{center}
    $
    (\srulep{\smu}): \;\;
      \muabs{\avar}{\funap{\almter}{\avar}}
        \red
      \funap{\almter}{\muabs{\avar}{\funap{\almter}{\avar}}}
    \hspace*{9ex}  
    \srulep{\smu}: \;\;
      \muabsCRS{\avar}{\cmetavar{\avar}}
        \red
      \cmetavar{\muabsCRS{\avar}{\cmetavar{\avar}}}
      $
  \end{center}
  This rule induces the \emph{unfolding rewrite relation} $\sunfoldred$ on $\Ter{\lambdamucal}$ and $\Ter{\lambdamuprefixcal}$.
  We say that a \lambdamuterm~$\almter$ \emph{expresses} an infinite \lambdaterm~$\citer$
  if $\almter \unfoldinfred \citer$ holds, that is, $\almter$ unfolds to $\citer$ 
  via a typically infinite, strongly convergent $\sunfoldred$\nb-rewrite sequence
  (similar for terms in $\Ter{\lambdamuprefixcal}$).
  And an infinite \lambdaterm~$\aiter$ is \emph{\lambdamucalexpressible} if there is a \lambdamuterm~$\almter$ that expresses $\aiter$.
\end{definition}

We sketch some intuition for the proof, which proceeds by a sequence of proof-theoretic transformations. 
We focus on the more difficult direction.
Let $\aiter$ be a strongly regular infinite \lambdaterm. 
We want to extract a \lambdamuterm~$\almter$ that expresses $\aiter$
from the finite $\sstregred$-re\-duc\-tion graph $G$ of $\aiter$. 
We first obtain a closed derivation $\Deriv$ of $\femptylabs{\aiter}$ in $\stRegzero$.
The derivation $\Deriv$ can be viewed as a finite term graph that has $G$ as its homomorphic image,
and that does not exhibit horizontal sharing (\cite[Sec.\hspace*{1.5pt}4.3]{blom:2001}). 
Such term graphs correspond directly to \lambdamuterms\ (analogous to \cite{blom:2001}). 
In order to extract the \lambdamuterm~$\almter$ corresponding to $\Deriv$ from this derivation,
we annotate it inductively to  
a \lambdamuterm-annotated derivation $\Hat{\Deriv}$
with conclusion $\femptylabsann{\almter}{\aiter}$ in a proof system $\Expr$ that is a variant of $\stRegzero$.
Then it remains to show that $\almter$ indeed unfolds to $\aiter$.
For this we prove
that $\Hat{\Deriv}$ unfolds to/gives rise to infinite derivations in the variant systems $\Exprinf$ and $\Unfinf$,
which witness infinite outermost rewrite sequences $\almter \unfoldinfred \aiter$.

%

\vspace*{0.7ex}
The \CRS\ consisting of the rule $\srulep{\smu}$ is orthogonal and fully-extended \cite{terese:2003}.
As a consequence of the result in \cite{kete:simo:2010} that 
outermost-fair strategies in orthogonal, fully extended \iCRSs\ are normalizing,
we obtain the following proposition. 

\begin{proposition}\label{prop:unfoldomred}
  Let $\almter\in\Ter{\lambdamucal}$ and $\aiter\in\Ter{\inflambdacal}$. 
  If $\almter$ expresses $\aiter$,
  then there is an outermost $\sunfoldred$\nb-rewrite sequence of length ${\le}\hspace*{1pt}\omega$ that witnesses
  $\almter \unfoldinfred \aiter$,
  and $\aiter$ is the unique \lambdaterm\ expressed by $\almter$. 
  Analogously for prefixed terms in $\lambdamucal$ that express prefixed \lambdaterms.
  %
  Hence
  the infinite outermost unfolding rewrite relation $\unfoldominfnfred$ to infinite normal form
  defines a partial mapping from $\Ter{\lambdamucal}$ to $\Ter{\inflambdacal}$, and from $\Ter{\lambdamuprefixcal}$ to $\Ter{\inflambdaprefixcal}$.
\end{proposition}

The  relation $\unfoldominfnfred$ can be defined via derivability in the proof system $\Unfinf$ in Fig.\,\ref{fig:Unfinf}:
the existence of a possibly infinite derivation 
that is closed in the sense of Def.~\ref{def:Lambdaprefixreg:Lambdaprefixstreg},
and \emph{admissible}, 
i.e.\ it is \Vacstregeager, and does not contain infinitely many consecutive instances of the rule ($\ssmuabs$).

\begin{figure}
\vspace*{-2.5ex}  
  
\begin{center}  
  \framebox{
\begin{minipage}{380pt}
 \vspace*{1.25ex}
\begin{center}
  \mbox{
    \AxiomC{$\phantom{\flabs{\vec{\avar}\bvar}{\bvar}  \unfoldsto  \flabs{\vec{\avar}\bvar}{\bvar}}$}
    \RightLabel{$\snlvar$}
    \UnaryInfC{$\flabs{\vec{\avar}\bvar}{\bvar}  \unfoldsto  \flabs{\vec{\avar}\bvar}{\bvar}$}
    \DisplayProof
        }
  \hspace*{1ex}
  \mbox{
    \AxiomC{$\flabs{\vec{\avar}}{\almteri{0}}
               \unfoldsto
             \flabs{\vec{\avar}}{\aiteri{0}}$}
    \AxiomC{$\flabs{\vec{\avar}}{\almteri{1}}
               \unfoldsto
             \flabs{\vec{\avar}}{\aiteri{1}}$}
    \RightLabel{$\sslapp$}
    \BinaryInfC{$\flabs{\vec{\avar}}{\lapp{\almteri{0}}{\almteri{1}}}
                  \unfoldsto
                \flabs{\vec{\avar}}{\lapp{\aiteri{0}}{\aiteri{1}}}$}
    \DisplayProof
        }          
  \\[1.75ex]
  \mbox{
    \AxiomC{$\flabs{\vec{\avar}\bvar}{\almter} 
               \unfoldsto
             \flabs{\vec{\avar}\bvar}{\aiter}$}
    \RightLabel{$\sslabs$}
    \UnaryInfC{$\flabs{\vec{\avar}}{\labs{\bvar}{\almter}}
                  \unfoldsto
                \flabs{\vec{\avar}}{\labs{\bvar}{\aiter}}$}
    \DisplayProof
        }
  \hspace*{1ex}
  \mbox{
    \AxiomC{$\flabs{\vec{\avar}}{\almter}
               \unfoldsto
             \flabs{\vec{\avar}}{\aiter}$}
    \RightLabel{$\snlvarsucc$
                \parbox{67pt}{\small \rule{0pt}{4.75pt}\\
                                     \rule{0pt}{4.75pt}\\
                                     (if the binding $\slabs{\bvar}$\\[-0.35ex]
                                      \hspace*{1.25em} is vacuous\\[-0.25ex] 
                                      \hspace*{1.25em} in $\almter$ and $\aiter$)}}
    \UnaryInfC{$\flabs{\vec{\avar}\bvar}{\almter}
                  \unfoldsto
                \flabs{\vec{\avar}\bvar}{\aiter}$}
    \DisplayProof         
        }  
  \\[1.75ex]    
  \mbox{
    \AxiomC{$\flabs{\vec{\avar}}{\funap{\almter}{\muabs{\arecvar}{\funap{\almter}{\arecvar}}}} 
               \unfoldsto
             \flabs{\vec{\avar}}{\aiter}$}
    \RightLabel{$\ssmuabs$}
    \UnaryInfC{$\flabs{\vec{\avar}}{\muabs{\arecvar}{\funap{\almter}{\arecvar}}}
                  \unfoldsto
                \flabs{\vec{\avar}}{\aiter}$}
    \DisplayProof
        }   
\vspace*{1ex}  
\end{center}
\end{minipage}
   }
\end{center}
\vspace*{-2ex} 
\caption{\label{fig:Unfinf}%
         Proof system $\Unfinf$ for completely unfolding of \lambdamuterms\
         into infinite \lambdaterms.
 }
\end{figure}

\begin{proposition}\label{prop:Unfinf}
  $\Unfinf$ is sound and complete w.r.t.\ $\sunfoldominfnfred\,$:
  For all $\flabs{\vec{\avar}}{\aiter}\in\Ter{\inflambdaprefixcal}$ and $\flabs{\vec{\avar}}{\almter}\in\Ter{\lambdamuprefixcal}$, 
  $\,\derivablein{\Unfinf}{\flabs{\vec{\avar}}{\almter} \unfoldsto \flabs{\vec{\avar}}{\aiter}}$ 
    holds 
    if and only if
  $\,\flabs{\vec{\avar}}{\almter} \unfoldominfnfred \flabs{\vec{\avar}}{\aiter}$.
\end{proposition}

\begin{definition}[proof systems $\Expr$, $\Exprinf$, and $\Exprmu$, $\Exprmuinf$]%
    \label{def:Expr:Exprinf:Exprmu:Exprmuinf}
  The natural-de\-duc\-tion-style proof system $\Expr$ has as its formulas abstraction-prefixed \lambdamuterms\ annotated by infinite \lambdaterms,
  and the rules in Fig.\,\ref{fig:Expr:Exprmu}.  
  The system $\Exprmu$ has abstraction-prefixed \lambdamuterms\ as formulas, and its rules arise from $\Expr$
  by dropping the \lambdaterms. 
  Derivability in these systems means the existence of a closed (no open assumptions), \Vacstregeager, finite derivation.
  
  The variant $\Exprinf$ of $\Expr$ arises by replacing the rule ($\sFIX$) with the rule ($\ssmuabs$) in Fig.\,\ref{fig:Exprinf:Exprmuinf}.
  $\Exprmuinf$ arises from $\Exprmu$ analogously. A derivation in either of these systems is called \emph{admissible}
  if it does not contain infinitely many consecutive instances of ($\ssmuabs$), and 
  if it is \Vacstregeager\ in the sense of Def.~\ref{def:Reginf:stReginf}.
  Derivability in these systems means the existence of an admissible derivation that is
  closed in the sense of Def.~\ref{def:Lambdaprefixreg:Lambdaprefixstreg}.  
\end{definition}

We first observe that derivations in $\Exprinf$ are closely linked to derivations in $\Unfinf$.

\begin{lemma}\label{lem:Exprinf:Unfinf}
  $\derivablein{\Exprinf}{\flabsann{\vec{\avar}}{\almter}{\aiter}}$ holds
    if and only if   
  $\derivablein{\Unfinf}{\flabs{\vec{\avar}}{\almter}  \unfoldsto  \flabs{\vec{\avar}}{\aiter}}$ holds,
  for all $\flabs{\vec{\avar}}{\almter}\in\Ter{\lambdamuprefixcal}$ and $\flabs{\vec{\avar}}{\aiter}\in\Ter{\inflambdaprefixcal}$.
\end{lemma}

\begin{proof}
  Derivations in $\Exprinf$ and in $\Unfinf$ differ only in the notation used for their formulas.
  A formula
  $\flabsann{\vec{\bvar}}{\blmter}{\biter}$ in $\Exprinf$
  corresponds to the formula
  $\flabs{\vec{\bvar}}{\blmter} \unfoldsto \flabs{\vec{\bvar}}{\biter}$
  in $\Unfinf$.
  This correspondence 
  extends to a correspondence between derivations of $\Exprinf$ and $\Unfinf$,
  which preserves and reflects the property of derivations to be closed and admissible. 
\end{proof}

The lemma below gathers basic properties of the proof systems $\Exprmu$ and $\Exprmuinf$. 

\begin{lemma}\label{lem:Exprmu:Exprmuinf}
  \begin{enumerate}[(i)]
    \item{}\label{lem:Exprmu:Exprmuinf:item:i}
      For every \lambdamuterm~$\almter$:
      $\derivablein{\Exprmu}{\femptylabs{\almter}}$
      if and only if
      $\derivablein{\Exprmuinf}{\femptylabs{\almter}}$.
    \item{}\label{lem:Exprmu:Exprmuinf:item:ii} 
      Every closed derivation in $\Exprmuinf$ contains only finitely many \lambdamuterms.
    \item{}\label{lem:Exprmu:Exprmuinf:item:iii} 
      For every \lambdamuterm~$\almter$ it holds:
      $\derivablein{\Exprmu}{\femptylabs{\almter}}$
      if and only if 
      there is no $\sstregred$\nb-ge\-ne\-ra\-ted subterm of $\almter$ (in $\gSTstreg{\almter}$) 
      of the form $\femptylabs{\muabs{\avari{0}\ldots\avari{n}}{\avari{0}}}$ for $n\in\nats$.  
  \end{enumerate}
\end{lemma}

\begin{proof}
  For (\ref{lem:Exprmu:Exprmuinf:item:i}), in order to show ``$\Rightarrow$'' let $\Deriv$
  be a finite, closed, \Vacstregeager\ derivation in $\Exprmu$ with conclusion $\femptylabs{\almter}$. 
  By `unfolding' this derivation through a process in which in each step: 
  \begin{center}
    \parbox{20ex}{\centering a subderivation\\of a bottommost\\instance of $\sFIXExpr$} \hspace*{1ex}
    \scalebox{0.9}{$
    \begin{aligned}[c]
      \AxiomC{$ [\flabs{\vec{\bvar}}{\aconstnamei{\amarker}}]^{\amarker} $}
      \noLine
      \UnaryInfC{$\funap{\Derivi{0}}{\aconstnamei{\amarker}}$}
      \noLine
      \UnaryInfC{$ \flabs{\vec{\bvar}}{\funap{\blmter}{\aconstnamei{\amarker}}} $}
      \RightLabel{$\sFIXExpr,l$ $\;$}                   
      \UnaryInfC{$ \flabs{\vec{\bvar}}{\muabs{\arecvar}{\funap{\blmter}{\arecvar}}} $}
      \DisplayProof
    \end{aligned}  
                   $}
    \hspace*{1ex}\parbox[c]{19ex}{\centering is `unfolded' into\\a subderivation}\hspace*{1ex} 
    \scalebox{0.9}{$ 
    \begin{aligned}[c]
      \AxiomC{$ [\flabs{\vec{\bvar}}{\aconstnamei{\amarker}}]^{\amarker} $}
      \noLine
      \UnaryInfC{$\funap{\Derivi{0}}{\aconstnamei{\amarker}}$}
      \noLine
      \UnaryInfC{$ \flabs{\vec{\bvar}}{\funap{\blmter}{\aconstnamei{\amarker}}} $}
      \RightLabel{$\sFIXExpr,l$ $\;$}
      \UnaryInfC{$ [\flabs{\vec{\bvar}}{\muabs{\arecvar}{\funap{\blmter}{\arecvar}}}] $}
      \noLine
      \UnaryInfC{$\funap{\Derivi{0}}{\muabs{\arecvar}{\funap{\blmter}{\arecvar}}}$}
      \noLine
      \UnaryInfC{$ \flabs{\vec{\bvar}}{\funap{\blmter}{\muabs{\arecvar}{\funap{\blmter}{\arecvar}}}} $}
      \RightLabel{$\ssmuabs$}                   
      \UnaryInfC{$ \flabs{\vec{\bvar}}{\muabs{\arecvar}{\funap{\blmter}{\arecvar}}} $}
      \DisplayProof
    \end{aligned}
                    $}
  \end{center}
  in the limit a closed derivation $\infDeriv$ in $\Exprmuinf$ is obtained with the same conclusion as $\Deriv$.
  Furthermore, $\infDeriv$ does not contain infinitely many consecutive instances of $\ssmuabs$,
  since the side-condition on ($\sFIX$) guarantees a
  guardedness condition analogous to Proposition\,\ref{prop:guardedness:Reg:stReg:stRegzero}. 
  Hence $\infDeriv$ is a closed admissible derivation in $\Exprmuinf$ with conclusion $\femptylabs{\almter}$.   
  For showing ``$\Leftarrow$'', suppose that $\infDeriv$ is a closed, admissible derivation in $\Exprmuinf$
  with conclusion $\femptylabs{\almter}$. Then there is a finite closed
  derivation $\Deriv$ with the same conclusion in the variant system
  $\Exprmumin$ that does not require the side-condition part
  $\depth{\Derivi{0}}\ge 1$ for instances of ($\sFIX$). 
  Via the process described above, $\Deriv$ unfolds to a closed, \Vacstregeager\ derivation in $\Exprmuinf$, which has to
  be equal to $\infDeriv$, since closed \Vacstregeager\ derivations in $\Exprmuinf$ are unique
  (due to the rules of this system). 
  If $\Deriv$ would not satisfy the guardedness condition described in Proposition$\,$\ref{prop:guardedness:Reg:stReg:stRegzero},
  and therefore would also violate the mentioned side-condition part, for any of its ($\sFIX$)\nb-in\-stan\-ces, 
  then $\infDeriv$ would not be admissible.
  It follows that $\Deriv$ is a closed derivation in $\Exprmu$ with conclusion $\femptylabs{\almter}$.
  
  For (\ref{lem:Exprmu:Exprmuinf:item:ii}) note that by the argument for ``$\Leftarrow$'' in (\ref{lem:Exprmu:Exprmuinf:item:i}),
  every closed derivation in $\Exprmuinf$ is the unfolding of a closed derivation in $\Exprmu$, and that the
  unfolding process can produce only finitely many \lambdamuterms.
  Statement~(\ref{lem:Exprmu:Exprmuinf:item:iii}) follows by an easy analysis of closed derivations in $\Exprmumin$
  that violate the guardedness condition in Proposition$\,$\ref{prop:guardedness:Reg:stReg:stRegzero}
  on any of its ($\sFIX$)\nb-instances.
\end{proof}

\begin{figure}[t!]
\vspace*{-3ex}  
  
\begin{center}  
  \framebox{
\begin{minipage}{380pt}
\begin{center}
  \mbox{}
  \\[1ex]
  \mbox{ 
    \AxiomC{$\phantom{\flabsann{\vec{\avar}\bvar}{\bvar}{\bvar}}$}
    \RightLabel{\bvarax}
    \UnaryInfC{$\flabsann{\vec{\avar}\bvar}{\bvar}{\bvar}$}
    \DisplayProof
        }
  \hspace*{3ex}       
  \mbox{
    \AxiomC{$ \flabsann{\vec{\avar}\bvar}{\almter}{\ater} $}
    \RightLabel{$\labscomp$}
    \UnaryInfC{$ \flabsann{\vec{\avar}}{\labs{\bvar}{\almter}}{\labs{\bvar}{\ater}} $}
    \DisplayProof
        } 
  \hspace*{3ex}
  \mbox{
    \AxiomC{$ \flabsann{\vec{\avar}}{\almteri{0}}{\ateri{0}}$}
    \AxiomC{$ \flabsann{\vec{\avar}}{\almteri{1}}{\ateri{1}}$}
    \RightLabel{$\lappcomp$}
    \BinaryInfC{$ \flabsann{\vec{\avar}}{\lapp{\almteri{0}}{\almteri{1}}}{\lapp{\ateri{0}}{\ateri{1}}} $}
    \DisplayProof 
    }     
  \\[2ex]
  \mbox{}\hspace*{30ex}
  \mbox{
    \AxiomC{$ \flabsann{\vec{\avar}}{\almter}{\aiter} $}
    \RightLabel{\annVacstreg\
                \small (if the binding $\slabs{\bvar}$ is vacuous)
                       }
    \UnaryInfC{$ \flabsann{\vec{\avar}\bvar}{\almter}{\aiter} $}
    \DisplayProof
        }
  \\[-4ex]
  \mbox{
    \AxiomC{$ [\flabsann{\vec{\avar}}{\aconstnamei{\amarker}}{\ater}]^{\amarker} $}
    \noLine
    \UnaryInfC{$\Derivi{0}$}
    \noLine
    \UnaryInfC{$ \flabsann{\vec{\avar}}{\funap{\almter}{\aconstnamei{\amarker}}}{\ater} $}
    \RightLabel{$\sFIXExpr,l$ $\;$
                \parbox{205pt}{\small 
                               (if $\depth{\Derivi{0}} \ge 1$, and $\length{\vec{y}} \ge \length{\vec{x}}$ for all
                                $\flabsann{\vec{y}}{\blmter}{\biter}$ on threads 
                                from open assumptions $(\flabsanntxt{\vec{\avar}}{\aconstnamei{\amarker}}{\ater})^{\amarker}$ down)}
                }                   
    \UnaryInfC{$ \flabsann{\vec{\avar}}{\muabs{\arecvar}{\funap{\almter}{\arecvar}}}{\ater} $}
    \DisplayProof
    }  
  \mbox{}  
\end{center}
\end{minipage}
   }
\end{center} 
  \vspace*{-2ex}  
  \caption{\label{fig:Expr:Exprmu}%
           Natural-deduction style proof system $\Expr$ for expressibility of infinite \lambdaterms\ by \lambdamuterms.
           The proof system $\Exprmu$ for \lambdamuterms\ that express infinite \lambdaterms\ 
        arises by dropping the colons `:' and the subsequent infinite \lambdaterms. 
        Derivations in $\Expr$ and $\Exprmu$ must be \Vacstregeager.}
\end{figure}
\begin{figure}[t!]
\begin{center}
  \framebox{
\begin{minipage}{300pt}
\begin{center}  
  \mbox{}
  \\[0.75ex]
  \mbox{
    \AxiomC{$\flabsann{\vec{\avar}}{\funap{\almter}{\muabs{\arecvar}{\funap{\almter}{\arecvar}}}}{\aiter}$}
    \RightLabel{$\ssmuabs$}
    \UnaryInfC{$\flabsann{\vec{\avar}}{\muabs{\arecvar}{\funap{\almter}{\arecvar}}}{\aiter}$}
    \DisplayProof
        } 
  \\[0.25ex]
  \mbox{}   
\end{center}
\end{minipage}
            }
\end{center}
  \vspace*{-2ex} 
  \caption{\label{fig:Exprinf:Exprmuinf}%
           The proof system $\Exprinf$ for expressibility of \lambdaterms\ by \lambdamuterms\
           arises from $\Expr$ by replacing the rule $\sFIXExpr$ with the rule $\ssmuabs$. 
           The proof system $\Exprmuinf$ for \lambdamuterms\ that express \lambdaterms\ 
           arises from $\Exprinf$ by dropping the colons `:' and the subsequent infinite \lambdaterms.
           Admissible derivations in $\Exprinf$ and in $\Exprmuinf$ do not have infinitely many consecutive
           instances of $\ssmuabs$.%
           }
  \vspace*{-2ex}
\end{figure}

The lemma below links derivability in $\Expr$ with derivability in $\Exprinf$.
Its proof establishes this link via `unfolding' and `folding' of derivations.

\begin{lemma}\label{lem:Expr:Exprinf}
  $\derivablein{\Expr}{\flabsann{\vec{\avar}}{\almter}{\aiter}}$ holds
    if and only if 
  $\derivablein{\Exprinf}{\flabsann{\vec{\avar}}{\almter}{\aiter}}$ holds,
  for all $\flabs{\vec{\avar}}{\almter}\in\Ter{\lambdamuprefixcal}$ and $\flabs{\vec{\avar}}{\aiter}\in\Ter{\inflambdaprefixcal}$.
\end{lemma}

\begin{proof}
  For ``$\Rightarrow$'' let $\Deriv$ be a closed, \Vacstregeager, and finite derivation in $\Expr$ 
  with the conclusion $\,\flabsanntxt{\vec{\avar}}{\almter}{\aiter}\,$.
  By an unfolding process and arguments analogous as described in the proof of Lemma~\ref{lem:Exprmu:Exprmuinf},~(\ref{lem:Exprmu:Exprmuinf:item:i}),
  $\Deriv$ unfolds to a closed, admissible derivation $\infDeriv$ in $\Exprinf$
  with the same conclusion
  $\,\flabsanntxt{\vec{\avar}}{\almter}{\aiter}\,$
  
  For ``$\Leftarrow$'', suppose that 
  $\infDeriv$ is a closed, admissible derivation in $\Exprinf$
  with conclusion $\flabsann{\vec{\avar}}{\almter}{\aiter}$.
  By changing the notation of the formulas used in $\infDeriv$
  to the notation for formulas in $\Unfinf$ as explained in the proof of Lemma~\ref{lem:Exprinf:Unfinf}, 
  a closed, admissible derivation $\infDerivacc$ in $\Unfinf$ with
  conclusion $\,\flabs{\vec{\avar}}{\almter} \unfoldsto \flabs{\vec{\avar}}{\aiter}\,$ is obtained.
  Since subderivations of closed admissible derivations in $\Unfinf$ are again such derivations,
  it follows from the soundness of $\Unfinf$ with respect to $\sunfoldinfnfred$ (cf.\ Proposition\,\ref{prop:Unfinf}),
  and from the uniqueness of the infinite unfolding (if it exists) of a \lambdamuterm\ (cf.\ Proposition\,\ref{prop:unfoldomred})
  that $\infDerivacc$
  does not contain more infinite prefixed \lambdaterms\ than prefixed \lambdamuterms.
  By dropping the symbols $\sunfoldsto$ and the infinite \lambdaterms\ on the right in $\infDerivacc$
  (or by dropping the colons `$:$' and the infinite \lambdaterms\ on the right in $\infDeriv$)
  a closed admissible derivation $\infDerivi{\mu}$ in $\Exprmuinf$ is obtained.
  Due to Lemma~\ref{lem:Exprmu:Exprmuinf}, (\ref{lem:Exprmu:Exprmuinf:item:ii}), 
  $\infDerivi{\mu}$ contains only finitely many prefixed \lambdamuterms.
  Due to the construction of $\infDerivi{\mu}$ this holds for $\infDeriv$ and $\infDerivacc$ as well.
  Since $\infDerivacc$ does not contain more infinite \lambdaterms\ than \lambdamuterms,
  it follows that $\infDerivacc$ contains only finitely many formulas. 
  Due to the correspondence between $\infDerivacc$ and $\infDeriv$,
  this holds also for $\infDeriv$. 
  
  Therefore $\infDeriv$ can be `folded', similar as in the proof of Theorem~\ref{thm:Reg:stReg:stRegzero}, 
  into a finite closed derivation $\Derivacc$ in $\Expr$ with conclusion $\,\flabsann{\vec{\avar}}{\almter}{\aiter}\,$
  by introducing ($\sFIX$)\nb-instances 
  to cut off the derivation above the upper occurrence of a repetition 
  (the side-condition on such instances of ($\sFIX$) is guaranteed due to the admissibility of $\infDeriv$).
\end{proof}

By gathering properties of the systems $\Expr$, $\Exprinf$, and $\Unfinf$ that we have shown
we can now justify the name of the system $\Expr$
on the grounds that this system formalizes the property of (prefixed) \lambdamuterms\ to express (prefixed) infinite \lambdaterms. 

\begin{theorem}\label{thm:Expr}
      The proof system $\Expr$ is sound and complete with respect to 
                                                                     \lambdamuexpressibility.
      That is,
      for all expressions $\flabsanntxt{\vec{\avar}}{\almter}{\aiter}$ of \lambdamuterm\nb-anno\-ta\-ted, prefixed infinite \lambdaterms\
      it holds that 
      $\:\derivablein{\Expr}{\flabsanntxt{\vec{\avar}}{\almter}{\aiter}}\:$
        if and only if  
      $\flabs{\vec{\avar}}{\almter}$ expresses $\flabs{\vec{\avar}}{\aiter}$.   
\end{theorem}

\begin{proof}
  By chaining the equivalences stated by
  Proposition~\ref{prop:Unfinf}, Lemma~\ref{lem:Exprinf:Unfinf}, and Lemma~\ref{lem:Expr:Exprinf}.
\end{proof}

The following lemma establishes the correspondence between
derivability in the proof system $\stRegzero$ from Section~\ref{sec:proofs} 
and
derivability in $\Expr$.
The crucial part of the proof consists in the constructive extraction, 
from a closed derivation $\Deriv$ in $\stRegzero$ with conclusion $\flabs{\vec{\avar}}{\aiter}$,
of a \lambdamuterm~$\flabs{\vec{\avar}}\almter$ that
describes the form of the derivation $\Deriv$,  
and (as the results gathered in this section will show) also unfolds to $\flabs{\vec{\avar}}{\aiter}$.

\begin{lemma}\label{lem:Expr:stRegzero} 
  For all infinite prefixed \lambdaterms~$\flabs{\vec{\avar}}{\aiter}$, 
  $\:\derivablein{\stRegzero}{\flabs{\vec{\avar}}{\aiter}}\;$ holds if and only if
  there exists a prefixed \lambdamuterm~$\flabs{\vec{\avar}}{\almter}$ such that $\:\derivablein{\Expr}{\flabsann{\vec{\avar}}{\almter}{\aiter}}\,$
  holds.  
\end{lemma}

%

\begin{proof}
  For the implication ``$\Rightarrow$''
  it suffices to show that 
  every derivation $\Deriv$ in $\stRegzero$ with conclusion $\flabs{\vec{\bvar}}{\biter}$ and possibly with open assumptions
  can be transformed, 
  by adding appropriate annotating \lambdamuterms\ in the formulas of $\Deriv$,
  into a derivation $\Derivann$ in \Expr\ with conclusion $\flabsanntxt{\vec{\bvar}}{\blmter}{\biter}$
  and corresponding (if any) open assumptions,
  and such that the same variables in $\vec{\bvar} = \tuple{\bvari{1},\ldots,\bvari{n}}$
  occur free in $\blmter$ as in $\biter$.
  This can be established by induction on the depth~$\depth{\Deriv}$ of $\Deriv$.
  In the base case, axioms ($\bvarax$) of $\stRegzero$ are annotated to axioms ($\bvarax$) of $\Expr$,
  and marked assumptions $(\flabs{\vec{\cvar}}{\citer})^{\amarker}$ in $\stRegzero$ 
  to marked assumptions $(\flabsanntxt{\vec{\cvar}}{\aconstnamei{\amarker}}{\citer})^{\amarker}$.
  
  In the induction step it has to be shown that a derivation $\Deriv$ in $\stRegzero$
  with immediate subderivation $\Derivi{0}$ can be annotated appropriately to a derivation $\Derivann$ in $\Expr$,
  making use of the induction hypothesis that guarantees an annotated version $\Derivanni{0}$ of $\Derivi{0}$.
  For obtaining $\Derivann$ from $\Derivanni{0}$ the fact is used that
  the rules in $\Expr$ uniquely determine the \lambdamuterm\ in the conclusion of an instance, once the \lambdamuterm(s) in the premise(s)
  (and in the case of ($\sFIXExpr$) additionally the constants in the assumptions that are discharged) are given. 
  In order to establish that instances of ($\Vacstreg$) in $\Deriv$ give rise to corresponding instances of ($\Vacstreg$) in $\Derivann$,
  the part of the induction hypothesis is used that guarantees that the \lambdamuterm\ annotation in the premise of the rule
  contains precisely the same variable bindings from the abstraction prefix as the \lambdaterm\ it annotates.

  For showing ``$\Leftarrow$'', let $\Deriv$ be a closed derivation in \Expr\ 
  with conclusion $\flabsanntxt{\vec{\avar}}{\allter}{\aiter}$.
  Then a closed derivation $\check{\Deriv}$ in \stRegzero\ with conclusion $\flabs{\vec{\avar}}{\aiter}$
  can be obtained by simply dropping the annotating $\lambdamucal$\nb-terms in formulas of $\Deriv$.
\end{proof}

\begin{example}\label{example:Expr}
  The derivation $\Derivi{l}$ in $\stRegzero$ from Example~\ref{ex:Reg:stReg}, (\ref{ex:Reg:stReg:item:stReg}), on the left can be annotated,
  as described by Lemma~\ref{lem:Expr:stRegzero}
  to obtain the following derivation $\Hat{\Deriv}_{l}$ in $\Expr\,$:
  \begin{center}
    \scalebox{0.92}{
    \AxiomC{$ (\femptylabsann{\aconstnamei{\amarker}}{\aiter})^{\amarker} $}
    \RightLabel{$\Vacstreg$}
    \UnaryInfC{$ \flabsann{\avar}{\aconstnamei{\amarker}}{\ater} $}
    \RightLabel{$\Vacstreg$}
    \UnaryInfC{$ \flabsann{\avar\bvar}{\aconstnamei{\amarker}}{\ater} $}
    \AxiomC{\mbox{}}
    \RightLabel{$\bvarax$}
    \UnaryInfC{$ \flabsann{\avar\bvar}{\bvar}{\bvar} $}
    \RightLabel{$\lappcomp$}
    \BinaryInfC{$ \flabsann{\avar\bvar}{\lapp{\aconstnamei{\amarker}}{\bvar}}{\lapp{\ater}{\bvar}} $}
    \AxiomC{\mbox{}}
    \RightLabel{$\bvarax$}
    \UnaryInfC{$ \flabsann{\avar}{\avar}{\avar} $}
    \RightLabel{$\Vacstreg$}
    \UnaryInfC{$ \flabsann{\avar\bvar}{\avar}{\avar} $}
    \RightLabel{$\lappcomp$}
    \BinaryInfC{$ \flabsann{\avar\bvar}{\lapp{\lapp{\aconstnamei{\amarker}}{\bvar}}{\avar}}{\lapp{\lapp{\ater}{\bvar}}{\avar}} $}
    \RightLabel{$\labscomp$}
    \UnaryInfC{$ \flabsann{\avar}{\labs{\bvar}{\lapp{\lapp{\aconstnamei{\amarker}}{\bvar}}{\avar}}}{\lapp{\lapp{\ater}{\bvar}}{\avar}} $}
    \RightLabel{$\labscomp$}
    \UnaryInfC{$ \femptylabsann{\labs{\avar\bvar}{\lapp{\lapp{\aconstnamei{\amarker}}{\bvar}}{\avar}}}{\labs{\avar\bvar}{\lapp{\lapp{\ater}{\bvar}}{\avar}}} $}
    \RightLabel{$\sFIXExpr, u$}
    \UnaryInfC{$ \femptylabsann{\muabs{\arecvar}{\labs{\avar\bvar}{\lapp{\lapp{\arecvar}{\bvar}}{\avar}}}}{\ater} $}
    \DisplayProof
      }
  \end{center}
  Note that the \lambdamuterm\ in the conclusion
  unfolds to $\aiter$, the infinite \lambdaterm\ in Fig.\,\ref{fig:ll:expressible}. 
%
\end{example}

Now we can prove our main result on \lambdamuexpressibility\ 
by composing the proof-theoretic transformations developed in this section,
and by applying the characterization, from Section~\ref{sec:proofs},
of strong regularity of infinite \lambdaterms\ via derivability in $\stRegzero$.

\begin{theorem}\label{thm:lm-expressible:streg}
  An infinite \lambdaterm\ is $\lambdamu$\nb-expressible if and only if it is strongly regular.
\end{theorem}

\begin{proof}
  For all infinite \lambdaterms~$\aiter$ it holds: 
  \begin{align*}
    \text{$\aiter$ is \lambdamuexpressible}    
      \;\; & \Longleftrightarrow\;\;
      \existsst{\almter\in\Ter{\lambdamucal}}{\;\almter \unfoldinfred \aiter}
      & & \text{(by } \parbox[t]{18ex}{the definition of\\[-0.5ex] \lambdamuexpressibility)} 
    \\[-0.25ex] 
    & \Longleftrightarrow\;\;
      \existsst{\almter\in\Ter{\lambdamucal}}{\;\almter \unfoldominfnfred \aiter}
      & & \parbox[t]{25ex}{(``$\Rightarrow$'' by Proposition$\,$\ref{prop:unfoldomred},\\[-0.5ex] 
                           \phantom{(}``$\Leftarrow$'' due to $\sunfoldominfnfred \subseteq \,\sunfoldinfred)$}
    \displaybreak[0]\\[-0.25ex] 
    & \Longleftrightarrow\;\;
      \existsst{\almter\in\Ter{\lambdamucal}}{\;\derivablein{\Unfinf}{\femptylabs{\almter} \unfoldsto \femptylabs{\aiter}}}
      & & \text{(by Proposition\,\ref{prop:Unfinf})}
    \displaybreak[0]\\[-0.25ex] 
    & \Longleftrightarrow\;\;
      \existsst{\almter\in\Ter{\lambdamucal}}{\;\derivablein{\Exprinf}{\femptylabsann{\almter}{\aiter}}}
      & & \text{(by Lemma\,\ref{lem:Exprinf:Unfinf})}
    \displaybreak[0]\\[-0.25ex]
    & \Longleftrightarrow\;\;  
      \existsst{\almter\in\Ter{\lambdamucal}}{\;\derivablein{\Expr}{\femptylabsann{\almter}{\aiter}}} \hspace*{2ex}
      & & \text{(by Lemma\,\ref{lem:Expr:Exprinf})} 
    \displaybreak[0]\\[-0.45ex]
    & \Longleftrightarrow\;\;  
      {\derivablein{\stRegzero}{\femptylabs{\aiter}}}
      & & \text{(by Lemma\,\ref{lem:Expr:stRegzero})}
    \\[-0.25ex]
    & \Longleftrightarrow\;\; 
    \text{$\aiter$ is strongly regular}
      & & \text{(by Theorem$\,$\ref{thm:Reg:stReg:stRegzero},$\,$(\ref{thm:Reg:stReg:stRegzero:item:stReg:stRegzero})),}
  \end{align*}
  which establishes the statement of the theorem.
\end{proof}

From Theorem~\ref{thm:lm-expressible:streg} 
and Corollary~\ref{cor:thm:streg:fin:bind:capt:chains} we obtain a theorem that condenses our main results. 

\begin{theorem}
  For all infinite \lambdaterms\ $\aiter$ the following statements are equivalent:
  \vspace*{-0.75ex}
  \begin{enumerate}[(i)]
    \item $\aiter$ is \lambdamuexpressible. 
    \item $\aiter$ is strongly regular.
    \item $\aiter$ is regular, and it only contains finite \bindcaptchains.
  \end{enumerate}
\end{theorem}

\section{Generalization to $\lambdaletreccal$ and practical perspectives}
\label{sec:conclusion}

In \cite{grab:roch:2012} we undertook an in-depth study of expressibility in $\lambdaletreccal$,  
and obtained the more general, but analogous result 
for full $\lambdaletreccal$ instead of only for $\lambdamucal$. 
While there are significantly more technicalities involved,  
the structure of the proofs is analogous to here.
Instead of demanding eager application of the scope-delimiting rules
$\srulep{\scompress}$ and $\srulep{\snlvarsucc}$, respectively, there we study
\lambdaterm\ decomposition $\sregred^\astrat$ and $\sstregred^\astrat$
for arbitrary scope-delimiting strategies $\astrat$.

Concepts introduced here and in \cite{grab:roch:2012} have the
potential to be practically relevant for the implementation of functional
programming languages.
In \cite{grab:roch:2013:TERMGRAPH} we study various higher-order and first-order
term-graph representations of cyclic \lambdaterms. Their definitions draw
heavily on the decomposition rewrite systems in this paper. That is, every term in
\lambdaletreccal\ can be translated into a finite first-order `\lambdatg' 
by applying the rewrite strategy $\sstregred$ to the expressed strongly
regular, infinite \lambdaterm. Thereby vertices with the labels
$\sslabs$, $\sslapp$, $\snlvarsucc$ are created according to the kind of
$\sstregred$\nb-step observed (plus variable occurrence vertices with label $\snlvar$).
The degree of sharing exhibited by \lambdatgs\ can be analyzed with functional
bisimulation. In \cite{grab:roch:2013:TERMGRAPH} we identify a class of first-order representations 
with eager application of scope closure
that faithfully preserves and reflects the sharing order on higher-order term
graphs. This leads to an algorithm for efficiently
determining the maximally shared form of a 
                                           term in \lambdaletreccal, 
which
can be put to use in a compiler as part of an optimizing transformation.

\enlargethispage{10ex}

Another aspect is that functional programming languages based on the
\lambdacalculus\ with \stxtletrec\ restrict the set of (in the unfolding semantics) expressible terms 
to the strongly regular infinite \lambdaterms. But members of the superclass of regular terms
are also finitely expressible via sets of equations or \CRS\nb-rules. Therefore the question
arises whether finite representations of regular terms afford new opportunities
in compiling functional programming languages.



\def\sortunder#1{}
\bibliography{expressibility}

\begin{thebibliography}{10}

\bibitem{ario:blom:1997}
Zena~M. Ariola and Stefan Blom.
\newblock {Cyclic Lambda Calculi}.
\newblock In Martin Abadi and Takayasu Ito, editors, {\em Proceedings of
  TACS'97}, volume 1281 of {\em LNCS}, pages 77--106. Springer, 1997.

\bibitem{ario:klop:1997}
Zena~M. Ariola and Jan~Willem Klop.
\newblock {Lambda Calculus with Explicit Recursion}.
\newblock {\em Information and Computation}, 139(2):154--233, 1997.

\bibitem{blom:2001}
Stefan Blom.
\newblock {\em Term Graph Rewriting -- Syntax and Semantics}.
\newblock PhD thesis, Vrije Universiteit Amsterdam, 2001.

\bibitem{bran:heng:1998}
Michael Brandt and Fritz Henglein.
\newblock Coinductive axiomatization of recursive type equality and subtyping.
\newblock {\em Fundamenta Informaticae}, 33:309--338, 1998.

\bibitem{cour:1983}
Bruno Courcelle.
\newblock {Fundamental Properties of Infinite Trees}.
\newblock {\em Theoretical Computer Science}, 25(2):95--169, 1983.

\bibitem{endr:grab:klop:oost:2011}
J\"{o}rg Endrullis, Clemens Grabmayer, Jan~Willem Klop, and Vincent van
  Oostrom.
\newblock {On Equal $\mu$-Terms}.
\newblock In I.\ Bethke, A.\ Ponse, and P.H.\ Rodenburg, editors, {\em
  Festschrift in Honour of Jan Bergstra}, Special Issue of TCS, 412 (28), pages
  3175--3202. Elsevier, June 2011.

\bibitem{grab:2005a}
Clemens Grabmayer.
\newblock {\em {Relating Proof Systems for Recursive Types}}.
\newblock PhD thesis, Vrije Universiteit Amsterdam, March 2005.

\bibitem{grab:roch:2012}
Clemens Grabmayer and Jan Rochel.
\newblock {Expressibility in the Lambda-Calculus with Letrec}.
\newblock Technical Report
  \href{http://arxiv.org/abs/1208.2383}{arXiv:1208.2383}, \url{arxiv.org},
  August 2012.

\bibitem{grab:roch:2013:RTA}
Clemens Grabmayer and Jan Rochel.
\newblock {Expressibility in the Lambda Calculus with $\smu$}.
\newblock In {\em Proceedings of RTA~2013}, 2013.

\bibitem{grab:roch:2013:TERMGRAPH}
Clemens Grabmayer and Jan Rochel.
\newblock {Term Graph Representations for Cyclic Lambda Terms}.
\newblock In {\em Proc.\ of TERMGRAPH 2013}, number 110 in EPTCS, 2013.
\newblock \href{http://arxiv.org/abs/1302.6338}{arXiv:1302.6338}.

\bibitem{kete:simo:2010}
Jeroen Ketema and Jakob~Grue Simonsen.
\newblock {Infinitary Combinatory Reduction Systems: Normalising Reduction
  Strategies}.
\newblock {\em Logical Methods in Computer Science}, 6(1:7):1--35, 2010.

\bibitem{kete:simo:2011}
Jeroen Ketema and Jakob~Grue Simonsen.
\newblock {Infinitary Combinatory Reduction Systems}.
\newblock {\em Information and Computation}, 209(6):893 -- 926, 2011.

\bibitem{mell:96}
Paul-Andr\'e Melli\`es.
\newblock {\em {D}escription {A}bstraite des {S}yst\`emes de {R}\'e\'ecriture
  (Th\`ese de doctorat)}.
\newblock PhD thesis, l'Universit\'{e} Paris 7, December 1996.

\bibitem{oost:97}
Vincent~van Oostrom.
\newblock {FD} \`a la {M}elli\`es, February 1997.
\newblock Vrije Universiteit Amsterdam.

\bibitem{terese:2003}
Terese.
\newblock {\em {Term Rewriting Systems}}, volume~55 of {\em Cambridge Tracts in
  Theoretical Computer Science}.
\newblock Cambridge University Press, 2003.

\end{thebibliography}

\end{document}